\documentclass[a4paper,UKenglish,cleveref, autoref, thm-restate]{lipics-v2021}
\usepackage{macros}
\usepackage{figures}
\title{Tight Bounds for some Classical Problems Parameterized by Cutwidth}

\hideLIPIcs

\nolinenumbers

\author{Narek Bojikian}{Humboldt-Universität zu Berlin, Germany}{bojikian@hu-berlin.de}{https://orcid.org/0000-0003-1072-4873}{}

\author{Vera Chekan}{Humboldt-Universität zu Berlin, Germany}{vera.chekan@informatik.hu-berlin.de}{https://orcid.org/0000-0002-6165-1566}{}

\author{Stefan Kratsch}{Humboldt-Universität zu Berlin, Germany}{kratsch@informatik.hu-berlin.de}{https://orcid.org/0000-0002-0193-7239}{}

\authorrunning{N.\ Bojikian, V.\ Chekan, and S.\ Kratsch} 

\date{} 
\ccsdesc[100]{Theory of computation~Parameterized complexity and exact algorithms}

\keywords{Parameterized complexity, cutwidth, Hamiltonian cycle, triangle packing, max cut, induced matching}

\begin{document}

\listoftodo
\maketitle

\begin{abstract} 
	Cutwidth is a widely studied parameter and it quantifies how well a graph can be decomposed along small edge-cuts.
	It complements pathwidth, which captures decomposition by small vertex separators, and it is well-known that cutwidth upper-bounds pathwidth.
	The SETH-tight parameterized complexity of problems on graphs of bounded pathwidth (and treewidth) has been actively studied over the past decade while for cutwidth the complexity of many classical problems remained open.
	
	For \textsc{Hamiltonian Cycle}, it is known that a $(2+\sqrt{2})^{\operatorname{pw}} n^{\bigoh(1)}$ algorithm is optimal for pathwidth under SETH~[Cygan et al.\ JACM 2022]. 
	Van Geffen et al.~[J.\ Graph Algorithms Appl.\ 2020] and Bojikian et al.~[STACS 2023] asked which running time is optimal for this problem parameterized by cutwidth.
	We answer this question with $(1+\sqrt{2})^{\operatorname{ctw}} n^{\bigoh(1)}$ by providing matching upper and lower bounds.
	Second, as our main technical contribution, we close the gap left by van Heck~[2018] 	for \textsc{Partition Into Triangles} (and \textsc{Triangle Packing}) by improving both upper and lower bound and getting a tight bound of $\sqrt[3]{3}^{\operatorname{ctw}} n^{\bigoh(1)}$, which to our knowledge exhibits the only known tight non-integral basis apart from \textsc{Hamiltonian Cycle}.
    We show that cuts inducing a disjoint union of paths of length three (unions of so-called $Z$-cuts) lie at the core of the complexity of the problem---usually lower-bound constructions use simpler cuts inducing either a matching or a disjoint union of bicliques.
	Finally, we determine the optimal running times for \textsc{Max Cut} ($2^{\operatorname{ctw}} n^{\bigoh(1)}$) and \textsc{Induced Matching} ($3^{\operatorname{ctw}} n^{\bigoh(1)}$) by providing matching lower bounds for the existing algorithms---the latter result also answers an open question for treewidth by Chaudhary and Zehavi~[WG 2023]. 
\end{abstract}

\section{Introduction}

In parameterized complexity the (worst-case) complexity of problems is expressed in terms of input size $n$ and one or more parameters, often denoted $k$. 
The parameter can, for example, be the size of the sought solution 
or some measure of the structure of the input.
The goal is to understand the influence of solution size or structure on the complexity.
A problem is said to be fixed-parameter tractable with parameter $k$ if it admits an algorithm with running time of $f(k) \cdot n^{\bigoh(1)}$ (also denoted by $\bigoh^*(f(k))$) for some computable function $f$. For $\mathsf{NP}$-hard problems, this function $f$ is usually exponential and it may be doubly exponential or worse.

This motivates a line of research devoted to the study of the smallest possible functions $f$ for various problem-parameter combinations.
In this context, one is often interested in $\mathsf{NP}$-hard problems and hence, conjectures stronger than $\mathsf{P} \neq \mathsf{NP}$ are assumed for conditional lower bounds. 
For example, it has been shown that unless the Exponential Time Hypothesis (ETH) fails, some problems do not admit algorithms with running time $\bigoh^*(c^k)$ for any constant $c$ (e.g., \cite{doi:10.1137/16M1104834})---algorithms with such a running time are called \emph{single-exponential}.
For problems admitting single-exponential algorithms, it is natural to search for the smallest value of $c$ for which such an algorithm exists.
An even stronger conjecture called the Strong Exponential Time Hypothesis (SETH) has been assumed to prove for many problems that existing algorithms with some running time $\bigoh^*(c^k)$ are essentially optimal, i.e., for any $\varepsilon > 0$, there is no algorithm for this problem running in time $\bigoh^*((c - \varepsilon)^k)$.
This conjecture states, informally speaking, that the \textsc{SAT} problem cannot be solved much more efficiently than brute-forcing all truth-value assignments.

SETH-tight complexity of problems parameterized by treewidth has been actively studied over the last decade. 
This was initiated by Lokshtanov et al.~\cite{DBLP:journals/talg/LokshtanovMS18} who showed that for many classical graph problems (e.g., \textsc{Independent Set} or \textsc{Max Cut}) the folklore dynamic-programming (DP) algorithms are essentially optimal under SETH.
In parallel, there is also a line of research devoted to accelerating the existing DP algorithms by employing more careful analysis and advanced tools like fast subset convolution (e.g., \cite{DBLP:conf/stoc/BjorklundHKK07,DBLP:conf/csr/Rooij21,DBLP:journals/talg/BjorklundHKKNP16}), Discrete Fourier Transform (e.g., \cite{DBLP:conf/birthday/Rooij20}), rank-based approach (e.g., \cite{DBLP:journals/iandc/BodlaenderCKN15}), isolation lemma (\cite{DBLP:journals/combinatorica/MulmuleyVV87}), and Cut\&Count (e.g., \cite{DBLP:journals/corr/abs-1103-0534/CyganNPPRW11}), we refer to the survey by Nederlof~\cite{DBLP:conf/birthday/Nederlof20} for more details.   
 
Such dynamic-programming algorithms on graphs of bounded treewidth employ the fact that those graphs can be decomposed along small vertex separators and therefore, when processing this decomposition in a bottom-up way, one only needs to remember how a partial solution interacts with the current small separator, also called a bag.
Thus it is also natural to study parameters based on small edge-cuts (as edge-counterparts of vertex separators).
A linear arrangement of a graph places its vertices on a horizontal line so that no two vertices have the same $x$-coordinate.
Now suppose every edge is drawn as an $x$-monotone curve, then the cutwidth of this linear arrangement is the maximum number of edges crossing any vertical line---observe that the edges crossing such a line separate the vertices on the left side of the vertical line from the vertices on the right side so they form an edge-cut. 
The cutwidth of the graph, denoted $\operatorname{ctw}$, is then the minimum over all of its linear arrangements.
It is well-known that pathwidth, denoted $\operatorname{pw}$, can be defined similarly but instead of the edges crossing the vertical line, one counts its end-vertices on, say, the right side of the cut.
In particular, this implies that the cutwidth of a graph upper-bounds its pathwidth. 

Due to this relation, every $\bigoh^*(f(\operatorname{pw}))$ algorithm is also a  $\bigoh^*(f(\operatorname{ctw}))$ algorithm.
However, it is possible that a problem admits a more efficient algorithm when parameterized by cutwidth than when parameterized by pathwidth.
For example, \textsc{Edge Disjoint Paths} is $\mathsf{paraNP}$-hard for treewidth~\cite{NISHIZEKI2001177} but becomes $\mathsf{FPT}$ for cutwidth~\cite{DBLP:journals/algorithmica/GanianO21}.
On a finer level, i.e., in terms of SETH-tight bounds, the complexity parameterized by cutwidth is known for \textsc{Independent Set} and \textsc{Dominating Set}~\cite{DBLP:journals/jgaa/GeffenJKM20}, \textsc{Odd Cycle Transversal}~\cite{DBLP:conf/stacs/BojikianCHK23}, \textsc{Chromatic Number}~\cite{DBLP:journals/tcs/JansenN19}, \#$q$-\textsc{Coloring}~\cite{DBLP:conf/stacs/GroenlandMNS22}, as well as a list of connectivity problems (e.g., \textsc{Feedback Vertex Set} and \textsc{Steiner Tree})~\cite{DBLP:conf/stacs/BojikianCHK23}.
The \textsc{Graph Coloring} problem exposes a particularly interesting behavior: for treewidth, it is known that for any $q \geq 3$, the folklore $\bigoh^*(q^{\operatorname{\tw}})$ algorithm is optimal under SETH~\cite{DBLP:journals/talg/LokshtanovMS18}, but for cutwidth Jansen and Nederlof~\cite{DBLP:journals/tcs/JansenN19} proved that there is a (randomized) algorithm computing the chromatic number of the graph in time $\bigoh^*(2^{\operatorname{\ctw}})$, i.e., independent of the number of colors in question.
Recently, a notable progress was also made for \textsc{List-$H$-Homomorphism}: Groenland et al.~\cite{DBLP:conf/icalp/GroenlandMNPR24} provided a non-algorithmic proof of the existence of so-called representative sets of certain small size on which a dynamic-programming algorithm can rely, and they also provided a lower-bound construction matching the size of these representative sets.
They leave it open, though, whether representative sets of small size can also be found efficiently.

The above-mentioned paper by Lokshtanov et al.~\cite{DBLP:journals/talg/LokshtanovMS18} provided SETH-tight lower bounds for six classical graph problems, namely \textsc{Independent Set}, \textsc{Dominating Set}, \textsc{Partition Into Triangles}, \textsc{Odd Cycle Transversal}, \textsc{$q$-Coloring} (for any fixed $q \geq 3$), and \textsc{Max Cut} when parameterized by treewidth \cite{DBLP:journals/talg/LokshtanovMS18}.
The SETH-tight complexity of these problems parameterized by cutwidth was only partially known till now. 

As the first contribution of our work, we resolve the complexity of the two remaining problems from this paper, namely \textsc{Partition Into Triangles} (and also \textsc{Triangle Packing}) as well as \textsc{Max Cut} when parameterized by cutwidth.
The study of \textsc{Partition into Triangles} parameterized by cutwidth was initiated by van Heck~\cite{thesisTrianglePacking}: it was shown that the problem admits a $\bigoh^*(\sqrt[4]{8}^{\operatorname{ctw}})$ algorithm and no $\bigoh^*((\sqrt{2} - \varepsilon)^{\operatorname{ctw}})$ algorithm exists for any $\varepsilon > 0$ unless SETH fails.
In this work, we improve both lower and upper bound: first, we show that \textsc{Triangle Packing} can be solved in $\bigoh^*(\sqrt[3]{3}^{\operatorname{ctw}})$ while no algorithm solves \textsc{Partition Into Triangles} quicker than this running time unless SETH fails.
There is a trivial reduction from \textsc{Partition Into Triangles} to \textsc{Triangle Packing} and hence, $\bigoh^*(\sqrt[3]{3}^{\operatorname{ctw}})$ is the SETH-optimal running time for both problems.
For \textsc{Max Cut} we provide a lower bound showing that no $\bigoh^*((2-\varepsilon)^{\operatorname{ctw}})$ algorithm can solve this problem for any $\varepsilon > 0$ implying that the folklore algorithm for tree decompositions is optimal for cutwidth as well.

Apart from those two problems, we solve two further open questions.
We show that SETH-tight complexity of \textsc{Hamiltonian Cycle} parameterized by cutwidth is $\bigoh^*((1+\sqrt{2})^{\operatorname{ctw}})$---this was asked by van Geffen et al.~\cite{DBLP:journals/jgaa/GeffenJKM20} and Bojikian et al.~\cite{DBLP:conf/stacs/BojikianCHK23}.
Let us remark, that although for pathwidth it is known that the $\bigoh^*((2+\sqrt{2})^{\operatorname{pw}})$ algorithm is optimal for \textsc{Hamiltonian Cycle}, the complexity relative to treewidth remains a challenging open problem.

Finally, Chaudhary and Zehavi~\cite{DBLP:conf/wg/ChaudharyZ23a} developed a $\bigoh^*(3^{\operatorname{tw}})$ algorithm for \textsc{Induced Matching} problem and a SETH-based lower bound excluding $\bigoh^*((\sqrt{6} - \varepsilon)^{\operatorname{ctw}})$ algorithms for any $\varepsilon > 0$.
They conjectured that their algorithm is optimal and asked for a matching lower bound.
We confirm their conjecture and provide a stronger result, namely that \textsc{Induced Matching} cannot be solved in $\bigoh^*((3 - \varepsilon)^{\operatorname{ctw}})$ for any $\varepsilon > 0$ when parameterized by cutwidth---this resolves the complexity of the problem for both treewidth and cutwidth.  

\begin{table}
\begin{center}
    \begin{tabular}{ |c|c|c|c| }
     \hline
     \textsc{Triangle Packing} (TP) & $2^{\tw}$ & $\sqrt[3]{3}^{\ctw}$\\
     \textsc{Partition into Trianlges} (PT) & $2^{\tw}$ & $\sqrt[3]{3}^{\ctw}$\\
     \textsc{Hamiltonian Cycle} (HC) & $(2+\sqrt{2})^{\pw}$ &$(1+\sqrt{2})^{\ctw}$\\
     \textsc{Max Cut} (MC) & $2^{\tw}$ &$2^{\ctw}$\\
     \textsc{Induced Matching} (IM) & $3^{\tw}$ & $3^{\ctw}$\\
     \hline
    \end{tabular}
    \end{center}
    \caption{Tight bounds for parameterizations by treewidth / pathwidth, and cutwidth. The results of this paper are in the right column. \label{table:results}}
\end{table}

\subparagraph{Technical contribution}
Given the strong relation between linear arrangements and path decompositions, many tight algorithms for problems parameterized by cutwidth start by constructing a path decomposition from a given linear arrangement. 
A crucial property is that, say, the right end-points of the edges of a cut form a vertex separator and therefore, a path decomposition of small width can be obtained by, essentially, creating one bag per cut of the linear arrangement.
Then one employs a dynamic-programming algorithm on this path decomposition and shows that small cutwidth of the original linear arrangement implies an upper bound on the number of states occurring in this algorithm. 
This approach was used, for example, in
\cite{DBLP:conf/stacs/BojikianCHK23, DBLP:conf/icalp/GroenlandMNPR24,DBLP:conf/stacs/GroenlandMNS22}, partly without calling it a path decomposition explicitly though.
This idea turns out to be useful for our algorithm for \textsc{Hamiltonian Cycle} as well.

For \textsc{Triangle Packing}, however, this approach fails to provide a tight bound: it can be shown that this ``straight-forward'' path decomposition may yield too many states, i.e., more than $\bigoh^*(\sqrt[3]{3}^{\operatorname{ctw}})$.
The main bottleneck for this problem are what we call \emph{Z-cuts} 
(i.e., a bipartite graph whose edge set is a path on three edges).
It can be shown that if the left side of this cut is used as a bag, then four states (i.e., four intersections with a triangle packing) are possible---this is too much for the desired $\bigoh^*(\sqrt[3]{3}^{\operatorname{ctw}})$ algorithm.
And if we use the right side of the cut instead, only three states are possible.
At the same time, there exists a family of cuts such that using the right side of the cut as a bag results in too many states as well. 
To overcome the issues of these extremal cases, we define a more involved path decomposition by carefully deciding which edges of the corresponding cut of the linear arrangement to process next, i.e., which vertices on the both sides of a cut to put into the bag. We achieve this by defining a degree constraint on the vertices of the cut to determine which vertices we are allowed to forget. 
Unfortunately, while the eventual DP itself is then straight-forward, the proof for not exceeding the time bound during the transition from one bag to the next becomes quite involved.
Altogether, we show that by carefully choosing not only which vertices of the cut to put into the bag but also in which ordering to do so allows to bound the number of possible states of each bag and obtain the desired tight algorithm.

Our lower-bound constructions for all problems follow the standard scheme for lower bounds for structural parameters as initiated in the work Lokshtanov et al.~\cite{DBLP:journals/talg/LokshtanovMS18}. 
As mentioned above, the $Z$-cuts constitute a bottleneck for our algorithm solving \textsc{Triangle Packing}. 
Our lower-bound construction relies on these cuts and justifies that $Z$-cuts are an obstacle inherent to the problem itself and are not just an artifact of our algorithm.
We show that a $Z$-cut can distinguish three different ``states'' of a solution while contributing three edges to the cutwidth. 
We emphasize that these cuts are novel as the cuts and separators in similar lower bounds 
are usually a single vertex, a single edge, or a biclique. 
The lower-bound construction for \textsc{Max Cut} follows the construction by Lokshtanov et al.~\cite{DBLP:journals/talg/LokshtanovMS18} for pathwidth but employs some adaptations to avoid too many paths between the same end-vertices leading to too large cutwidth.
For \textsc{Induced Matching} 
we define certain gadgets in such a way that any maximum induced matching has some specific structural properties and use these properties in our reduction.
Finally, our matching lower bound for \textsc{Hamiltonian Cycle} relies on the construction by Cygan et al.~\cite{DBLP:journals/jacm/CyganKN18} where we employ a careful splitting argument to obtain small edge-cuts instead of small vertex separators.

\subparagraph{Related Work}  
The (S)ETH tight complexity of problems parameterized by structural parameters has been widely studied.
Apart from the already mentioned papers, there is a long list of papers related to such algorithms on graphs of bounded treewidth (e.g., \cite{DBLP:conf/iwpec/BorradaileL16,DBLP:conf/soda/CurticapeanLN18,DBLP:conf/soda/CurticapeanM16,DBLP:journals/talg/CyganNPPRW22,DBLP:journals/corr/abs-2210-10677/EsmerFMR22,DBLP:conf/soda/FockeMINSSW23,DBLP:conf/soda/FockeMR22,DBLP:journals/dam/KatsikarelisLP22,DBLP:conf/esa/OkrasaPR20,DBLP:journals/siamcomp/OkrasaR21}), treedepth (e.g., \cite{HegerfeldK20,DBLP:journals/dam/KatsikarelisLP19,DBLP:journals/dam/KatsikarelisLP22,NederlofPSW20}), clique-width (e.g., \cite{DBLP:journals/tcs/BergougnouxK19,DBLP:journals/siamdm/BergougnouxK21,DBLP:journals/corr/abs-2307-14264/BojikianK23,DBLP:conf/icalp/GanianHKOS22,DBLP:conf/esa/HegerfeldK23,DBLP:journals/dam/KatsikarelisLP19,DBLP:journals/siamdm/Lampis20}), rank-width (e.g., \cite{DBLP:conf/stacs/BergougnouxKN23,Bui-XuanTV10}), and cutwidth (e.g., \cite{DBLP:conf/stacs/BojikianCHK23,DBLP:conf/icalp/MarxSS21}).
There is also a line of work devoted to conjectures weaker than SETH and yielding the same lower bounds for structural parameterizations (e.g., \cite{canesmer_et_al:LIPIcs.ICALP.2024.34,DBLP:conf/soda/Lampis25,lampis2024circuitsbackdoorsshadesseth}). 

\subparagraph{Organization}  
We start by providing a short summary of the used notation.
\cref{sec:tp} is devoted to \textsc{Triangle Packing}, there we provide our algorithm together with the main steps required to justify its running time.
After that we provide a lower bound matching this running time.
After that in
\cref{sec:hc-ub}
we present the 
algorithm and 
the lower bound 
for \textsc{Hamiltonian Cycle}.
Next in \cref{sec:mim} we constrcut the lower bound for \textsc{Induced Matching} and finally in \cref{sec:maxcut} we provide the lower bound for \textsc{Max Cut}.
We conclude in \cref{sec:conclusion} by providing some open questions.

\section{Preliminaries}
We use $\ostar$ notation to suppress factors polynomial in the input size.
We use $\bigoh_c$ notation to suppress any dependency on $c$. We only use this for constant values $c$ to emphasize that the hidden constant is independent of the input size.
For an integer $i \in \mathbb{N}_0$ by $[i]$ we denote the set $\{1, \dots, i\}$ (in particular, we have $[0] = \emptyset$) and by $[i]_0$ we denote the set $[i] \cup \{0\}$.
For a vector $a= (a_1, \dots, a_n)$ over a ground set $U$ and a mapping $f:U\rightarrow V$ for some set $V$, we denote by $\big(f(v)\big)_{v \in a}$ the vector $(f(a_1), \dots, f(a_n))$. 

Given a graph $G = (V,E)$ and a vertex $v\in V$, the neighborhood of $v$ in $G$ is defined as $N_G(v) = \{w\in V(G)\colon \{v,w\}\in E(G)\}$. We omit the index $G$ when clear from the context. For an edge set $F\subseteq E$, we define $G[F] = (V, F)$ and we define $V(F)$ as the set of end-points of $F$. We define $N_F(v) = N_{G[F]}(v)$ and $\deg_F(v) = |N_F(v)|$. For a vertex set $S\subseteq V$ and a vertex $v\in V$ we define $\deg_S(v) = N_G(v) \cap S$.
We define $\cc(G)$ as the set of all connected components of $G$, where a connected component of $G$ is a maximal connected subgraph of $G$.

We base our lower bound on the Strong Exponential Time Hypothesis (SETH)~\cite{DBLP:journals/jcss/ImpagliazzoP01}:
\begin{conjecture}[SETH]
	For every $\varepsilon > 0$, there exists a constant $d > 0$ such that $d$-SAT cannot be solved in time $\ostar((2-\varepsilon)^n)$, where $n$ is the number of variables.
\end{conjecture}

In order to prove some of our lower bounds, instead of reducing from the $d$-SAT problem, we reduce from the $d$-CSP-$B$ problem for some value $B$ that we choose to suit our reduction:

\begin{definition}
Let $d, B$ be two integers. The $d$-CSP-$B$ problem is defined as follows: Given a set $\Var$ of $n$ \emph{variables}, and a set $\mathcal{C}$ of $m$ \emph{constraints}, where each constraint is a $d$-ary relation over $\Var$, i.e.\ a constraint $C\in \mathcal{C}$ is a $d$-tuple $V_C = (v_1, \dots, v_d) \in \Var^d$ of variables (with $v_i \neq v_j$ for $i \neq j \in [d]$), and a set $R_C\subseteq \big([B-1]_0\big)^{d}$ of tuples of values. An \emph{assignment} $\pi:\Var \rightarrow [B-1]_0$ satisfies a constraint $C$, if $(\pi(v))_{v\in V_C} \in R_C$. The goal is to decide whether there exists an assignment that satisfies all constraints.
\end{definition}

We use the following result by Lampis~\cite{DBLP:journals/siamdm/Lampis20}:

\begin{theorem}[{\cite[Theorem 2]{DBLP:journals/siamdm/Lampis20}}]\label{lb:theo:lampis}
    For every $B\geq 2$ and every $\varepsilon > 0$, unless SETH fails, there exists a positive integer $d$ such that the $d$-CSP-$B$ problem cannot be solved in time $\ostar\big((B-\varepsilon)^{n}\big)$.
\end{theorem}

\subparagraph{Cutwidth}
A linear arrangement $\ell= v_1,\dots, v_n$ of a graph $G=(V,E)$ is an ordering of~$V$. We define $V_i = \{v_1,\dots, v_i\}$ for $i\in[n]$, $V_0=\emptyset$, and $\overline{V}_i = V\setminus V_i$ for $i \in [n]_0$. We define the \emph{cut-graph} at $i \in [n]_0$ as the bipartite graph $H_i = G[V_i, \overline{V}_i]$. 
The set $E_i$ denotes the edge set of $H_i$.
The cutwidth of $\ell$ is defined as $\ctw(\ell) = \max_{i\in [n]} |E(H_i)|$. The cutwidth of $G$ is defined as $\ctw(G) = \min_{\ell} \ctw(\ell)$, where the minimum is taken over all linear arrangements of~$G$. 
We define $L_i$ and $R_i$ as the set of the left and right endpoints of edges in $E_i$, respectively. 
Finally, for $i \in [n]$, we define $Y_i = L_{i-1}\cup \{v_i\}$, i.e., $Y_i$ contains all left end-points of the edges of $E_{i-1}$ together with $v_i$.

\subparagraph{Path decompositions}
A \emph{path decomposition} of a graph $G$ is a pair $(P, \mathcal{B}:V(P)\rightarrow 2^V)$, where $P$ is a simple path and the following properties hold:
\begin{enumerate}
	\item For every vertex $v \in V(G)$, the set $\{x\in V(P)\colon v\in \mB(x)\}$ induces a non-empty connected subgraph of $P$.
	\item For every edge $\{u, v\} \in E(G)$, there exists a node $x \in V(P)$ with $\{u, v\} \subseteq \mB(x)$.
\end{enumerate}
Let $x_1,\dots, x_r$ be the \emph{nodes} of $P$ in the order they occur on $P$. 
The sets $\mB(x_1), \dots, \mB(x_r)$ are called \emph{bags}.
For every $x_i\in V(P)$, we define $B_{x_i} = \mB(x_i)$, $V_{x_i} = \cup_{j \leq i}\mB(x_i)$, and $G_{x_i} = G[V_{x_i}]$.
A path decomposition $(P, \mathcal{B})$ is \emph{nice}, if for any two consecutive nodes $x,x'$ on $P$, it holds that $|\mB(x)\Delta\mB(x')|\leq 1$, and for each endpoint $x$ of $P$ it holds that $\mB(x)=\emptyset$. Hence, one can define a nice path decomposition by providing a sequence of introduce- and forget-vertex operations. It is sometimes useful to have designated \emph{introduce-edge} operations as well, in this case, we call the path decomposition \emph{very nice}. For a node $x$ of a very nice path decomposition of a graph $G$, by $G_x = (V_x, E_x)$ we denote the (not necessarily induced) subgraph of $G$ whose vertex resp.\ edge set consists of the vertices resp.\ edges of $G$ introduced in this decomposition up to the node $x$.

\section{Triangle Packing}\label{sec:tp}
 A \emph{triangle packing} of a graph $G=(V,E)$ is a subgraph $T$ of $G$ such that each connected component of $T$ is a cycle of length three.  
 The \emph{size} of a triangle packing $T$ is the number of connected components of $T$.
 In the \Tpacp problem, given a graph $G$ and a positive integer $\budget$, we are asked whether there exists a triangle packing of size $\budget$ in $G$. 
 In the \Tpartp problem, we are asked whether a triangle packing in $G$ of size $|V|/3$ exists.
\begin{definition}
    A \emph{triangle packing} in a graph $G=(V,E)$ is a subgraph $T$ of $G$ such that each connected component of $T$ is a simple cycle of length three. The size of a triangle packing $T$ is the number of connected components of $T$.
    We define the \Tpacp and the \Tpartp problems as follows:
    
    \probdef{\Tpacp}{A graph $G=(V,E)$, and a positive integer $\budget$.}{Is there a triangle packing in $G$ of size $\budget$?}

    \vspace{.2em}

    \probdef{\Tpartp}{A graph $G=(V,E)$.}{Is there a triangle packing in $G$ of size $|V|/3$?}
\end{definition}

\subsection{Upper Bound}\label{subsec:tripack-algo}

The main result of this subsection can be stated as follows:
\begin{theorem}\label{thm:tripack}
    There exists an algorithm that given an instance $(G, \budget)$ of \Tpacp together with a linear arrangement of $G$ of width at most $\ctw$, runs in time $\ostar(\sqrt[3]{3}^{\ctw})$ and outputs whether $G$ admits a triangle packing of size $\budget$.
\end{theorem}

In the next subsection we will provide a matching lower bound proving the tightness of the algorithm.
Let $(G = (V,E), \budget)$ be an instance of \Tpacp and let $\ell = v_1, \dots, v_n$ be a linear arrangement of $G$ of cutwidth at most $\ctw$.
First, we describe a dynamic-programming algorithm over a nice path decomposition $(P, \mathcal{B})$ of $G$. 
After that, we construct a nice path decomposition of $G$ from the linear arrangement $\ell$ and show that it has certain useful properties, namely, that the number of possible states of our dynamic-programming algorithm is bounded for each bag of this decomposition.

\subsubsection*{Algorithm over a path decomposition}
Let $(P, \mathcal{B})$ be a nice path decomposition of $G$. For every node $x$ of $P$, every set $S \subseteq B_x$, and every integer $b$, we define the family $\mathcal{H}_x^b[S]$ of all triangle packings of $G_x$ of size $b$ whose intersection with $B_x$ is precisely $S$, and we also require that
each triangle in this packing contains at least one vertex of $V_x\setminus B_x$, i.e., at least one forgotten vertex.

\begin{definition}
    Let $x$ be some node of $\mathcal{B}$, $b\in [\lfloor\frac{n}{3}\rfloor]_0$ and $S\subseteq B_x$. We define the family
    \begin{align*}
        \mathcal{H}_x^b[S] = \{H \colon 
        &H \text{ is a triangle packing of $G_x$ of size $b$ s.t.}\\
        &V(H)\cap B_x = S \text{ and each triangle of $H$ contains a vertex in $V_x\setminus B_x$}\}.
    \end{align*}
\end{definition}

We define the set $\mathcal{S}^b_x$ consisting of all subsets $S \subseteq B_x$ such that $\mathcal{H}_x^b[S]$ is non-empty.
We call $\bigcup_{b \in \mathbb{N}} \mathcal{S}_x^b$ the set of \emph{realizable states} at $x$.
A straight-forward algorithm computing all sets $\mathcal{S}^b_x$ can be informally summarized as follows.
Let $x'$ be the node preceding $x$.
At every node $x$ forgetting a vertex, say $v$, we iterate over all triangles containing the vertex $v$ whose other end-points are still in the bag, iterate over all sets in $\mathcal{S}^b_{x'}$, and ``combine'' the two if they are vertex-disjoint.
For every $S \in \mathcal{S}^b_{x'}$ we also add $S \setminus \{v\}$ to $\mathcal{S}^b_x$ as the corresponding triangle packing of $G_{x'}$ remains ``valid'' in $G_x$.
At every introduce-node $x$ we just keep the family $\mathcal{S}^b_{x'}$.
Recall that the root-node, say $r$, of the nice path decomposition $(P, \mathcal{B})$ is an empty bag.
Then the graph $G$ admits a triangle packing of size $b$ if and only if $\emptyset \in \mathcal{S}_r^b$ holds.

We define the tables $T_x^b$ as the characteristic vectors of the sets $\mathcal{S}_x^b$.
After that,
we show that $S\in \mathcal{S}_x^b$ if and only if $\mathcal{H}_x^b[S]$ is not empty for all $S, b$ and $x$. It follows that $G$ contains a triangle packing of size $\budget$ if and only if $\emptyset \in \mathcal{S}_{x_r}^{\budget}$, where $x_r$ is the root of $\mathcal{B}$.

\begin{definition}\label{tripack::def:states}
    We define the tables $T_x^b\in\{0,1\}^{2^{B_x}}$, where for $S\subseteq B_x$ we define $T_x^b[S]$ recursively depending on the type of bag $B_x$. Let $x'$ be the node preceding $x$ in $\mathcal{B}$ if it exists (the child of $x$), and let $v$ be the vertex introduced or forgotten at $x$.
    \begin{itemize}
        \item Leaf bag: we define $T_x^b[\emptyset] = 1$ if $b=0$ and $T_x^b[\emptyset] = 0$ otherwise.
        \item Introduce vertex: We define $T_x^b[S] = T_{x'}^b[S]$ if $v\not\in S$, and $T_x^b[S] = 0$ otherwise.
        \item Forget vertex: we define 
        \[
        T_x^b[S] = T_{x'}^b[S] \lor T_{x'}^b[S \cup \{v\}] \lor \bigvee\limits_{
            \substack{u,w\subseteq S\\
            uvw\text{ is a triangle}}} T_{x'}^{b-1}[S\setminus \{u,w\}].\]
    \end{itemize}
    Finally, we define the set $\mathcal{S}_x^b = \{S \subseteq B_x \colon T_x^b[S] = 1\}$ for every $b \in [\lfloor\frac{n}{3}\rfloor]_0$. 
    We call the set $\mathcal{S}_x = \bigcup_{b\in [\lfloor\frac{n}{3}\rfloor]_0} \mathcal{S}_x^b$ the set of \emph{realizable states} at $x$.
    For $S\in \mathcal{S}_x^b$, we call a subgraph $H \in \mathcal{H}_x^b[S]$ a \emph{certificate} of $S$.
    For a vertex set $V' \subseteq V_x$, let $\mathcal{S}_x^b[V']$ be the restriction of $\mathcal{S}_x^b$ to $B_x\cap V'$, i.e.\
    $\mathcal{S}_x^b[V'] = \{S\cap V'\colon S\in\mathcal{S}_x^b\}$.
\end{definition}

\begin{lemma}
    It holds for all $i\in [r]$, $b\in [\lfloor\frac{n}{3}\rfloor]_0$ and $S\subseteq B_x$, that $T_x^b[S] = \big[|\mathcal{H}_x^b[S]|\geq 1\big]$.
\end{lemma}

\begin{proof}
    We prove this by induction over $i\in[r]$. For the leaf bag $B_1$, it holds that $B_1 = \emptyset$. The family $\mathcal{H}_1^0$ contains the empty packing, and all other families $\mathcal{H}_1^b$ are empty. This corresponds to the definition of $T_1^b$.
    For an introduce vertex bag introducing a vertex $v$, the set $V_x\setminus B_x$ contains no neighbor of $v$---thus no triangle packing of $G_x$ contains the vertex $v$.
    This already implies $\mathcal{H}_x^b[S] = \emptyset$ whenever $v \in S$. 
    The graphs $G_x - v$ and $G_{x'}$ coincide and we also have $B_x = B_{x'} \setminus \{v\}$ as well as $V_x \setminus \{v\} = V_{x'}$.
    Hence, $\mathcal{H}_x^b[S] = \mathcal{H}_{x'}^b[S]$ holds whenever $v \notin S$ holds.

    For a forget vertex node, let $v$ be the forgotten vertex at $x$, and let $H\in\mathcal{H}_x^b[S]$. If $v\notin V(H)$, then it holds that $H\in\mathcal{H}_{x'}^b[S]$. Otherwise, let $A$ be the triangle containing $v$ in $H$, with $V(A)=\{u,v,w\}$. Then it holds that either $v$ is the only vertex of $A$ in $V_x\setminus B_x$, and hence, $V(A)\subseteq B_{x'}$, and $H$ results from the packing $H - A \in \mathcal{H}_{x'}^{b-1}[S\setminus\{u,w\}]$,
    or there exists another vertex of $A$ in $V_x\setminus B_x$, and hence, $H$ belongs to $\mathcal{H}_{x'}^b[S\cup \{v\}]$.
    For the other direction, note that for all $S\subseteq B_{x'}$, both $\mathcal{H}_{x'}^b[S] \subseteq \mathcal{H}_{x}^b[S]$ and $\mathcal{H}_{x'}^b[S \cup \{v\}] \subseteq \mathcal{H}_{x}^b[S]$ as every triangle packing of $G_{x'}$ in which every triangle contains a vertex already forgotten before $x'$ is also a triangle packing of $G_x$ in which every triangle contains a vertex already forgotten before $x$, and the intersection of the triangle packing with the bags $B_x$ and $B_{x'}$ differ (and only differ on $v$) precisely if $v$ is used in this packing.
    Finally, consider vertices $u, w \in S$ such that $u, v, w$ induce a triangle and consider $H \in \mathcal{H}_{x'}^b[S \setminus \{v, u, w\}]$.
    Since $v, u, w \notin S$ (and hence, not in $H$), the graph $H$ together with the triangle $u, v, w$ form a triangle packing of $G_x$, whose intersection with $B_x$ is precisely $S\cup\{u,w\}$. Hence $T_x^b[S\cup\{u,w\}] = 1$ holds.
\end{proof}

\begin{corollary}\label{tripack::cor:algo-correct}
    The graph $G$ admits a triangle packing $H$ of size $\budget$ if and only if $\emptyset \in \mathcal{S}_{x_r}^{\budget}$.
\end{corollary}

\begin{lemma}\label{lem:time-bound-in-states}
    Let $\alpha$ be the maximum size of $\mathcal{S}_x^b$ over all nodes $x$ and $b\in [\lfloor\frac{n}{3}\rfloor]_0$. Then all families $\mathcal{S}_x^b$ can be computed in time $\ostar(\alpha)$.
\end{lemma}

\begin{proof}
    First, we describe a data structure that allows the insertion of a single set in time polynomial in $n$, and allows the iteration over all sets in the data structure in time $\ostar(\alpha)$. There exists multiple such data structures, for example a binary trie, where we define an arbitrary linear ordering over $V$, and for a bag $B_x$ let $v_1,\dots v_k$ be the vertices of $B_x$ in ascending order. 
    Then we assign to a set $S\subseteq B_x$ the string $x_1\dots x_k$ where $x_i = 1$ if $v_i\in S$ and $x_i = 0$ otherwise. This allows an insert operation in time $\ostar(1)$, and we can iterate over all sets in time $\ostar(k)$ traversing the trie in a depth-first manner.

    Assuming such a data structure exists, for the leaf bag $B_1$, we define $\mathcal{S}_1^0 = \{\emptyset\}$, and $\mathcal{S}_1^b = \emptyset$ for all $b>0$. This can be done in time $\ostar(1)$.
    For an introduce bag, we set $\mathcal{S}_x^b = \mathcal{S}_{x'}^b$ for all $b\in [\lfloor\frac{n}{3}\rfloor]_0$. This can be done in time $\ostar(k)$.
    Finally, for a forget bag forgetting a vertex $v$, we iterate over all values of $b$ and initialize $\mathcal{S}_x^b$ to be empty. After that we iterate over all values of $b$ as follows.
    We iterate over each set $S\in \mathcal{S}_{x'}^b$ and do the following:
    First we add $S\setminus \{v\}$ to $\mathcal{S}_x^b$. Then we iterate over all pairs $u,w\in B_{x'}\setminus S$ of adjacent vertices that are neighbors of $v$.
    We add $S\cup\{u,w\}$ to $\mathcal{S}_x^{b+1}$. This can be done in time $\ostar(k)$, since the total number of all pairs of vertices is at most $n^2$.
    Finally, recall that a nice path decomposition has a linear in $n$ number of bags (every vertex is introduced once and forgotten once), this concludes the proof.
\end{proof}

\subsubsection*{From linear arrangement to path decomposition}
Here we describe the bags of the path decomposition $(P, \mathcal{B})$ that we derive from the linear arrangement $\ell$. 
We will start by defining a set of the so-called \emph{checkpoint} bags $X_0, \dots, X_n$ corresponding to the cut graphs $H_0, \dots, H_n$ of the linear arrangement $\ell$.
Later we will show how to add so-called \emph{transition} bags to turn this sequence into a nice path decomposition, while still keeping the number of possible states bounded.

First, for a bipartite graph $H$ (think of $H = H_0, \dots, H_n$), we define the sets $F_H$ and $I_H$ by an iterative process formally defined next.
The set $F_H$ is intended to represent the set of vertices from the left side of the cut that are forgotten already, while the set $I_H$ corresponds to the vertices on the right side of the cut that are introduced already. 
To ensure that the corresponding ``ordering'' of the forget- and introduce-operations even yields a valid path decomposition, i.e., no vertex is forgotten before one of its neighbors is introduced, we will ensure that $I_H$ contains all neighbors of $F_H$.

Clearly, one can safely forget a vertex if all of its neighbors have already been introduced. 
We additionally apply the following non-trivial rule:
We forget a vertex $v$ if (1) it has a single neighbor $w$ in $H$ that is not introduced yet, and (2) the vertex $v$ has the degree at most two in $H$.
In this case we introduce the unique missing neighbor of $v$ before forgetting $v$.
The second constraint plays a crucial role to achieve the desired bound.
We define $Q_H$ as the vertices on the right side that have not been introduced yet, i.e., do not belong to $I_H$.
We provide formal definitions of these sets (see \cref{fig:cut-sets}).

\begin{definition}\label{tripack::def:cut-sets}
    Let $H=(L, R, E)$ be a bipartite graph. 
    First, we define $F^{(0)} = \emptyset$. For $j\geq 0$, let $I^{(j)} = N_H(F^{(j)})$, and  let $Q^{(j)} = R\setminus I^{(j)}$. 
    We define the sets $F^{(j+1)}$ recursively as follows:
    \[
        F^{(j+1)} = \big\{v\in L\colon \deg_{Q^{(j)}}(v) = 0 \lor \big(\deg_{Q^{(j)}}(v) = 1 \land \deg(v) \leq 2\big)\big\}.
    \]
\end{definition}

\begin{figure}[t]
	\centering
	\includegraphics[width=0.9\textwidth]{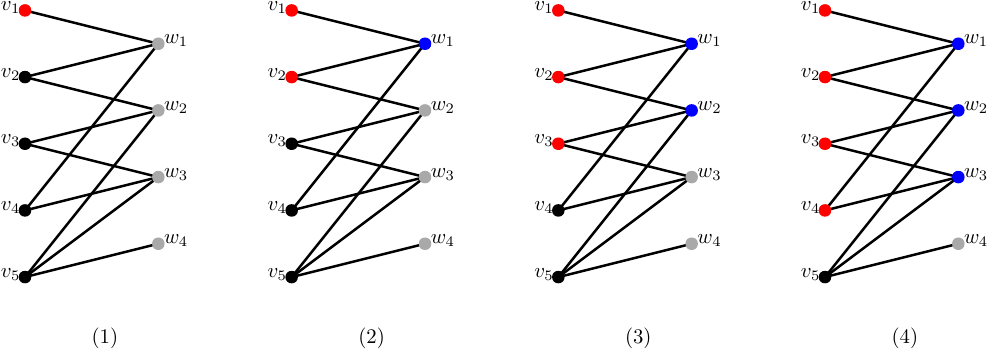}
	\caption{\label{fig:cut-sets}
	Illustration for \cref{tripack::def:cut-sets}.
	For $i=1,2,3,4$, the family $F^{(i)}$ is red on the left-hand side of the cut while the families $I^{(i-1)}$ and $Q^{(i-1)}$ are blue and gray, respectively, on the right-hand side. 
	The sets $F_H, I_H, Q_H$ are the red, blue, and gray vertices, respectively, in $(4)$. 
	The black vertex in $(4)$ belongs to $L^1_H$ as it has a single gray neighbor. The ``bag'' $X_H$ is the set of all black and blue vertices in $(4)$.\\	
	}	
\end{figure}

\begin{lemma}
    It holds that $F^{(j)} \subseteq F^{(j+1)}$. Moreover, there exists some integer number $k$ such that $F^{(k')} = F^{(k')}$ for all $k'\geq k$.
\end{lemma}

\begin{proof}
    We prove the claim by induction over $i$. For $i=0$ it holds that $F^{(0)} = \emptyset \subseteq F^{(1)}$. Now, assume that $F^{(i)} \subseteq F^{(i+1)}$. We prove that $F^{(i+1)} \subseteq F^{(i+2)}$. From the definition of $I^{(i)}$, it holds that $I^{(i)} \subseteq I^{(i+1)}$, and hence, $Q^{(i+1)} \subseteq Q^{(i)}$. It follows that $\deg_{Q^{(i+1)}}(v) \leq \deg_{Q^{(i)}}(v)$. Hence, it follows from the definition of $F^{(i)}$ that if a vertex $v\in L$ belongs to $F^{(i+1)}$, then it belongs to $F^{(i+2)}$ as well. Hence, $F^{(i+1)} \subseteq F^{(i+2)}$.

    Now we prove the second claim.
    For every $i \geq 0$, we have $F^{(i)} \subseteq L$ holds. 
    The set $L$ is finite, and hence, not all inclusions in $F^{(0)} \subseteq F^{(1)} \subseteq \dots \subseteq F^{(|L|+1)}$ can be proper. It follows that $k \leq |L|$. Let $k$ be the smallest value such that $F^{(k)} = F^{(k+1)}$. Then it holds that $Q^{(k)} = Q^{(k+1)}$, and hence, $F^{(k+2)} = F^{(k+1)}$. It holds by induction, that $F^{(k)} = F^{(k')}$ for all $k'\geq k$.
\end{proof}

\begin{definition}
    Let $k$ be the smallest value such that $F^{(k)} = F^{(k+1)}$. 
    We define $F_H = F^{(k)}$, $I_H = I^{(k)}$ and $Q_H = Q^{(k)}$. We also define
    \begin{align*}
        L^2_H &:= \{v\in L\colon |N(v)\cap Q_H| \geq 2\},\\
        L^1_H &:= \{v\in L\colon |N(v)\cap Q_H| = 1\}.
    \end{align*}
    We define $X_H = L^2_H \cup L^1_H \cup I_H$.
    We remove the index $H$ from all sets when clear from the context.
\end{definition}

    For every $i\in[n]_0$, we define $L^2_i = L^2_{H_i}$, $L^1_i = L^1_{H_i}$, $F_i = F_{H_i}$, $I_i = I_{H_i}, Q_i = Q_{H_i}$, and $X_i = X_{H_i}$.
    We call $F_i$ the set of \emph{forgotten vertices} at $i$, $I_i$ the set of \emph{introduced vertices}, and $Q_i$ the set of \emph{unintroduced vertices}.

It is can be shown that $X_i$ is a vertex separator for every $i\in[n]$: this follows from the fact that $I_i$ is the neighborhood of $F_i$. 
Furthermore, we show that 
if in \cref{tripack::def:cut-sets} instead of $F^{(0)} = \emptyset$, we start with the set $F^{(0)} = F_{i-1}$ of the vertices already forgotten at the previous cut, then we obtain the same set $F_i$ of vertices forgotten at the current cut.
From this we can then conclude that a nice path decomposition of $G$ can be constructed by using $X_1, \dots, X_n$ as the main building blocks.
We emphasize that the sequence of vertex separators $X_1, \dots, X_n$ itself does not even neecessarily contain every vertex of $G$. 
To resolve this, we will turn it into a nice path decomposition by adding the so-called \emph{transition bags}.
The bags $X_1, \dots, X_n$ are called the \emph{checkpoint bags}, and $x_1, \dots, x_n$ denote the corresponding nodes in the arising path decomposition.

The reason why we distinguish the sets $L^1_i$ and $L^2_i$ is two-fold.
First, to bound the number of states in terms of cutwidth, we will ``assign'' edges of the cut $H_i$ to certain sets of vertices in the bag.
As every vertex in $L^2_i$ has at least two neighbors in $Q_i$ and no vertex of $Q_i$ belongs to the bag $X_i$ by definition, every vertex of $L^2_i$ will ``certify'' that there are enough edges to allow all possible states in a certain component---we will provide more details later. 
Second, we also need to ensure that along the way between the checkpoint bags, i.e., in transition bags, we do not have too many possible states.
This will be achieved by a careful choice of the ordering in which the vertices are forgotten and introduced.
The sets $L^1_i$ and $L^2_i$ will be used to determine this ordering.

Next we will describe a nice path decomposition $(P, \mathcal{B})$ of $G$ obtained from $\ell$.
This path decomposition will contain the sequence $X_0, \dots, X_n$ of bags as subsequence---we will call these bags \emph{checkpoint} bags as each of them corresponds to a cut in $\ell$ and therefore, they represent the crucial checkpoints of the algorithm.
The remaining bags will be called \emph{transition bags} and they should be thought as intermediate bags that first, ensure that the arising sequence is indeed a path decomposition and second, that it is nice---both properties will be ensured by enriching $X_0, \dots, X_n$ with introduce-vertex- and forget-vertex-bags.
Later we will show that for the constructed path decomposition, the number of realizable states (recall \cref{tripack::def:states}) is small.
This will be proven in two steps: first, we argue this for the checkpoint bags and then show for transition bags as well.
Using \cref{lem:time-bound-in-states}, this bounds shows that the dynamic programming algorithm will be quick enough on this decomposition---we state this more formally later.
Now we construct the desired path decomposition $(P, \mathcal{B})$ by describing which introduce- and forget-bags are inserted between $X_i$ and $X_{i+1}$ for every $i \in [n-1]_0$.

Crucially, our bound only holds when we exhaustively forget all vertices, whose neighbors are all introduced (when ever such vertices arise), before forgetting a vertex with an unintroduced neighbor, where this neighbor is introduced in an additional bag. We formalize this in the following definition.

\begin{definition}\label{tripack::def:transition-bags}
    For $i\in [n]$, let $B_i^1 = X_{i-1}$. We define the sequence of bags $B_i^1,\dots B_i^{r_i}$ between $X_{i-1}$ and $X_i$, as follows. 
    First, if $v_i\notin X_{i-1}$, we add an introduce bag, introducing $v_i$. Let $B_i^{z_0}$ be equal to this bag if $v_i \notin X_{i-1}$ and let $B_i^{z_0} = X_{i-1}$ otherwise. 
    
    We repeat the following process exhaustively for $j\geq1$---by $B_i^{z_{j-1}+1}, \dots, B_i^{z_j}$ we will denote the bags added next.
    First, we define the set $S_i^{j,0}$ as the set of all vertices $v\in L_i\cap B_i^{z_{j-1}}$ such that $N_i(v) \setminus B_i^{z_{j-1}}=\emptyset$ holds.
    And second, we define $S_i^{j, 1}$ as the set of all vertices $v\in L_i\cap B_i^{z_{j-1}}$ such that $\deg_i(v)\leq 2$, and $|N_i(v)\setminus B_i^{z_{j-1}}|=1$ hold. 

    If the set $S_i^j = S_i^{j,0} \cup S_i^{j,1}$ is empty, we are done and no more bags are inserted between $X_{i-1}$ and $X_i$.
    Otherwise, we proceed as follows.
    First, for every vertex $v \in S_i^{j,0}$ we add a bag forgetting $v$.
    After that for each vertex $v \in S_i^{j,1}$, let $\{w\} = N_i(v)\setminus B_i^{z_{j-1}}$. 
    We first add a bag introducing $w$, and follow it with a bag forgetting $v$. By $B_i^{z_j}$ we denote the last bag inserted for the current value of $j$. 
    And we define $r_i = z_j$ for the largest value $j$ such that $S_i^j$ was not empty.

    We define the mapping $\pi_i\colon I_i\rightarrow F_i$, where $\pi_i$ is an extension of $\pi_{i-1} - \{v_i\}$. We define $\pi_i(v)$ for $v\in I_i\setminus I_{i-1}$ as the unique vertex of $F_i \setminus F_{i-1}$ forgotten directly after introducing $v$, i.e.\ $v$ was the only neighbor of $\pi_i(v)$ in $N_i(v)\setminus B_i^{z_{j-1}}$ for some value of $j$. The mapping $\pi_i$ is well-defined, since $\pi_i(v)$ is unique for all $v\in I_i$. Moreover, $\pi_i$ is injective, since each vertex is forgotten exactly once, and introduces thereby at most one vertex. We say the $\pi_i(v)$ is the vertex that \emph{introduces} $v$ in the path decomposition.
\end{definition}

Now we show that $B_i^{r_i} = X_i$ for all $i\in[n]$. We achieve this in two steps. First, we show that the families defined in \cref{tripack::def:cut-sets} are invariant under the order of forgetting vertices. Second, we show that after each iteration of $j$ in the definition above, we forget the sets $(F')^{(j)}$ exactly, where $(F')^{(j)}$ are defined in a similar way to $(F_i)^{(j)}$, with the difference that $(F')^{(0)} = F_{i-1}$, instead of $\emptyset$.
 
\begin{lemma}\label{tripack::lem:cut-sets-order-irrelevant}
    Let $H=(L,R,E)$ be a bipartite graph, let $F, I, L^2, L^1, Q$ be the sets defined by \cref{tripack::def:cut-sets}, and let $X\subseteq F$. We define $F^{(0)}_X = X$. We define the sets $F^{(j)}_X$, $I^{(j)}_X$ and $Q^{(j)}_X$ in a similar way as in \cref{tripack::def:cut-sets}, but depending on $F^{(j)}_X$ and $Q^{(j-1)}_X$ instead. Formally speaking, for $j\geq 0$, let $I^{(j)}_X = N_H(F^{(j)}_X)$, and  let $Q^{(j)}_X = R\setminus I^{(j)}_X$. 
    We define the sets $F^{(j+1)}_X$ recursively as follows:
    \[
        F^{(j+1)}_X = \big\{v\in L\colon \deg_{Q^{(j)}_X}(v) = 0 \lor \big(\deg_{Q^{(j)}_X}(v) = 1 \land \deg(v) \leq 2\big)\big\}.
    \]
    Finally, let $j$ be the smallest index such that $F^{(j)}_X = F^{(j+1)}_X$. We define $F_X = F^{(j)}_X$, then it holds that $F_X = F$.
\end{lemma}

\begin{proof}
    First, we show that $F_X\supseteq F$. It holds, by induction over $j$, that $F^{(j)}\subseteq F^{(j)}_X$, $I^{(j)}\subseteq I^{(j)}_X$ and $Q^{(j)} \supseteq Q^{(j)}_X$;
    this holds by definition, in particular, $Q^{(j)} \supseteq Q^{(j)}_X$ implies that $|N(v)\cap Q^{(j)}|\geq |N(v)\cap Q^{(j)}_X|$ for all $v\in L$, and hence, $v\in F^{(j+1)}$ implies $v\in F^{(j+1)}_X$.

    Now we prove that $F_X\subseteq F$. Suppose this is not the case and let $j$ be the smallest index such that there exists a vertex $v \in F^{(j)}_X \setminus F$. 
    We have $j\neq 0$ as $F^{(0)}_X = X\subseteq F$. 
    Then there exists a vertex $w\in Q^{(j-1)}_X\setminus Q$ (by definition of $F^{(j)}$ and $F^{(j)}_X$). 
    But this means that there exists a vertex $u\in F^{(j-1)}_X\setminus F$ such that $w\in N(u)$, which contradicts the minimality of $j$.
\end{proof}

\begin{corollary}\label{tripack::cor:process-yeilds-cut-set-bag}
    It holds for each $i\in[n]$ that $B_i^{r_i} = X_i$.
\end{corollary}

\begin{proof}
    Let $H = H_i$ and let $S =F_{i-1}$, then by induction over $j$, it is easy to see for that the set $F^{(j)}_S$ defined in \cref{tripack::lem:cut-sets-order-irrelevant} is exactly the same as $V_i\setminus B_i^{z_j}$, since the conditions to forget a vertex in \cref{tripack::def:transition-bags} are exactly the same conditions of inclusion in $F^{(j)}_S$ in \cref{tripack::lem:cut-sets-order-irrelevant}. Hence, $R_i\cap B_i^{z_j} = N_i(L_i\setminus B_i^{z_j}) = I_i^{(j)}$. It follows by \cref{tripack::lem:cut-sets-order-irrelevant} that $B_i^{r_i} = X_{i}$.
\end{proof}

\begin{lemma}
    The sequence of bags
    \[\mathcal{B} = (B_1^1 = X_0 = \emptyset), \dots, (B_1^{r_1} = X_1 = B_2^1), \dots, (B_2^{r_2}= X_2 = B_3^1) \dots (B_{n}^{r_{n}} = X_n = \emptyset)\]
    define a nice path decomposition of $G$.
\end{lemma}

\begin{proof}
    Each vertex $v_i \in V(G)$ is introduced once, since if $v_i\notin X_{i-1}$, then $v_i$ is introduced in the bag $B_i^2$. It is then forgotten in some bag $B_{i''}^j$ with $j\leq n$, since all vertices have degree zero in $H_n$. Now let $i'$ be the smallest index, such that $v_i$ is a neighbor of $v_{i'}$, and let $i''$ be the largest such index. Now let $\{v, v'\}\in E(G)$, and assume that $v$ is forgotten at a bag $B$ before $v'$. It holds by \cref{tripack::def:transition-bags}, that all neighbors of $v$ (including $v'$) are introduced before $B$, but then the bag preceding $B$ contains both $v$ and $v'$.
\end{proof}

\begin{definition}\label{tripack::def:arrangement-to-path-decomposition}
    We define the nice path decomposition $\mathcal{P} = (P, \mathcal{B})$, given by the sequence
    \[\mathcal{B} = (B_1^1 = X_0 = \emptyset), \dots, (B_1^{r_1} = X_1 = B_2^1), \dots, (B_2^{r_2}= X_2 = B_3^1) \dots (B_{n}^{r_{n}} = X_n = \emptyset).\]
    The path $P$ is defined by the nodes $x_0^1,\dots x_n^1$, where $x_i^j$ is the node corresponding to the bag $B_i^j$, i.e.\ $B_i^j = B_{x_i^j}$. We denote $\mathcal{S}_{x_i^j}^b$ by $\mathcal{S}_{i,j}^b$. Moreover, for $i\in[n]$, we define $x_i = x_i^{r_i}$ and $\mathcal{S}_i^b = \mathcal{S}_{i,r_i}^b$.
\end{definition}

\subsubsection*{Bounding the number of realizable states}

\begin{definition}
    Let $i\in[n]$ and $S\subseteq V(H_i)$. We define
    \[
    E_i^S = \{e\in E(H_i)\colon e\cap V(C) \neq \emptyset\}
    \quad\text{and}\quad
    m_i(S) = |E_i^S|,
    \]
    where $E_i^S$ is the set of all edges that intersect $S$, and $m_i(S)$ is their number.

    For $i\in[n]$, we define $\hat{H}_i = H_i[V_i, I_i]$ as the graph resulting from $H_i$ be removing all vertices of $Q_i$, i.e.\ $\hat{H}_i$ is the restriction of $H_i$ to the vertices introduced so far. Also for $S_0\subseteq R_i$, we define $H_i^{S_0} = H_i[V_i, I_i\setminus S_0]$.
\end{definition}

For some $S_0\subseteq R_i$, 
Since each edge of $H_i$ is either an edge of some connected component of $\hat{H}_i$, or an edge between a connected component of $\hat{H}_i$ and $Q_i$, it holds that each edge of $H_i$ intersects exactly one connected component of $\hat{H}_i$. Hence,

\begin{lemma}\label{obs:partition}
    Let $i\in[n]$ and $S_0\subseteq R_i$. Then it holds that the sets $E_i^{V(C)}$ for $C \in \cc(\hat{H}_i)$ form a partition of $E(H_i)$.
\end{lemma}

\begin{proof}
    On the one hand, every edge in some $E_i^C$ belongs to $E(H_i)$, since it has at least one endpoint in $L_i$.
    On the other hand, every edge in $E(H_i)$ either lies completely in one component $C$, or has one endpoint in $Q_i$ and another in $C$. Hence, it cannot belong to two different sets $E_i^C$.
\end{proof}

The following corollary follows from \cref{obs:partition} by considering $S_0 = \emptyset$.

\begin{corollary}\label{cor:partition}
    Let $i\in[n]$. Then it holds that the sets $E_i^{V(C)}$ for $C \in \cc(\hat{H}_i)$ form a partition of $E(H_i)$.
\end{corollary}

Now we aim at showing that for each connected component $C\in\cc(\hat{H}_i)$, the number of possible states of $\mathcal{S}_i[V(C)]$ is upper-bounded by $\sqrt[3]{3}^{m_i(C)}$. The bound on $\mathcal{S}_i$ then follows by the fact that the sets $E_i^{V(C)}$ are pairwise disjoint. This will actually be the part where it becomes evident, as we shall see in the following proof, that the so-called $Z$-cuts form the bottleneck of the algorithm: 
We will distinguish different types of components $C$ of $\hat{H}_i$ and the tight upper bound on the number of possible states is reached by the $Z$-cuts. 

We actually prove a more general statement. First, we prove that the claimed bound holds not only for $C\in\cc(\hat{H}_i)$, but for any connected component $C$ of $H_i^{S_0}$ where $S_0\subseteq R_i$ is arbitrary. 
Here it is important to remember that $S_0$ only contains vertices from the right side of the cut. 
Second, we will show that the bound holds even if we allow to add any subset of $L_i\setminus F_i$ to every possible state.
This motivates the next definition of the set $\mathcal{T}_i^b[S]$ which can be considered as a robust generalization of the set of possible states: intuitively, this permits us to say that even if we allow arbitrary triangles on the left-hand side of the cut, 
the number of states is still bounded.
This will later allow us to prove the bound for the transition bags as well, given that the transitions are carried out in the correct order.
For every $i \in [n]$ and every $b\in [\lfloor\frac{n}{3}\rfloor]_0$, we will use $\mathcal{S}_i^b$ as a shorthand for $\mathcal{S}_{x_i}^b$.

\begin{definition}
    For $b\in [\lfloor\frac{n}{3}\rfloor]_0$, $i\in[n]$, and $S\subseteq V(H_i)$, we define  the families 
    \[
    \mathcal{T}_i^b[S] = \{S'\cup T\colon S'\in\mathcal{S}_i^b[S], T \subseteq S\cap (L^1_i\cup L^2_i)\}.
    \]
    We also define 
    $\mathcal{T}_i[S] =\bigcup\limits_{b\in\left[\left\lfloor\frac{n}{3}\right\rfloor\right]_0} \mathcal{T}_i^b[S]$.
\end{definition}

For a connected component $C$ of some subgraph of $H_i$, we use $\mathcal{T}_i[C]$ and $m_i(C)$ as shorthands for $\mathcal{T}_i[V(C)]$ and $m_i(V(C))$, respectively.
Now we are ready to prove the main technical lemma. 

\begin{lemma}\label{lem:states-bound-after-set-remove}
    For all $i\in[n]$, all $S_0 \subseteq R_i$, and all $C\in \cc(H_i^{S_0})$ it holds that 
    $|\mathcal{T}_i[C]| \leq \sqrt[3]{3}^{m_i(C)}$.
\end{lemma} 

\begin{proof}
    Let $\hat{\mathcal{T}} = \mathcal{T}_i[C]$.
    Let $L^2 = L^2_i \cap V(C)$, $L^1 = L^1_i \cap V(C)$, $F = F_i \cap V(C)$, $I = I_i \cap V(C)$ and $X = X_i \cap V(C)$.
    Observe that we have 
    \begin{equation}\label{eq:subset-of-x}
        \forall S \in \hat{\mathcal{T}} \colon S \subseteq X,
    \end{equation}
    and therefore, also
    \begin{equation}\label{eq:states-subsets-of-x}
        \hat{\mathcal{T}} \leq 2^{|X|}.
    \end{equation}
    We will now carry out a careful analysis showing that sometimes this bound is already good enough to prove the claim and in the remaining cases we show how to improve it.
    Recall that $F, L^1, L^2$ partition $V(C) \cap  L_i$.
    Also, we emphasize that the left side of $H_i^{S_0}$ consists of all vertices in $L_i$.

    First, we show that the inequality $|F|\geq |I|$ holds. We use this inequality multiple times in this proof.
    It holds by definition, that $S_0 \subseteq R_i$, and therefore, $F \cap S_0 = \emptyset$. Hence, for all $w \in I$, it holds that $\pi_i(w)\in F$ (see \cref{tripack::def:transition-bags}).
    The inequality $|F| \geq |I|$ follows by the fact that $\pi_i$ is injective.

    We start by resolving the cases where $m_{i}(C)\leq 2$.
    First, let $m_{i}(C) \leq 1$ hold, we show that $\hat{\mathcal{T}} = \{\emptyset\}$ holds then.
    In this case, it must hold that $C\cap(L^1\cup L^2) = \emptyset$, as any vertex of $L^1$ or $L^2$ adds at least two edges to $E_i^C$, which would imply that $m_{i}(C) \geq 2$.
    If $C$ consists of a single vertex in $F$, then we have $\hat{\mathcal{T}} \subseteq \{\emptyset\}$, since $X=\emptyset$.
    Now we may assume that $C$ consists of a single edge between an $F$ vertex, say $v$ and an $I$ vertex, say $w$.
    Suppose that the vertex $w$ has a neighbor different from $v$ in $H_i$.
    By definition of the graph $H_i^{S_0}$, the vertex $v$ belongs to the same connected component of this graph as $w$ does, i.e., to $C$, and we have $m_{i} \geq 2$ contradicting the assumption on the current case.
    Hence, $w$ is a degree-1 vertex in $H_i$ and no triangle packing of $G[V_i \cup I_i]$ contains the vertex in $I$. It follows again that $\hat{\mathcal{T}} \subseteq \{\emptyset\}$.
    So in both cases, we have
    $|\hat{\mathcal{T}}| = 1 = \sqrt[3]{3}^{0} \leq \sqrt[3]{3}^{m_{i}(C)}$.
    
    Now assume that $m_{i}(C) = 2$ holds. We claim that this implies that $|X|\leq1$. It follows from this claim, by \cref{eq:states-subsets-of-x}, that $|\hat{\mathcal{T}}| \leq 2^1 \leq \sqrt[3]{3}^{2} = \sqrt[3]{3}^{m_{i}(C)}$.
    In order to prove the claim, we distinguish two cases. If $C$ contains two edges, then $C$ is a path on three vertices $u, v, w$. It must hold in this case, that $L^1 \cup L^2 = \emptyset$, since each vertex of $L^1\cup L^2$ has at least one neighbor in $Q$ which contradicts the assumption that $m_i(C)=2$. Moreover, since $|F|\geq |I|$ holds, $C$ can contain at most one vertex of $I$, and hence, $X = L^1\cup L^2\cup I$ has size at most one.
    Otherwise, $C$ contains at most one edge. In particular, $C$ contains exactly one vertex $v$ in $L_i$ and at most one in $I$. Moreover, if $C$ contains a vertex in $I$, then it holds by $|F|\geq |I|$, that $v\in F$. Hence, in both cases it holds that $|X| = |L^1\cup L^2\cup I| = 1$.

    In the remainder of the proof we assume that $m_{i}(C) \geq 3$ holds.
    By connectivity of $C$, it holds that 
    \begin{equation*}
        |E(C)| \geq |V(C)| - 1 = |L^2|+|L^1|+|F|+|I|-1.
    \end{equation*}
    Moreover, the equality is achieved precisely if $C$ is a tree, i.e., it does not contain a cycle.  
    \begin{equation}\label{eq:eq-1}
        |E(C)| \geq 2|I| + |L^2| + |L^1| - 1.
    \end{equation}
    Moreover, the equality is achieved exactly if $C$ is a tree and $|F| = |I|$ holds.
    
    Further, each vertex of $L^2$ has at least two neighbors in $Q_i$, and each vertex of $L^1$ has one neighbor in $Q_i$, thus we have:
    \[
        m_{i}(C) \geq |E(C)| + 2|L^2| + |L^1| + |E(C, S_0)| \stackrel{\eqref{eq:eq-1}}{\geq} 3|L^2| + 2|L^1| + 2|I| - 1  + |E(C, S_0)|. 
    \]
    Finally, recall that $L^1$, $L^2$, and $I$ partition $X$ so we have $|X| = |L^1| + |L^2| + |I|$, hence
    \begin{equation}\label{eq:eq-2}
        m_{i}(C) \geq 2|X| + |L^2| - 1  + |E(C, S_0)|
    \end{equation}
    holds.
    So if at least one of the following is true:
    \begin{itemize}
     \item 
        $C$ contains a cycle, 
     \item $L^2$ is non-empty, 
     \item $|F|$ is strictly greater than $|I|$, 
     \item or if $C$ has a neighbor in $S_0 \cap I_i$ in $H_i$,
    \end{itemize}
    then we have $m_{i}(C)\geq 2|X|$. 
    It follows that
    \[
    |\hat{\mathcal{T}}| \stackrel{\eqref{eq:states-subsets-of-x}}{\leq} 2^{|X|} \leq 2^{\frac{m_{i}(C)}{2}} = \sqrt{2}^{m_{i}(C)}\leq \sqrt[3]{3}^{m_{i}(C)}
    \]
    holds as claimed.

    In the remainder of the proof we assume that all of the following hold:
    \begin{itemize} 
     \item $C$ is a tree, 
     \item $L^2$ is empty, 
     \item $|F| = |I|$, thus each vertex of $F$ introduced some vertex of $I$,
     \item and all neighbors of $V(C) \cap L_i$ in $H_i$ belong to $Q_i$.
    \end{itemize}
    Then the following properties hold.
    First, by definition (recall \cref{tripack::def:cut-sets}) a vertex in $L_i$ of degree 3 is only forgotten if all of its neighbors have been introduced already---then the third item implies that every vertex in $F$ has at most two neighbors in $H_i$.
    Second, by definition of $L^1_i$, every vertex in $L^1$ has precisely one neighbor in $Q_i$ and (since it was not forgotten, recall \cref{tripack::def:cut-sets}) at least two neighbors in $I$.
    
    First, let $L^1 = \emptyset$, then $C$ is a bipartite tree on the vertex set $F\cup I$, with $|F| = |I|$. 
    Hence, it contains a leaf, say $v$, in $I$. 
    Let $w \in F$ be the unique neighbor of $v$ in $H_i$. 
    If $w$ would have the degree of $1$ in $H_i$, the properties above would imply that $m_{i} = 1$ holds contradicting the assumption $m_{i} \geq 3$.
    So the vertex $w$ has degree exactly $2$ in $H_i$. 
    Let $v' \neq v \in I$ be the other neighbor of $w$. 
    Then for each triangle packing $T$ in $G[V_i\cup I_i]$, it holds that that if $T$ contains $v$, then it contains $v'$ as well. 
    Recall that $v, v' \in I \subseteq X$ holds. 
    Hence, for every $S \in \hat{\mathcal{T}}$, if $v\in S$ then we also have $v'\in S$. So it follows that 
    \[
    |\hat{\mathcal{T}}| \stackrel{\eqref{eq:subset-of-x}}{\leq} \frac{3}{4}2^{|X|} \stackrel{\eqref{eq:eq-2}}{\leq} \frac{3}{4}2^{\frac{m_{i}(C)+1}{2}} = \frac{3}{4}\sqrt{2} \sqrt{2}^{m_{i}(C)} \leq \sqrt[3]{3}^{m_{i}(C)},
    \]
    where the last inequality holds due to $m_{i}(C) \geq 3$.

    Finally, we may assume that $L^1\neq \emptyset$ holds. 
    So consider a vertex $v \in L^1$.
    Since $v$ was not forgotten, it has at least two neighbors, say $w_1$ and $w_2$ in $I$, and we have $w_1, w_2 \in C$.
    For the vertices $w_1$ and $w_2$ in $I$, there exist vertices $u_1 \neq u_2 \in F$ due to which $w_1$ and $w_2$ were introduced, and we have $u_1, u_2 \in C$.
    Thus we have $|C| \geq |\{v, w_1, w_2, u_1, u_2\}| = 5$.
    Thus we have $m_{i} \geq 5$ because $C$ contains at least four edges and there is also an edge from $v$ to its unique neighbor in $Q_i$.
    In the remainder of the proof, we will show that $|\hat{\mathcal{T}}| \leq \frac{49}{64} 2^{|X|}$ holds. This would conclude the proof since then it follows that:
    \[
    |\hat{\mathcal{T}}| \leq \frac{49}{64} 2^{|X|} \leq \frac{49}{64} 2^{\frac{m_{i}(C)+1}{2}}\leq \frac{49}{64}\sqrt{2} \sqrt{2}^{m_{i}(C)} \leq \sqrt[3]{3}^{m_{i}(C)},
    \]
    holds where the last inequality holds due to $m_{i}(C) \geq 5$.

    In order to prove the remaining claim we distinguish two cases:
    If $C$ contains a leaf in $I$, then it holds that $|\hat{\mathcal{T}}| \leq \frac{3}{4} 2^{|X|} \leq \frac{49}{64} 2^{|X|}$ by the same argument as in the case $L^1=\emptyset$ where such a leaf always existed. 
    Otherwise, let $\tilde{C}$ be the graph resulting from $C$ by removing $L^1$ from it.     
    Let $\mathcal{I}$ be the family of all connected components of $\tilde C$.
    First, we claim that every connected component in $\mathcal{I}$ has exactly the same number of $F$ and $I$ vertices. 
    On the one hand, the number of $F$-vertices is at least as large as the number of $I$-vertices because every vertex in $I$ belongs to the same connected component in $\mathcal{I}$ as the vertex in $F$ due to which it was introduced.
    On the other hand, we have $|F| = |I|$ and hence, the equality also holds for each connected component in $\mathcal{I}$.
    Therefore, each connected component $X$ of $\mathcal{I}$ contains a leaf that belongs to $I$.
    
    And let $\hat{C}$ be the graph resulting from $C$ by contracting each connected component $X$ of $\tilde C$ into a single vertex $v_X$ (i.e., we contract all edges between $I$ and $F$.) Then $\hat{C}$ is a bipartite tree consisting of $L^1$ on one side and $\{v_X \mid X \in \mathcal{I}\}$ on the other side. 
    Recall that every vertex in $L^1$ has at least to neighbors in $I$: if these two neighbors would belong to the same connected component in $\mathcal{I}$, then $C$ would contain a cycle contradicting the fact that $C$ is a tree.
    Hence, every vertex in $L^1$ also has at least two neighbors after the contraction of every element of $\mathcal{I}$, i.e., it has at least two neighbors in $\hat{C}$.
    Therefore, $T$ contains at least two leaves $v_{X}$ and $v_{X'}$ with $X \neq X' \in \mathcal{I}$. 
    Moreover, recall that both $X$ and $X'$ are trees and each of them contains a leaf, say $y$ and $y'$, respectively, in $I$. 
    Hence, each of the vertices $y$ and $y'$ is of degree two in $C$ and has one neighbor in $F$ and one neighbor in $L^1$. 
    Let $b$ and $b'$, respectively, be the unique neighbor of $y$ and $y'$ in $L^1$, and $z$ and $z'$, respectively, be their $F$-neighbors. If either $z$ or $z'$ has degree one in $H_i$, w.l.o.g.\ let $z$ be this vertex. Then any triangle containing $y$ must contain $b$ as well. Hence, if $y$ is in a solution, then $b$ must be in this solution as well. This implies that
    $\hat{\mathcal{T}} \leq \frac{3}{4} 2^{|X|} \leq \frac{48}{64}2^{|X|}$.
    Otherwise, there exist two vertices $w,w'\in I$, the neighbor of $z$ and $z'$ respectively.
    Then for a triangle packing $T$ of $G[V_i \cup I_i]$, if $T$ contains $y$ then it contains at least one of $b$ and $w$ as well. Similarly, if $T$ contains $y'$ then it must contain at least one of $b'$ and $w'$ as well.
    
    Finally, we distinguish two subcases. If $b \neq b'$, then there exists no $S\in\hat{\mathcal{T}}$ with $y\in S$ but both $w$ and $b$ are not in $S$. Similarly, there exists no $S\in\hat{\mathcal{T}}$ with $y'\in S$ but both $w'$ and $b'$ are not in $S$. Hence, it holds that $|\hat{\mathcal{T}}| \leq \frac{7}{8}\frac{7}{8} 2^{|X|} \leq \frac{49}{64} 2^{|X|}$ as claimed. 
    
    Now let $b=b'$ hold. 
    Recall that we only care about triangle packings where each triangle contains at least one vertex in $F_i$.
    Note that the triangle $b,y,y'$ does not satisfy this property.
    Therefore, for every triangle packing of $G[V_i \cup I_i]$ with this property and its intersection, say $S \in \hat{\mathcal{T}}$ with $C$, if we have $b,y,y'$, then it contains at least one of the vertices $w$ and $w'$. 
    Moreover, there is no set $S \in \hat{\mathcal{T}}$ that excludes $b,w$ but contains $y$, or excludes $b,w'$ but contains $y'$. 
    Hence, these arguments exclude exactly the fraction of $\frac{1}{32}+(\frac{1}{8}+\frac{1}{8}-\frac{1}{32}) = \frac{1}{4}$ of all subsets of $X$, where the subtracted addend corresponds to the case where we include both $y$ and $y'$ and exclude all $b,w,w'$ and this case is counted in each of the two $\frac{1}{8}$. 
    It follows that $|\hat{\mathcal{T}}| \leq \frac{3}{4}2^{|X|} = \frac{48}{64} 2^{|X|} < \frac{49}{64} 2^{|X|}$.
\end{proof}

\begin{corollary}\label{cor:states-bound-at-checkpoint-bags}
    For all $i\in[n]$ and $b\in [\lfloor\frac{n}{3}\rfloor]_0$, it holds  that $|\mathcal{S}_i^b| \leq \sqrt[3]{3}^{|E_i|}$.
\end{corollary}

\begin{proof}
    We apply \cref{lem:states-bound-after-set-remove} with $S_0 = \emptyset$ as follows.
    Observe that we have $\hat{H}_i = H_i^\emptyset$.
    Further, by definition, it holds that $m_{i,\emptyset}(C) = |E_i^C|$ for each component $C\in \cc(\hat{H}_i) = \cc(H_i^{\emptyset})$.
    Recall that $(C \cap X_i)_{C\in \cc(\hat{H}_i)}$ partition the bag $X_i$. 
    It follows that
    \begin{align*} 
    &|\mathcal{S}_i^b| = |\mathcal{S}_i^b[X_i]| = \left|\mathcal{S}_i^b\left[\dot\bigcup_{C \in \cc(\hat{H}_i)} X_i \cap C\right]\right|
    \leq \prod_{C\in\cc(\hat{H}_i)} |\mathcal{S}_i^b[C]|
    = \prod_{C\in\cc(H_i^{\emptyset})} |\mathcal{S}_i^b[C]| 
    \stackrel{\text{\cref{lem:states-bound-after-set-remove}}}{\leq} \\ &\prod_{C\in\cc(H_i^{\emptyset})} \sqrt[3]{3}^{m_{i,\emptyset}(C)}
    = \sqrt[3]{3}^{\sum_{C\in\cc(H_i^{\emptyset})} |E_i^C|} = \sqrt[3]{3}^{\sum_{C\in\cc(\hat{H}_i)} |E_i^C|} = \sqrt[3]{3}^{|E_i|},
    \end{align*}
    where the last inequality holds, since the sets $(E_i^C)_{C \in \cc(\hat{H}_i)}$ partition the set $E(H_i)= E_i$ (recall \cref{obs:partition}).
\end{proof}

We partition the vertices of $B_i^j$ into sets of vertices, corresponding to the connected components of $H_{i-1}^{\{v_i\}}$. 
We achieve this by defining a mapping $\tau$ from the vertex set of $\hat{H}_i$ to the set of connected components of $H_{i-1}^{\{v_i\}}$ together with the ``special'' set $\{v_i\}$ as follows---for simplicity, we also refer to $\{v_i\}$ as a \emph{component}.
First, we map each vertex of $H_{i-1}$ to the component it belongs to. 
After that we proceed as follows:
Our construction of the transition bags ensures that each introduce-node (say, it introduces a a vertex $v$) is directly followed by a forget-node (say, it forgets a vertex $w$). 
Then we map $v$ to the same component to which $w$ is mapped. Intuitively, we assign each vertex to the component that introduces it. 
Since we forget $w$ right after introducing $v$, the number of vertices mapped to the same component remains the same for each forget-node.

More formally, let $\mathcal{C}_i$ denote the family of all connected component of $\cc(H_{i-1}^{\{v_i\}})$ together with the special set $\{v_i\}$. 
We aim at showing that for each $C\in\mathcal{C}_i$, it holds that 
$|\mathcal{S}_{x_i^j}^b[\tau^{-1}[C]]| \leq \sqrt[3]{3}^{m_{i-1}(C)}$.
We will show that the lemma follows from this claim.
We actually prove the claim for forget-nodes only: since the dynamic-programming algorithm only adds triangles to a partial solution when we forget a vertex, the number of possible states does not change at an introduce-node. 

To achieve this, we define a special family of so-called \emph{marked} components $\mathcal{M}_i^j\subseteq \mathcal{C}_i$ for each transition bag $B_i^j$.
The set $\mathcal{M}_i^j$ consists of components of $H_{i-1}^{\{v_i\}}$ that are adjacent (1) to $v_i$ in $H_{i-1}$, or (2) to some vertex in $Q_{i-1}$ that was introduced in a transition bag $B_i^{j'}$ preceding $B_i^j$.
The intuition behind the proof (whose details we omit here) is that the components, that are not marked as $\mathcal{M}_i^j$, are not affected by the changes applied to partial solutions in $B_i^1,\dots, B_i^j$. 
Hence, the possible states of these components are the same as in $X_{i-1}$, and one can apply the bound from \cref{lem:states-bound-after-set-remove} directly.
\begin{definition}
    Let $H_i^0 = H_{i-1}^{\{v_i\}}$, and let $\mathcal{C}_i = \cc(H_i^0) \cup \{\{v_i\}\}$, where we consider each connected component as a vertex set.
    We define a mapping $\tau:V(\hat{H}_{i})\rightarrow \mathcal{C}_i$ as follows.
    First, we assign each vertex of $H_i^0$ to the component it belongs to, and we assign $v_i$ to $\{v_i\}$.
    Then, we assign each vertex $w$ of $I_i\setminus I_{i-1}$ to the same component $\pi_i(w)$ is assigned to, i.e., we set $\tau(w) = \tau(\pi_i(w))$ (see \cref{tripack::def:transition-bags} for reference).
    
    Now we define the set of \emph{marked components} $\mathcal{M}_i^j$ at each bag $B_i^j$ as follows.
    For a vertex $w\in Q_{i-1}$, let $Z_i^j(w)$ be the set of all components in $\mathcal{C}_i$ that contain a vertex $v$ in $L_{i-1}^2$ adjacent to $w$, such that $|N_{H_{i-1}}(v)\setminus B_i^j| = 2$, i.e.\ $v$ has exactly one unintroduced neighbor other that $w$. 
    Now we define $\mathcal{M}_i^1$ as follows depending on the vertex $v_i$.
    \begin{itemize}
	\item If $v_i\in I_{i-1}$, then we define $\mathcal{M}_i^1$ as $\{v_i\}$ together with all components $C \in \mathcal{C}_i \setminus \{\{v_i\}\}$ adjacent to $v_i$ in $H_{i-1}$.
       \item If $v_i\in Q_{i-1}$, then we define $\mathcal{M}_i^1 = \{\{v_i\}\}\cup Z_i^j(v_i)$. 
   \end{itemize}
    Finally, for $j \geq 2$, the set $\mathcal{M}_i^j$ is defined as follows depending on the type of the bag $x_i^j$:
    \begin{itemize}
    	\item If $x_i^j$ is a forget-vertex-$v$-bag, we set $\mathcal{M}_i^j = \mathcal{M}_i^{j-1} \cup \{\tau(v)\}$.
    	\item If $x_i^j$ is an introduce-vertex-$w$-bag for some $w \in Q_{i-1} \cap I_i$, then we set $\mathcal{M}_i^j = \mathcal{M}_i^{j-1} \cup Z_i^j(w)$.
    \end{itemize}
\end{definition}

\begin{lemma}\label{lem:bound-states-unmarked-components}
    For all values $i\in[n]$, $j\in[r_i]$, $b\in [\lfloor\frac{n}{3}\rfloor]_0$ and for each $C \in \mathcal{C}_i \setminus \mathcal{M}_i^j$, it holds that $|\mathcal{S}_{i,j}^b[\tau^{-1}(C)]| \leq \sqrt[3]{3}^{m_{i-1}(C)}$.
\end{lemma}

\begin{proof}
    Let $S_0 = \{v_i\}$. Then it holds that $C\in\cc(H_{i-1}^{S_0})$. 
    It follows by \cref{lem:states-bound-after-set-remove} that $|\mathcal{T}_{i,1}[C]| = |\mathcal{T}_{i-1}[C]| \leq  \sqrt[3]{3}^{m_{i-1}(C)}$.
    First, we claim that $\tau^{-1}(C) \cap V_{x_i^j} = C$, i.e.\ no introduced vertex was assigned to $C$, and no vertex of $C$ was forgotten in the bags $B_i^1, \dots, B_i^j$. 
    This follows by induction over $j$, where one can show that each forget node $x_i^j$ forgets a vertex $v$ where $\tau(v)\in \mathcal{M}_i^j$. 

    For each vertex in $v\in I_{i-1}\cap C$, it holds that $N_{i-1}(v)\subseteq C$, since $C$ is a connected component of $H_i^0$, and $L_{i-1}\subseteq V(H_i^0)$. Hence, for each $S \in \mathcal{S}_{i-1, j}^b[\tau^{-1}(C)]$, each witness of $S$ results from a packing $T$ in $G_{x_i^1}$ by packing triangles disjoint from $I_{i-1}\cap C$, since triangle containing vertices of $I_{i-1}\cap C$ must contain a forgotten vertex of $C$ as well, and hence, $C$ must be marked by the claim above.

    Let $S' = V(T)\cap C$. Then it holds that $S' \in \mathcal{S}_{i,1}^{|V(T)|}[C] \subseteq \mathcal{S}_i[C]$, and $S$ results from $S'$ by adding some vertices in $(L_{i-1}^1\cup L_{i-1}^2)\cap C$. Hence, it holds by definition of $\mathcal{T}_i$ that $S\in \mathcal{T}_i[C]$. It follows by \cref{lem:states-bound-after-set-remove} that 
    $|\mathcal{S}_{i-1,j}^b[\tau^{-1}(C)]| \leq |\mathcal{T}_{i,1}[C]| = |\mathcal{T}_{i-1}[C]| \leq \sqrt[3]{3}^{m_{i-1, S_0}} = \sqrt[3]{3}^{m_{i-1}(C)}$.
\end{proof}
For a marked component $C\in \mathcal{M}_i^j$, we prove that the number of vertices mapped to $C$ in $B_i^j$ is actually even bounded by a smaller value $m_{i-1}(C)/2$. 
This follows from the fact that whenever we introduce a vertex and map it to some component $C$, we forget another vertex $w$ in $\tau^{-1}(C)$. 
Hence, the number of vertices assigned to the same component does not increase in this process. 
Hence, it holds that $|\mathcal{S}_{x_i^j}[C]| \leq2^{\frac{m_{i-1}(C)}{2}} \leq \sqrt[3]{3}^{m_{i-1}(C)}$.

\begin{lemma}\label{lem:bound-vertices-adjacent-to-vi}
    It holds for each forget node $x_i^j$ and $C \in \mathcal{M}_i^j\setminus\{v_0\}$ that $m_{i-1}(C) \geq 2 |B_i^{j} \cap C|$.
\end{lemma}

\begin{proof}
    Let $L^1 = L^1_{i-1}\cap C$, $L^2 = L^2_{i-1} \cap C$, $F = F_{i-1} \cap C$ and $I = I_{i-1} \cap C$. It holds that $E(C) \geq |L^1| + |L^2| + |F| + |I| - 1$ as $C$ is connected by definition. 
    It also holds that 
    $|E_{H_{i-1}}[C, Q_{i-1}\cup \{v_i\}]| \geq 2|L^2| + |L^1|$.
    Moreover, if at least one of the following holds:
    \begin{itemize}
     	 \item $C$ contains a cycle, 
     	 \item $|F| > |I|$, 
    	 \item $L^2$ is non-empty, 
    	 \item if $v_i \in I_{i-1}$, and $C$ is adjacent to $v_i$, 
     \end{itemize}
     then we get: $m_i(C) \geq 2|I| + 2|L_2| + 2|L_1| = 2|B_i^j \cap C|$: Indeed, in each of these cases we get an additional edge (in the first two cases we get an additional edge in $E(C)$, and in the latter two cases we get an additional edge in $E_{H_{i-1}}[C, Q_{i-1}\cup \{v_i\}]$). Hence, it holds that $m_{i-1}(C) \geq 2|B_i^j \cap C|$.
     
    In the remainder we assume that none of the above applies.
    Let $j'$ be the smallest value, such that $C \in M_i^{j'}$. Then $j'\leq j$, and there exists a vertex $v \in L^1$, that was forgotten in $B_i^{j'}$. Hence, it holds that
    \[
    |B_i^j\cap C|\leq |B_i^{j'}\cap C|\leq |B_i^1\cap C| - 1.
    \]
    If follows that
    \[
    m_{i-1}(C) \geq 2|X_{i-1}\cap C| - 1 \geq 2\big(|B_i^j\cap C| + 1\big) - 1 \geq 2 |B_i^j \cap C| + 1.
    \]
\end{proof}

\begin{lemma}\label{lem:transition-component-size}
It holds for all $i\in[n]$, and $j\in [r_i]$, where $x_i^j$ is a forget node, and for all $C\in \mathcal{C}_i$, that $|B_i^j\cap \tau^{-1}(C)|\leq |B_i^j\cap C|$.
\end{lemma}

\begin{proof}
    This holds by induction over $j$, since each bag that introduces a node $v$ with $\tau(v) = C$ is followed directly by a forget bag of a node $w$ with $\tau(v) = C$ as well.
\end{proof}

\begin{corollary}\label{cor:marked-component-size-at-all-bags}
    It holds for each forget node $x_i^j$, and for each connected component $C$ of $\mathcal{M}_i^j \setminus \{v_i\}$, that $|B_i^j \cap \tau^{-1}(C)| \leq \frac{m_{i-1}(C)}{2}$.
\end{corollary}

\begin{proof}
    This follows from \cref{lem:bound-vertices-adjacent-to-vi} and \cref{lem:transition-component-size}.
\end{proof}

\begin{lemma}\label{lem:bound-states-in-all-bags}
    It holds for all $i\in[n]$, $j\in [r_i]$ and $b\in [\lfloor\frac{n}{3}\rfloor]_0$ that $|\mathcal{S}_{i,j}^b| \leq 4\cdot\sqrt[3]{3}^{\ctw}$.
\end{lemma}

\begin{proof}
    First, we claim that $|\mathcal{S}_{i,j}^b[\tau^{-1}(C)]|\leq \sqrt[3]{3}^{m_{i-1}(C)}$ holds for all $C\in \mathcal{C}_i\setminus \{v_i\}$.
    Note that the families $(E_{i-1}^C)_{C\in\cc(\hat{H}_i)}$ form a partition of $E_{i-1}$. Hence, it follows from this claim that
    \[|\mathcal{S}_{i,j}^b| 
    \leq \prod\limits_{C \in \mathcal{C}_i} |\mathcal{S}_{i,j}^b[\tau^{-1}(C)]|
    \leq 2^2\cdot \sqrt[3]{3}^{\sum_{C\in \mathcal{C}_i\setminus \{v_0\}} m_{i-1}(C)} 
    = 4\cdot\sqrt[3]{3}^{\ctw}\]
    holds, where $|\tau^{-1}(\{v_i\})\cap B_i^j| \leq 2$ follows from \cref{lem:transition-component-size}.

    Now it remains to prove the claim.
    If $x_i^j$ is an introduce bag, then by \cref{tripack::def:states}, we have $\mathcal{S}_{i,j}^b = \mathcal{S}_{i,j-1}^b$. 
    Hence, it suffices to prove the claim for each forget bag $x_i^j$. If $C \notin \mathcal{M}_i^j$, then the claim holds by \cref{lem:bound-states-unmarked-components}. 
    So we may assume $C\in \mathcal{M}_i^j$. 
    Then it holds by \cref{cor:marked-component-size-at-all-bags} that $|B_i^j \cap \tau^{-1}(C)| \leq \frac{m_i(C)}{2}$, and hence,
    \[
    \mathcal{S}_{i,j}^b[\tau^{-1}(C)] \leq 2^{|B_i^j\cap \tau^{-1}(C)|} \leq 2^{\frac{m_{i-1}(C)}{2}} = \sqrt{2}^{m_{i-1}(C)} \leq \sqrt[3]{3}^{m_{i-1}(C)}.
    \]
\end{proof}

\begin{proof}[Proof of \cref{thm:tripack}]
    The algorithm first builds the path decomposition $\mathcal{P}$ from \cref{tripack::def:arrangement-to-path-decomposition}---this can be clearly done in polynomial time. 
    After that, the algorithm computes all sets 
    $\mathcal{S}_x^b$ for all nodes $x\in V(P)$ and all values $b\in [\lfloor\frac{n}{3}\rfloor]_0$ in time 
    \[
    	\bigoh^*(\max_{x \in P, b\in [\lfloor\frac{n}{3}\rfloor]_0} |\mathcal{S}_x^b|) \in \bigoh^*(\sqrt[3]{3}^{\ctw})
     \]
     as shown by \cref{lem:time-bound-in-states}, where $|\mathcal{S}^b_x| \leq \sqrt[3]{3}^{\ctw}$ follows from \cref{lem:bound-states-in-all-bags}.
    The algorithm accepts, if $\emptyset \in \mathcal{S}_{x_n}^b$.
\end{proof}

\subsection{Lower Bound}

Now we show the tightness of our algorithm, by proving the following theorem:

\begin{theorem}\label{lb:theo:partition}
Assuming SETH, there exists no algorithm that solves the \Tpartp problem in time $\ostar\big((\sqrt[3]{3}-\varepsilon)^{\ctw}\big)$ for any positive value $\varepsilon$, when a linear arrangement of cutwidth at most $\ctw$ of the graph part of the input.
\end{theorem}

In order to prove \cref{lb:theo:partition}, we provide a reduction from the $d$-CSP-$B$ problem for $B=3$, to the \Tpartp problem. We use \cref{lb:theo:lampis} to prove the correctness of our results.
We call a partition of a graph $G$ into triangles a \emph{solution}.
 
We fix $B=3$ and let $d$ be an arbitrary but fixed positive integer. 
Further, let $I$ be an instance of the $d$-CSP-$B$ problem.
Let $\Var$ be the variables of $I$ and $\mathcal{C}$ the set of constraints. Let $\mathcal{C} = \{C_1, \dots C_m\}$. For the sake of readability, we denote $V_{C_i}$ by $V_i$ and $R_{C_i}$ by $R_i$ for all $i\in[m]$. 
We will now show how to construct, in polynomial time, a graph $G = (V, E)$ and a linear arrangement of $G$ of cutwidth at most $3n + \bigoh(1)$ such that $G$ admits a partition into triangles if and only if $I$ is satisfiable.

\subparagraph*{Overview of the construction.} We will start by describing two main gadgets to be used as the building clocks in the construction of $G$: a so-called \emph{selection gadget} and a so-called \emph{connection gadget}. A \emph{gadget} here refers to a subgraph of constant size. We will combine different gadgets by taking the disjoint union of them, and identifying specified vertices of one gadget with some vertices of the other. The graph $G$ will consist of path-like structures (we call them \emph{paths} for ease of description) consisting of selection gadgets. Between each pair of consecutive selection gadgets of the same path, we will add a connection gadget to combine them. We will achieve this by identifying some vertices of the connection gadget with vertices of the two selection gadgets preceding and following it.
When combining two different gadgets, we only identify vertices of their boundaries.
Finally, we will add a selection gadget for each constraint $C$.

We remark that in comparison to similar SETH-based lower bounds for different problems and structural parameters, our reduction uses a more involved connection gadget. 
In most reductions the connection graph is given implicitly as a single vertex, a single edge (or a matching), or a constant biclique:
In our case, the connection gadget will yield the cuts on three edges having the so-called $Z$-shape (i.e., a path on four vertices) which, as already emphasized by our algorithm in \cref{subsec:tripack-algo}, lies at the core of the problem's complexity.

The intuition behind this construction is the following.
Each path corresponds to a variable in $\Var$. 
A solution of the \Tpartp problem restricted to some selection gadget yields one of $B = 3$ so-called \emph{states}, and these states correspond to assignments of the corresponding variable. The connection gadgets ensure that a solution defines the same state in all selection gadgets of the same path. The selection gadgets corresponding to constraints ensure that the states defined by a solution are consistent with the constraints (i.e., the corresponding assignment satisfy the constraints). 

\subparagraph*{Connection Gadget.} A \emph{connection gadget} consists of three consecutive $Z$-graphs: 
Here a $Z$-graph is a graph on $4$-vertices $u_1,u_2,v_1, v_2$ with the edges $\{u_1,v_1\}$, $\{u_2, v_2\}$, $\{u_1, u_2\}$, $\{v_1,v_2\}$ and $\{u_1,v_2\}$. We identify the vertices $v_1$ and $v_2$ of each of the first two $Z$-graphs with the vertices $u_1$ and $u_2$ of the following $Z$-graph, respectively. 
We refer to \Cref{fig:connection-gadget} for illustration.

For a connection gadget $R$, we denote the vertices $u_1$ and $u_2$ of the first $Z$-graph by $u_1(R)$ and $u_2(R)$, and call them the \emph{entry vertices}.
Similarly, we denote the vertices $v_1, v_2$ of the last $Z$-graph by $v_1(R)$ and $v_2(R)$, and call them the \emph{exit vertices}. We call the entry and the exit vertices together the \emph{boundary} of $R$.
We also denote the two vertices adjacent to $u_1(R)$ and $u_2(R)$ (namely, the vertices $v_1$ and $v_2$ of the first $Z$-graph) by $a_1(R)$ and $a_2(R)$ respectively, and the two vertices adjacent to $v_1(R)$ and $v_2(R)$ (the vertices $u_1$ and $u_2$ of the last $Z$-graph) by $b_1(R)$ and $b_2(R)$.
We omit the symbol $R$ when clear from context.

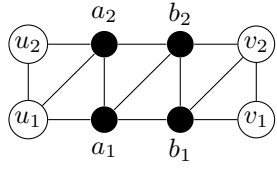
\begin{figure}[h]
    \centering
\begin{tikzpicture}
\pic {connection-gadget};
\end{tikzpicture}
\caption{\label{fig:connection-gadget}A connection gadget $R$.}
\end{figure}

\subparagraph*{Selection Gadget.} Let $U$ be a finite ground set, and $\mathcal{F} = \{F_1, \dots, F_r\}\subseteq \binom{U}{b}$ a family of subsets of size $b$ each for some integer $b$. We define the \emph{selection gadget} $Q = Q(U, \mathcal{F})$ as follows: First, we add a set vertices $\{v_z \colon z \in U\}$, called the boundary of $Q$. For $i\in[r]$, we add $2 \cdot b$ vertices $u^i_1,\dots, u^i_b$ and $w^i_1, \dots, w^i_b$ to $Q$. 
We also define the matchings
\[
    M^i_1 = \big\{\{u^i_j, w^i_j\} \colon j\in[b]\big\},
\]
and 
\begin{equation*}
    M^i_2 = 
    \begin{cases}
        \big\{e^i_j = \{u^i_j, u^{i}_{j+1}\}\colon j \equiv_2 1\big\} \cup \big\{e^i_j = \{w^i_j, w^{i}_{j+1}\}\colon j \equiv_2 0\big\} \cup \{e^i_b = \{u_b,w_1\}\} & b\equiv_2 1,\\
        \big\{e^i_j = \{u^i_j, u^{i}_{j+1}\}\colon j \equiv_2 1\big\} \cup \big\{e^i_j = \{w^i_j, w^{i}_{j+1}\}\colon j \equiv_2 0\big\} \cup \{e^i_b = \{w_b, w_1\}\} & b\equiv_2 0.\\
    \end{cases}
\end{equation*}
We add the edges of $M^i_1$ and $M^i_2$ to $Q$. Note that $M^i_1$ and $M^i_2$ are two disjoint matchings, whose union forms a simple cycle of length $2b$. Finally, for each $i\in[r-1]$, we add a set of $b$ vertices $x^i_1, \dots x^i_b$, and for $j\in[b]$ we add the two edges between $x^i_j$ and both endpoints of $e^i_j$, as well as the two edges between $x^i_j$ and both endpoints of $e^{i+1}_j$. For each $i\in[r]$, let $F_i = \{z^i_1, \dots z^i_b\}$. We add the edges $\{v_{z^i_j}, u^i_j\}$ and $\{v_{z^i_j}, w^i_j\}$ for all $j\in[b]$. An example of a selection gadget is depicted in \Cref{fig:selection-gadget}.

\medskip

We will use two kinds of selection gadgets: The first one is defined by the ground set
\[U = \{x, x', y, y', \szero, \sone, \stwo\}, \text{ and }
\mathcal{F}=\big\{\{x,x',\szero, \sone\},\{x,y',\szero, \stwo\},\{y,y',\sone, \stwo\}\big\}.\]
 We call the vertices $v_{\szero}, v_{\sone}, v_{\stwo}$ the \emph{assignment vertices}, the vertices $v_x, v_{x'}$ the \emph{entry vertices}, and $v_y, v_{y'}$ the \emph{exit vertices}. We call this gadget a \emph{path gadget}, and denote it by $Q$. 
 
 The second kind of selection gadgets (so-called \emph{clause gadgets}) will correspond to the constraints of the input instance. We denote them by $Z = Q(U, \mathcal{F})$, where $U = [d] \times [2]_0$, and the family $\mathcal{F}$ depends on $C$. We will use the clause gadget $Z$ to ensure that a solution defines a state in some selection gadget that corresponds to an assignment that satisfies $C$.
 
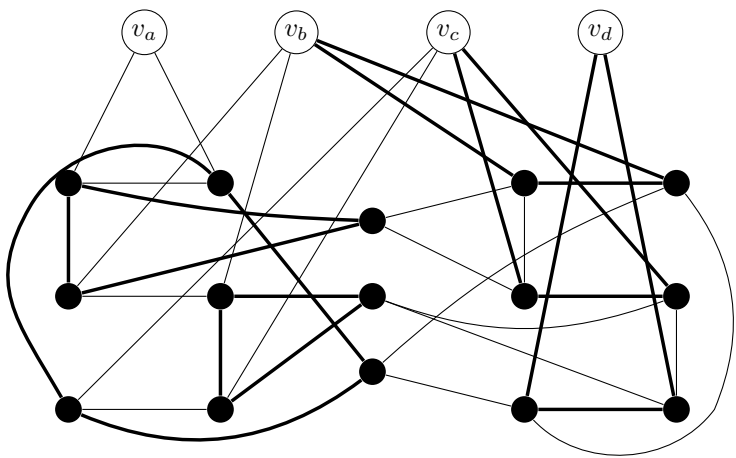
\begin{figure}[h]
    \centering
\begin{tikzpicture}
\pic {selection-gadget};
\end{tikzpicture}
\caption{\label{fig:selection-gadget}Selection Gadget $Q(\mathcal{F})$ for the family $\mathcal{F}=\big\{F_1 = \{a,b,c\}, F_2 = \{b,c,d\}\big\}$. Thick lines depict a solution defining the state $F_2$ in $Q$.}
\end{figure}

For a gadget $X$ that will occur in the graph $G$ and a vertex $v$ of this gadget, we will denote by $v(X)$ the vertex $v$ of the gadget $X$ of $G$.
When we identify two vertices $v(X)$ and $w(Y)$ of two gadgets $X$ and $Y$, both $v(X)$ and $w(Y)$ will be used to denote the same resulting vertex. 

Now we are ready to describe the construction of the graph $G$---we start with the empty graph. 
For each $i \in [n]$, let $I_i = \{j \mid v_i \in \Var(C_j)\}$ be the set of indices of the constraints containing $v_i$, and let $\ell_i = |I_i|$. 
Then we add to $G$ for every $j \in I_i$, 
a path gadget $Q_i^j$ consisting of fresh vertices.
And second we add two further selection gadgets 
$Q_i^0$ and $Q_i^{m+1}$ on fresh vertices each 
defined by $U = \{x,x',y,y'\}$ and $\mathcal{F} = \big\{\{x,x'\},\{y,y'\},\{x,y'\}\big\}$---for the ease of notation we also refer to $Q_i^0$ and $Q_i^{m+1}$ as path gadgets.
Let $Q_i^0, \dots Q_i^{m+1}$ now be the sequence of these gadgets in the order of increasing superscripts. We connect each pair of consecutive gadgets $Q_i^j$ and $Q_i^{j'}$ in this sequence by adding a connection gadget $R_i^j$ on fresh vertices: For this connection, we identify the exit vertices $v_y$ and $v_{y'}$ of $Q_i^j$ with the entry vertices $u_1$ and $u_2$ of $R_i^j$ respectively, and we identify the exit vertices $v_1$ and $v_2$ of $R_i^j$ with the entry vertices $v_x$ and $v_{x'}$ of $Q_i^{j'}$ respectively.

Next for each constraint $j \in [m]$, we add a so-called \emph{clause gadget} $Z_j = Q(U, \mathcal{F}_j)$, where $U = [d]\times [2]_0$, and
\[
    \mathcal{F}_j= \big\{\{(i,x_i)\colon i\in[d]  \} \colon x \in R_j\big\}.
\]
For every $s \in [d]$, let $v_i = (V_j)_s$: We identify $v_{(s, 0)}, v_{(s,1)}, v_{(s,2)}$ with the vertices $v_{\szero}, v_{\sone}, v_{\stwo}$, respectively, of the selection gadget $Q_i^j$. 

Finally, for each $i\in [n]$, we add two cliques $K_i$ and $K'_i$ on five fresh vertices each. We add all edges between $K_i$ and the entry vertices of $Q_i^0$, as well as all edges between $K'_i$ and the exit vertices of $Q_i^{m+1}$. For $i\in[n-1]$, we add all edges between $K_i$ and $K_{i+1}$, and between $K'_i$ and $K'_{i+1}$. We add all edges between $K_n$ and $K'_n$. This concludes the construction of $G$ (see \cref{fig:hc-lb} for an illustration). 

\begin{figure}[t] 
    \centering
    \includegraphics[width=0.8\textwidth]{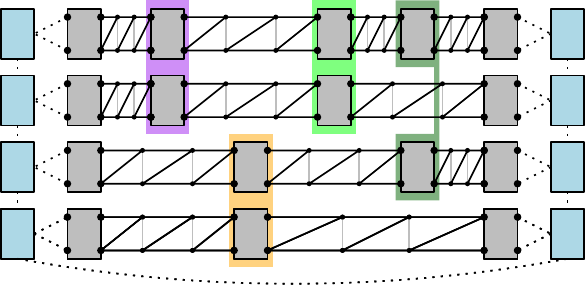}
    \caption{Sketch of the lower-bound construction for $n = 4$ and $m = 4$. Every gray box represents a path gadget, each of them may have one of the three states. Blue boxes are cliques on five vertices. Dotted edges reflect that there exist all possible edges between the sets. Every colored box represents a constraint gadget. The cutwidth of the construction is essentially upper-bounded by $3 n$ as every $Z$-cut contributes three edges to the size of a cut.}
    \label{fig:tp-lb}
\end{figure}

\begin{definition}
    Let $X$ be a gadget in $G$, let $v\in V(X)$, and let $T$ be a partition of $G$ into triangles. 
    We say that a triangle $S$ is \emph{contained} in $X$ if all edges of $S$ belong to $X$.
    We call $v$ \emph{$X$-blocked} in $T$ if the triangle in $T$ containing $v$ is contained in $X$. 
    We call $v$ \emph{$X$-free} otherwise.
    We omit both $X$ and $T$ when clear from context.
\end{definition}

\begin{lemma}\label{tripack-lb::lem:state-selection-gadget}
    Let $T$ be a partition of $G$ into triangles, and let $Q = Q(U, \mathcal{F})$ be a selection gadget of $G$. Let $S_Q$ be the set of $Q$-blocked vertices in $T$, and $I_Q = \{x \colon v_x \in S_Q\}$. Then it holds that $I_Q\in \mathcal{F}$. Moreover, let $G'$ be the graph resulting from $G$ by removing $Q$ but keeping its boundary. If there exists a triangle packing $T'$ of $G'$ keeping exactly a set $\{v_x \colon x \in F\}$ not used in $T'$, for some $F\in \mathcal{F}$, then $T'$ can be extended to a solution in $G$.
\end{lemma}

\begin{proof}
    First, observe that for every $i\in [r]$ and $j \in [b]$, any triangle containing the vertex $u^i_j$ contains precisely one of the two edges incident with $u^i_j$ in $M_1^i \cup M_2^i$: this is because no triangle containing $u^i_j$ completely avoids these two edges and on the other hand, no triangle uses both of these edges.
    Hence, for every $j \in [r-1]$ and every $a \in [2]$, if the triangle in $T$ containing $u^i_j$ uses an edge of $M_a^i$ (and hence, no edge from $M_{3-a}^i$), then the triangle in $T$ containing $u^i_{j+1}$ uses an edge of $M_a^i$ (and hence, no edge from $M_{3-a}^i$). 
    Therefore, for each $i\in [r]$, there exists an index $a(i) \in [2]$ with $(M_1^i \cup M_2^i) \cap E(T) = M_{a(i)}^i$. 
    
    Further, consider $i \in [r-1]$ and $j \in [b]$.
    There exist precisely two triangles in $X$ containing the vertex $x^i_j$: one of them uses the edge $e^i_j$ (and not the edge $e^i_{j+1}$) and the other uses the edge $e^i_{j+1}$ (and not the edge $e^i_j$).
    With above observations and the fact that every edge in contained in at most one triangle, this implies that the sequence $a(1), \dots, a(m)$ contains precisely one value equal to $1$.     
    So let $i^* \in [r]$ be unique such that $a(i^*) = 1$ holds.
    
    On the one hand, for every $j \in [b]$ the only vertex that forms a triangle with the edge $\{u_{i^*}^j, w_{i^*}^j\}$ is the vertex $v_{z_{i^*}}^j$, hence this vertex is $Q$-blocked.
    Now consider a value $u \in U \setminus \{z_{i^*}^1, \dots, z_{i^*}^b\} = U \setminus F_{i^*}$.
    For $v_u$ to be $Q$-blocked, $T$ has to use one of the edges of $\bigcup_{i \in [r]} M_1^i$ in a triangle with $v_u$ but this is impossible:
    First, no edge of $M_1^i$ with $i \neq i^* \in [r]$ is used by $T$ and second, the edges of $M_1^{i^*}$ form triangles with the vertices in $\{v_{z_{i^*}}^1, \dots, v_{z_{i^*}}^b\}$.
    Therefore, $v_u$ is $Q$-free.
    Altogether, we get $S_Q = \{v_{z_{i^*}}^1, \dots, v_{z_{i^*}}^b\}$ and therefore, also $I_Q = F_i \in \mF$ proving the first claim of the lemma.
    
    For the second claim, let $F_i = \{z^i_1\dots z^i_b\}$ be the set of $\mathcal{F}$ such that $\{v_x\colon x\in F_i\}$ are the vertices not used in $T'$. Then we extend $T'$ to a solution $T$ by adding the following triangles: First, we add all triangle consisting of the vertices $v_{z_i^j}, u^i_j, w^i_j$ for all $j\in [b]$ (using hte edges of the matching $M_1^i$). For each constraint $F_{i'}$ with $i' < i$, and for each edge $e^{i'}_j$ of $M_2^{i'}$, we add the triangle between the endpoints of $e^{i'}_j$ and $x^{i'}_j$. Finally, for each $F_{i'}$ with $i'>i$, and for each edge $e^{i'}_j$ of $M_2^{i'}$, we add the triangle between the endpoints of $e^{i'}_j$ and $x^{i'-1}_j$.
\end{proof}

Let $\states=\{\szero, \sone, \stwo\}$ be a set of states. A solution $T$ defines a state of $\states$ in each gadget. We use the following definition and lemmas to prove the correctness of the reduction.

\begin{definition}
    Let $T$ be a partition of $G$ into triangles, and $R$ some connection gadget. We define the state of $R$ defined by $T$ (denoted by $\statef_T(R) \in \states$) as the number of $R$-blocked in $T$ vertices among the entry vertices of $R$, i.e.,
    \[
        \sigma_T(R) = |\{u_i(R) \colon i \in [2], u_i(R) \text{ is $R$-blocked in $T$}\}|.
    \]
    Now let $Q$ be a path gadget. We define the state of $Q$ defined by $T$ (denoted by $\statef_T(Q)\in\states$) as the unique value $i$ such that $v_i$ is $Q$-free in $T$.
\end{definition}

\begin{lemma}\label{tripack-lb::lem:state-connection-gadget}
    Let $T$ be a partition of $G$ into triangles and $R$ some connection gadget. Then it holds that $u_1$ is $R$-blocked if and only if $v_1$ is $R$-free; and $u_2$ is $R$-blocked if and only if $v_2$ is $R$-free.
\end{lemma}

\begin{proof}
    Note that there exist only three possible cases: both $u_1$ and $u_2$ are free, both are blocked, or $u_1$ is blocked but $u_2$ is free---this is because any triangle containing $u_2$ whose edges are fully contained in $R$ also contains $u_1$, i.e., if $u_2$ is blocked, then also $u_1$ is. 
    Recall that $a_1, a_2, b_1, b_2$ are not boundary vertices and hence, for any $y \in \{a_1, a_2, b_1, b_2\}$ the triangle in $T$ containing $y$ needs to be contained in $R$.
    
    In the first case, i.e., both $u_1$ and $u_2$ are free, $T$ contains the triangle $\{a_1, a_2, b_2\}$ (this is the only way to have a triangle containing $a_2$ and none of $u_1$) and $\{b_1, v_1, v_2\}$ (this is the only way to have a triangle containing $b_1$ but none of $a_2$ and $b_2$), and hence, both $v_1$ and $v_2$ are blocked. 
    
    Analogously, if both $u_1$ and $u_2$ are blocked, then $T$ contains the triangles $\{u_1, u_2, a_2\}$ and $\{a_1, b_1, b_2\}$. Hence, both $v_1$ and $v_2$ are free. 
    
    And finally, if $u_1$ is blocked but $u_2$ is free, then $T$ contains the triangles $\{u_1, a_1, a_2\}$ and $\{b_1, b_2, v_2\}$. Hence, $v_1$ is free but $v_2$ is blocked.
\end{proof}

\begin{lemma}\label{tripack-lb::lem:same-state-on-path}
    Let $T$ be a partition of $G$ intro triangles and let $i \in [n]$. Then $T$ defines the same state on all path gadgets and all connection gadgets of the $i$.th path, i.e., for any $X, X'$ where each of $X$ and $X'$ is either a path or a connection gadget on the $i$th path we have $\sigma_T(X) = \sigma_T(X')$.
\end{lemma}

\begin{proof}
    We prove the claim by induction over the gadgets of the $i$th path.
    Let $Q$ be a path gadget followed by a connection gadget $R$ and let $s = \sigma_T(Q)$ be the state of $Q$ defined by $T$. 
    Then 
    it holds by the definition of the family $\mathcal{F}$ of a path gadget, and \cref{tripack-lb::lem:state-selection-gadget}, that the number of $Q$-free vertices among $v_y(Q), v_{y'}(Q)$ is equal to $s$. 
    By construction we have $v_y(Q) = u_1(R)$ and $v_{y'}(Q) = u_2(R)$.
    Furthermore, the every triangle containing one of these vertices is contained in precisely one of $R$ and $Q$.
    Hence, the vertex $v_y(Q) = u_1(R)$ (resp.\ $v_{y'}(Q) = u_2(R)$) is $Q$-free if and only if it is $R$-blocked.
    Hence, the number of $R$-blocked vertices among $u_1(R), u_2(R)$ in $R$ is also $s$. 
    Therefore, the state of $R$ defined by $T$ is equal to $s$.

    On the other hand, let $Q'$ be the path gadget following the connection gadget $R$. 
    By \cref{tripack-lb::lem:state-connection-gadget}, the number of $R$-free vertices among $v_1(R), v_2(R)$ equals $s$.
    As we have $v_1(R) = v_{x}(Q')$ and $v_2(R) = v_{x'}(Q')$ and any triangle using at least one of these vertices is fully contained in precisely one of $R$ and $Q$, the number of $Q'$-blocked vertices among $v_x(Q'), v_{x'}(Q')$. 
    Hence, by the definition of the family $\mathcal{F}$ of a path gadget and \cref{tripack-lb::lem:state-selection-gadget}, the state of $Q'$ defined by $T$ is $s$ as well.
\end{proof}

\begin{lemma}
    If $I$ is satisfiable, then there exists a solution $T$ of $G$.
\end{lemma}

\begin{proof}
    Let $\tau:\Var \rightarrow [2]_0$ be a satisfying assignment of $I$. For each $i\in[n]$, let $\tau_i = \tau(v_i)$. We choose the state $\tau_i$ for all path gadgets $Q_i^j$. 
    As implied by \cref{tripack-lb::lem:same-state-on-path}, and by the definition of the state of a path gadget,
    it holds that the only vertices not packed by this on the $i$th path are the vertices $v_{\tau_i}$ of each path gadget $Q_i^j$ for $j\notin\{0,m+1\}$ as well as some entry vertices of $Q_i^0$ and some exit vertices of $Q_i^{m+1}$.
    Hence, for each constraint $j \in [m]$ with $V_j = (v_{i_1}, \dots, v_{i_d})$,
    the vertices $v_{(z, \tau_{i_z})}$ for $z\in[d]$ are the only vertices on its boundary not packed by the solution, and it holds that $\{(z, \tau_{i_z})\colon z\in[d]\} \in \mathcal{F}$ since $\tau$ is a satisfying assignment.
    By \cref{tripack-lb::lem:state-selection-gadget}, we can choose the solution $T$ in $Z_j$ to pack exactly these vertices among its boundary.

    The only vertices left are all vertices of the cliques $K_i$ and $K'_i$, as well as entry vertices of $Q_i^0$, and exit vertices of $Q_i^{m+1}$ that are not packed yet. 
    For each $i\in[n]$ in ascending order, we pack the vertices $u_1, u_2$ of $Q_i^0$ not packed yet together with free vertices of $K_i$, and the vertices of $K_i$ with each other. 
    If some vertices of $K_i$ are left, we pack them with vertices of $K_{i+1}$. Finally, it holds by \cref{tripack-lb::lem:state-connection-gadget} and \cref{tripack-lb::lem:same-state-on-path} that among $u_1(Q_i^{0}), u_2(Q_i^{0}), v_1(Q_i^{m+1}), v_2(Q_i^{m+1})$ there are exactly two free
    vertices, which together with the vertices of $K'_i$ and $K'_{i+1}$ add up to $12\equiv_3 0$ vertices. Hence, on can show inductively that the number of vertices of $K_i$ and $K'_i$ that are not packed yet
    (i.e.\ must be packed with vertices of $K_{i+1}$ and $K'_{i+1}$) is also a multiple of three. It follows that the number of unpacked vertices of $K_n$ and $K'_n$ is also a multiple of three. Hence, we can pack the remaining vertices of $K_n$ and $K'_n$ with each other. This concludes the proof.
\end{proof}

\begin{lemma}
    If there exists a solution $T$ of $G$, then $I$ is satisfiable.
\end{lemma}

\begin{proof}
    It holds by \cref{tripack-lb::lem:same-state-on-path} that all path gadgets on the $i$th path have the same state defined by $T$. Let $a_i$ by the state defined on the $i$ path. We define the assignment $\tau: \Var \to [2]_0$ so that $\tau(v_i) = a_i$ for all $i\in [n]$. 
    We now show that $\tau$ satisfies all constraints in $\mathcal{C}$.
    So let $j \in [m]$ be arbitrary.
    For every $i \in [n]$ with $v_i \in \Var(C_j)$, 
    it holds that $v_{a_i}(Q_i^j)$ is $Q_i^j$-free by definition of the state of $Q_i^j$. 
    Since every triangle containing $v_{a_i}(Q_i^j)$ is fully contained in one of $Q_i^j$ and $Z_j$, the vertex $v_{(z, a_i)}$ must be $Z_j$-blocked, where $v_i$ is the $z$th variable of the constraint $C_j$.
    For $V_j = (v_{i_1},\dots v_{i_d})$, it follows by \cref{tripack-lb::lem:state-selection-gadget} that $\{(i, a_{i_z})\colon z\in[d]\} \in \mathcal{F}_j$. Hence, by the definition of $Z_j$, it holds that $(a_{i_1}, \dots a_{i_d})\in R_j$. Therefore, $\tau$ is a satisfying assignment of $I$.
\end{proof}

\begin{lemma}
    There exists a constant $k_0$ such that the graph $G$ has the cutwidth of at most $3n + k_0$.
    Furthermore, a linear arrangement of $G$ of such cutwidth can be computed in polynomial time from $I$.
\end{lemma}

\begin{proof}
    We now describe the desired linear arrangement with the given width. 
    In the arrangement, we will add one gadget at a time. Hence, the vertices of each gadget, except for its boundary, will form a consecutive segment of the arrangement. For two gadgets with a common boundary, the vertices of the boundary will appear between the corresponding segments of these two gadgets.
    
    For $i = 1, \dots, n$, we insert the vertices of $K_i$ followed by the vertices $Q_i^0$. 
    Then for each $j\in [m]$ in ascending order, we do the following: we add $Q_i^j$ followed by $R_i^j$ (excluding its boundary) if $Q_i^j$ exists. 
    We follow these by the vertices of $Z_j$ not placed yet. 
    In the very end, for $i = 1, \dots, n$, we add $Q_i^{m+1}$ followed by $K'_i$.
    
    We claim that this arrangement has cutwidth upper-bounded by $3n + k_0$, where $k_0$ is a constant. To see this, note that any cut of the arrangement can be crossed by the edges of at most on path gadget (together with the edges between some clique and this path gadget),
    the edges of one clause gadget, the edges of one connection gadget, the edges of one clique $K_i$, edges between $K_n$ and $K'_n$, and all edges between $b_1, b_2$ and $u_1, u_2$ of the connection gadgets. 
    Note that all these are constant except for the last, where each a cut can be crossed by at most one connection gadget from each path contributing at most three edges. 
    Hence, the cutwidth is at most $3n + k_0$.
\end{proof}

\begin{proof}[Proof of \cref{lb:theo:partition}]
    Suppose there exist $\varepsilon > 0$ and an algorithm $\mathcal{A}$ solving the \Tpartp problem in time $\ostar((\sqrt[3]{3}-\varepsilon)^{\ctw})$.
    Let $d$ be arbitrary but fixed.
    Choose an arbitrary value $0 < \delta < 3 - (\sqrt[3]{3} - \varepsilon)^3$.
    We can solve the $d$-CSP-3 problem as follows in time $\ostar((3-\delta)^n)$ where $n$ denotes the number of variables in $I$ contradicting SETH.
    Given an instance $I$ of $d$-CSP-$3$, 
    we first compute, in polynomial time, the graph $G$ as well as a linear arrangement of $G$ of cutwidth at most $3n+k_0$.
    The above claims show that $G$ admits a partition into triangles if and only if $I$ is satisfiable.
    Now we run the algorithm $\mathcal{A}$ on $G$, the running time is upper-bounded 
    \[\bigoh^*\big((\sqrt[3]{3}-\varepsilon)^{3n+k_0}) = \bigoh^*(\big(\sqrt[3]{3} - \varepsilon \big)^{3n}) \leq \bigoh^*((3-\delta)^n)
    \]
    to solve $I$.
\end{proof}

\section{Hamiltonian Cycle}\label{sec:hc-ub}

For a graph $G=(V,E)$, a set $C \subseteq E$ of edges is called a \emph{Hamiltonian cycle} of $G$ if $G[C]$ is connected and every vertex in $V$ has degree exactly two in $G[C]$, i.e., if $C$ induces a single cycle visiting all vertices of $G$.
In the \textsc{Hamiltonian Cycle} problem, we are given a graph $G$ and asked if there is a Hamiltonian cycle in $G$.
\probdef{\textsc{Hamiltonian Cycle}}{A graph $G=(V,E)$.}{Does $G$ admit a Hamiltonian cycle?} 

In this section we provide an algorithm solving the \textsc{Hamiltonian Cycle} problem in time $\bigoh^*((1 + \sqrt 2)^k)$ on graphs provided with a linear arrangement of cutwidth $k$.
For this we adapt the algorithm by Cygan et al.~\cite{DBLP:journals/jacm/CyganKN18} solving the problem in time $\bigoh^*((2 + \sqrt 2)^p)$
on graphs provided with a path decomposition of pathwidth $p$.
We will transform a linear arrangement of cutwidth $k$ into a path decomposition such that this path decomposition has a useful algorithmic property, namely, the dynamic-programming table as defined by Cygan et al.\ has only $\mathcal{O}^*((1 + \sqrt 2)^k)$ non-zero entries and there is an efficient way to determine the ``certainly zero'' entries. 
Now we provide 
the 
details.

\subsection{Known notation and results}
The algorithm by Cygan et al.~\cite{DBLP:journals/jacm/CyganKN18} strongly relies on their algebraic result about the $\mathbb{F}_2$-rank of a certain ``compatibility'' matrix reflecting, for each pair of perfect matchings, whether their union is a Hamiltonian cycle. We summarize the necessary parts of this result:
\begin{theorem}[\cite{DBLP:journals/jacm/CyganKN18}]\label{lem:number-of-base-matchings}
	Let $t \in \mathbb{N}$ be even.
	Then there exists a set $\mX_t$ of perfect matchings of the complete graph $K_t$ with the following properties:
	\begin{enumerate}
		\item $|\mX_t| = 2^{t/2-1}$,
		\item $\mX_t$ can be computed in time $\sqrt{2}^t t^{\bigoh(1)}$,
		\item for every matching $M \in \mX_t$, there exists a unique matching $M' \in \mX_t$ such that $M \cup M'$ forms a Hamiltonian cycle of $K_t$.
	\end{enumerate}
\end{theorem}
Observe that the third property implies that one can partition the set $\mX_t$ into pairs of matchings such that the union of two perfect matchings forms a Hamiltonian cycle, if and only if the two matchings are paired in this partition. 
In other words, the family $\mX_t$ induces a permutation submatrix of the compatibility matrix mentioned above.
Like Cygan et al., we will sometimes identify some ordered set, say $S$, of even cardinality $t$ with the vertex set $[t]$ of $K_t$ and by $\mX(S)$ denote the set obtained from $\mX_t$ by identifying the elements of $S$ with $[t]$. 

In the following, we assume that the graph $G$ is provided with a weight function $\omega \colon E(G) \to \bN$.
As many other algorithms solving connectivity problems, the algorithm of Cygan et al.~\cite{DBLP:journals/jacm/CyganKN18} makes use of the classic Isolation lemma (see \cite{DBLP:journals/combinatorica/MulmuleyVV87}) and samples $\omega$ in a certain probabilistic way.
Their algorithm works on a very nice path decomposition of the input graph $G$---in the next subsection we will show how to obtain a ``useful'' path decomposition of $G$ from a linear arrangement of $G$.  
Essentially, their algorithm processes the very nice path decomposition and for every bag counts, modulo 2, the number of subgraphs of the already processed graph such that every vertex in this subgraph has degree 2, except for vertices in the current bag which may have lower degree.
The dynamic-programming table refines these counts depending on the weight of the subgraph, its degree sequence on the bag, and whether this subgraph forms a single cycle with a certain matching defined on the vertex set of the bag.
A crucial implication of \cref{lem:number-of-base-matchings} (as proven in the their paper) is that instead of taking all such matchings into account, it suffices to consider only these ``base matchings'' from $\mX_t$ to solve the problem.
Now we summarize this more formally:

\begin{definition}[\cite{DBLP:journals/jacm/CyganKN18}]\label{def:dp-table}
A \emph{partial cycle cover} of a graph is a set of edges such that every vertex has at most two incident edges in this set.
For every node $x$ of the provided very nice path decomposition, every assignment $s \colon B_x \to \{0, 1, 2\}$ where $s^{-1}(1)$ has even cardinality, every matching $M \in \mX(s^{-1}(1))$, and every non-negative integer $w \in \bN$, the value $T_x[s, M, w]$ is defined as the number, modulo 2, of partial cycle covers $C$ of the graph $G_x$ with the following properties:
\begin{enumerate}
	\item $C \cup M$ is a single cycle,
	\item the total weight (with respect to $\omega$) of the edges in $C$ is $w$,
	\item every vertex $v \in B_x$ has precisely $s(v)$ incident edges in $C$,
	\item every vertex $v \in V(G_x) \setminus B_x$ has precisely two incident edges in $C$.
\end{enumerate}
\end{definition}
We say that a partial cycle cover $C$ of $G_x$ has \emph{footprint} $s$ on $x$ if it satisfies the last two items.
The main result of their paper can now be stated as follows. 
Let us remark that even though this is a more refined statement---they state the running time for the whole table $T_x$ at once---the version below immediately follows from their algorithm, in particular, that the computation of any fixed table entry requires only $\mathcal{O}(1)$ table entries of the predecessor node:
\begin{theorem}[\cite{DBLP:journals/jacm/CyganKN18}]\label{thm:single-entry-computation}
	Let $x$ be a non-first node of a very nice path decomposition of a graph $G$ and let $y$ denote its predecessor. 
	For any fixed $s \colon y \to \{0,1,2\}$, $M \in \mX(s^{-1}(1))$, and $w \in \bN$, given the table $T_y$, the value $T_x[s, M, w]$ can be computed in time $\bigoh^*(1)$ time by querying $\bigoh(1)$ entries of $T_y$.
\end{theorem}
We will also use some further results from this paper but for the sake of readability we cite them later. 
\subsection{Constructing a Path Decomposition}
Let $G$ be a graph and let $\ell = v_1, \dots, v_n$ be a linear arrangement of $G$ of cutwidth at most $k$.
The aim 
of this subsection 
is to compute from $\ell$ a path decomposition of $G$ with certain nice properties to which we will later apply the algorithm from \cref{thm:single-entry-computation}.
We recall that by definition for every $i \in [n]$, the set $Y_i$ consists of all left end-points of the edges in the $(i-1)$th cut $E_{i-1}$ of $\ell$ together with the vertex $v_i$.
Thus, we have $|E_i| \leq k$ and $|Y_i| \leq k + 1$ for all $i \in [n]$.
It is well-known that the sequence $Y_1, \dots, Y_n$ of bags is a path decomposition of $G$ (see e.g., \cite{DBLP:conf/stacs/GroenlandMNS22} for the idea and \cite{DBLP:conf/stacs/BojikianCHK23} for a formal proof).
To reverse the ordering in which these bags are traversed, for $i \in [n]$, we define the sets
\begin{equation}\label{eq:X-path-dec}
 	X_i = Y_{n+1-i}.
\end{equation}
 Thus $X_1, \dots, X_n$ is also a path decomposition. 

The key to obtaining our algorithm is the fact that for the path decomposition $X_1, \dots, X_n$, in general, not all degree sequences on bags can be achieved by partial cycle packings.
Recall that, except for $v_{n+1-i}$, the bag $X_i$ consists of the left end-points of edges in the cut $E_{n-i}$.
Let $x_i$ denote the node corresponding to the bag $X_i$.
The ordering of introduce- and forget-nodes we later choose to make this decomposition very nice ensures the graph $G_{x_i}$ contains no edges with both end-points in $X_i \setminus \{v_{n+1-i}\}$.
Therefore, 
all edges of any partial cycle cover of the graph $G_{x_i}$ incident with $X_i \setminus \{v_{n+1-i}\}$ cross the cut $E_{i-1}$.
Therefore, if $s$ is a footprint of a such a partial cycle cover, then it holds that $\sum_{u \in X_i \setminus \{v_{n+1-i}\}} s(u) \leq k$.
Furthermore, every vertex $u \in X_i \setminus \{v_{n+1-i}\}$ has at least $s(u)$ incident edges in $E_{i-1}$.
In the following, we capture this formally.

\begin{definition}
	For a node $x$ of a path decomposition we define the sets 
	\[
		B_1(x) = \{u \in B_x \mid \deg_{G_x}(u) = 1\} \text{ and } B_2(x) = \{u \in B_x \mid \deg_{G_x}(u) \geq 2\}.
	\]
	We call a mapping $s \colon B_x \to \{0, 1, 2\}$ \emph{relevant}  for $x$  if all of the following hold:
	\begin{itemize}
		\item $|s^{-1}(1)|$ is even,
		\item $s^{-1}(2) \subseteq B_2(x)$,
		\item and $s^{-1}(1) \subseteq B_1(x) \cup B_2(x)$.
	\end{itemize}
\end{definition}

\begin{lemma}\label{lem:number-of-footprints}
	Let $x$ be a node of a very nice path decomposition of $G$, let $P$ be a partial cycle cover of $G_x$, and let $s \colon B_x \to \{0, 1, 2\}$ such that $P$ has footprint $s$. Then $s$ is relevant.
 	
 	Furthermore, if $|B_1(x)| + 2\cdot |B_2(x)| \leq k+\mathcal{O}(1)$ holds, then the number of pairs $(s, M)$ such that $s \colon B_x \to \{0, 1, 2\}$ is relevant and $M \in \mX(s^{-1}(1))$ is upper-bounded by $\bigoh((1+\sqrt{2})^k)$.
\end{lemma}

\begin{proof}
	By definition of a partial cycle cover, the connected components of the subgraph induced by $P$ are paths and cycles.
	Furthermore, a vertex has the degree of $1$ in this subgraph if and only if it is a end-point of such a path of non-zero length.
	First, by definition of a footprint, all vertices $v \in V_x \setminus B_x$ have the degree of $2$ in $P$ and therefore, all end-points of such paths belong to $B_x$.
	Since each non-zero length path has precisely two end-points, we get that $|s^{-1}(1)|$ is even.
	
	Further, for every vertex $u \in B_x$ with $s(u) = 2$ we have $u \in B_2(x)$---this is because a vertex $u$ of degree $2$ in $P$, in particular, has at least two incident edges in $G_x$.
	Similarly, for every vertex $u \in B_x$ with $s(u) = 1$ we have $u \in B_1(x) \cup B_2(x)$---this is because a vertex $u$ of degree $1$ in $P$ has at least one incident edge in $G_x$.
	Thus, $s$ is relevant concluding the proof of the first claim.
	
	Now let $\ell_1 = |B_1(x)|$ and $\ell_2 = |B_2(x)|$ and suppose $\ell_1 + 2\ell_2 \leq k + \mathcal{O}(1)$ holds.
	Then we can upper-bound the number of pairs $(s, M)$ such that $s \colon B_x \to \{0, 1, 2\}$ is relevant and $M \in \mX(s^{-1}(1))$ by
	\[
		\sum_{0 \leq i_2 \leq \ell_2} {\ell_2 \choose i_2} \sum_{0 \leq i_1 \leq (\ell_2 - i_2) + \ell_1} {(\ell_2 - i_2) + \ell_1 \choose i_1} 2^{i_1/2 - 1}.
	\]
	This is for the following reason.
	The first sum considers all options for the set $s^{-1}(2)$: by above observation this set has to be a subset of $B_2(x)$.
	Similarly, the second sum considers the options for the set $s^{-1}(1)$ which, due to $s$ being relevant, has to be a subset of $(B_1(x) \cup B_2(x)) \setminus s^{-1}(2)$.
	Finally, for the fixed $s \colon B_x \to \{0, 1, 2\}$ with even $i_1 = |s^{-1}(1)|$ by \cref{lem:number-of-base-matchings} we have $\mX(i_1) = 2^{i_1/2 - 1}$, and therefore, that many options for the matching $M$. 
	In general, this upper bound can be proper since it permits odd values of $i_1$.
	In the remainder of the proof we show that the above value is upper-bounded by $\bigoh((1+\sqrt{2})^k)$.
	It holds that:
	{\allowdisplaybreaks
	\begin{align*}
		&\sum_{0 \leq i_2 \leq \ell_2} {\ell_2 \choose i_2} \sum_{0 \leq i_1 \leq (\ell_2 - i_2) + \ell_1} {(\ell_2 - i_2) + \ell_1 \choose i_1} 2^{i_1/2 - 1} = \\
		\frac{1}{2} &\sum_{0 \leq i_2 \leq \ell_2} {\ell_2 \choose i_2} \sum_{0 \leq i_1 \leq (\ell_2 - i_2) + \ell_1} {(\ell_2 - i_2) + \ell_1 \choose i_1} \sqrt{2}^{i_1} = \\
		\frac{1}{2} &\sum_{0 \leq i_2 \leq \ell_2} {\ell_2  \choose i_2} (1+\sqrt 2)^{(\ell_2 - i_2) + \ell_1} = \\
		\frac{1}{2} &(1+\sqrt{2})^{\ell_1} \sum_{0 \leq i_2 \leq \ell_2} {\ell_2 \choose i_2} (1+\sqrt 2)^{\ell_2 - i_2} = \\
		\frac{1}{2} &(1+\sqrt{2})^{\ell_1} (2 + \sqrt{2})^{\ell_2} \stackrel{2 + \sqrt{2} \leq (1 + \sqrt 2)^2}{\leq} \\
		\frac{1}{2} &(1+\sqrt{2})^{\ell_1} (1 + \sqrt{2})^{2 \ell_2} = \\
		\frac{1}{2} &(1+\sqrt{2})^{\ell_1 + 2 \ell_2} \stackrel{\ell_1 + 2\ell_2 \leq k + \mathcal{O}(1)}{\leq} \\
		\frac{1}{2} &(1+\sqrt{2})^{k + \mathcal{O}(1)} = \\
		&\mathcal{O}((1+\sqrt{2})^k) \\
	\end{align*} 
	}
	So the second claim of the lemma holds as well.
\end{proof}

Now we formally prove that we can obtain from $\ell$ a very nice path decomposition where every node $x$ satisfies the prerequisite of \cref{lem:number-of-footprints}, i.e., $B_1(x) + 2 \cdot B_2(x) \leq k+\bigoh(1)$.
The additional property about the last three bags is a technicality required later: shortly, the algorithm of Cygan et al.~\cite{DBLP:journals/jacm/CyganKN18} guesses an edge of the sought Hamiltonian cycle incident to a fixed vertex $v_1$, the vertex $v^*$ should be thought of as the other end-vertex of this edge.
\begin{lemma}\label{lem:useful-path-decomposition-hc-ub}
	From the linear arrangement $v_1, \dots, v_n$ of the graph $G$ and a vertex $v^*$ adjacent to $v_1$, 
	in polynomial time we can construct a very nice path decomposition of $G$ in which each node $x$ satisfies $|B_1(x)| + 2 \cdot |B_2(x)| \leq k+\bigoh(1)$.
	Furthermore, the last three nodes of this decomposition are introduce-edge-$v_1 v^*$, forget-vertex-$v_1$, and forget-vertex-$v^*$.
\end{lemma}

\begin{proof}
	To make the proof description simpler, we will first obtain a path decomposition that satisfies all conditions up to the one describing the last three bags.
	In the end we will show how to make this condition satisfied as well by an easy change. 
	
	We start with the path decomposition of $G$ consisting of the nodes $x_1, \dots, x_n$ corresponding to the bags $X_1, \dots, X_n$ (see \eqref{eq:X-path-dec} for the definition) and make it very nice as described next. 
	Let us remark that the transformation we will carry out to make the decomposition very nice is standard---we will insert introduce-vertex-, introduce-edge- and forget-vertex-nodes---so the arising object will indeed be a very nice path decomposition of $G$. 
	The ordering of these introduce- and forget-nodes will be chosen in a careful so way that the desired property on the degrees of the vertices in any bag is satisfied by the arising path decomposition as we show later.

	First, we add a node $x_0$ with an empty bag before $x_1$.
	After that, we add the following nodes between $x_0$ and $x_1$---we will list them in the order they appear between $x_0$ and $x_1$.
	Since $v_n \in X_1 = Y_n$, we add an introduce-vertex-$v_n$-node and denote it by $x_0^0$. 
	Now let $t_0 = |X_1 \setminus \{v_n\}|$ and let $\{w^1_0, \dots, w_0^{t_0}\} = X_1 \setminus \{v_n\}$.
	For $j = 1, \dots, t_0$ we insert the introduce-vertex-$w_0^j$-node $x_0^{0,j}$. 
	Now observe that $X_1 \setminus \{v_n\}$ is precisely the set of neighbors of $v_n$. 
	So for $j = 1, \dots, t_0$ we insert an introduce-edge-$w_0^j v_n$-node $x_0^{1,j}$.For simplicity of the notation later, we define $r_0 = t_0$ and $u_0^j = w_0^j$ for every $j \in [r_0]$.
	
	Similarly, for every $i = 2, \dots, n$ we will now insert some nodes between $x_{i-1}$ and $x_i$---again, we will list the inserted nodes in the order they appear between $x_{i-1}$ and $x_i$.
	First, we claim that $X_{i-1} \setminus X_i = Y_{n+1-(i-1)} \setminus Y_{n+1-i} = \{v_{n+1-(i-1)}\}$ holds---this is true for the following reason.
	First, $v_{n+1-(i-1)}$ belongs to $Y_{n+1-(i-1)}$ and not to $Y_{n+1-i}$ by definition.
	Second, the vertex $v_{(n+1)-i}$ belongs to $Y_{n+1-i}$.
	And finally, every edge that belongs to $E_{n+1-(i-1)-1} = E_{n+1-i}$ and has its left endpoint not equal to $v_{n+1-i}$ also belongs to $E_{n+1-i-1}$.
	So we first insert the forget-vertex-$v_{n+1-(i-1)}$-node and denote it by $x_{i-1}^{f}$.
	
	Second, recall that we have $v_{n+1-i} \in X_i = Y_{n+1-i}$. 
	So if $v_{n+1-i} \notin X_{i-1}$ we now insert the introduce-vertex-$v_{n+1-i}$-node and denote it by $X_{i-1}^{0}$.
	Third, let $\{w^1_{i-1}, \dots, w^{t_{i-1}}_{i-1}\} = X_i \setminus (X_{i-1} \cup \{v_{n+1-i}\}) = Y_{n+1-i} \setminus (Y_{n+1-(i-1)} \cup \{v_{n+1-i}\})$.
	For $j = 2, \dots, t_{i-1}$ we insert an introduce-vertex-$w_{i-1}^j$-node and denote it by $x_{i-1}^{0,j}$.
	Now the vertices in the bag of $x_{i-1}^{0,t_{i-1}}$ are precisely the vertices of $X_i$. 
	
	And fourth, let $\{u^1_{i-1}, \dots, u^{r_{i-1}}_{i-1}\}$ be the set of vertices to the left of $v_{n+1-i}$ adjacent to $v_{n+1-i}$.
	For $j = 1, \dots, r_{i-1}$ we add an introduce-edge-$u_{i-1}^j v_{n+1-i}$-node and denote it by $x_{i-1}^{1,j}$.
	Observe that the following holds:
	First, $v_{n+1-i} \in Y_{n+1-i} = X_i$ holds by definition.
	Second, we have $u_{i-1}^j = v_q$ for some $q < n+1-i$ and the existence of the edge $u_{i-1}^j v_{n+1-i}$ implies that we have $u_{i-1}^j \in Y_{n+1-i} = X_i$.
	Since the vertices in the bag of $x_{i-1}^{0,t_{i-1}}$ are precisely the vertices in $x_i$, the bag of $x_{i-1}^{1,j}$ indeed contains the end-vertices of the edge it introduces.	
	
	Next for every $i \in [n]$, we suppress the node $x_i$ as the vertices in the bag of $x_{i-1}^{1,r_i}$ are (still) precisely the vertices in $X_i$---for simplicity we refer to $x_{i-1}^{1,r_i}$ as $x_i$.
	Finally, note that we have $X_n = \{v_1\}$.
	So to conclude the construction, we add a forget-vertex-$v_1$-node after $x_n$ and denote it by $x_{n+1}$.
	The constructed decomposition now consists of the bags:
	\begin{align*}
		& x_0, \\ 
		& \Bigl((x_i^f)_{\text{if } i \geq 1}, (x_i^0)_{\text{if } v_{n+1-(i+1)} \notin X_i}, x_i^{0,1}, \dots, x_i^{0,t_i}, \\
		& x_i^{1,1}, \dots, x_i^{1,r_i-1}, x_i^{1,r_i}=x_{i+1}\Bigr)_{i = 0, \dots, n-1}, \\
		& x_{n+1}. 
	\end{align*}
	This is indeed a very nice path decomposition of $G$: as for introduce- and forget-vertex-nodes, this is because we use a standard construction of making a given path decomposition very nice; and every edge, say $v_{n+1-q} v_{n+1-i}$ with $i < q$ is introduced, by construction, precisely once, namely in $x_{i-1}^{1, j}$ for some $j \in [r_{i-1}]$.
	
	Now we show that each node, say $x$, of this decomposition satisfies the claimed inequality $|B_1(x)| + 2 \cdot |B_2(x)| \leq k+\mathcal{O}(1)$. 
	The bags of $x_0$ and $x_{n+1}$ are empty so we can now restrict our attention to the remaining ones.
	First, we show that the bags of $x_1 = x_0^{1,r_0}, \dots, x_n = x_{n-1}^{1,r_{n-1}}$ satisfy the desired property---we will later show how this implies the claim for the remaining bags as well. 
	Recall that by construction, for every $i \in [n]$, the edges introduced on the path between $x_{i-1}$ and $x_i$ are precisely the edges whose right end-point is $v_{n+1-i}$.
	This implies that the vertices introduced up to $x_i$ are precisely the edges $E_{n-1} \cup \dots \cup E_{n-i}$.
	Therefore, it holds that $G_x = (\{v_{n+1-i}, \dots, v_n\} \cup B_{x_i}, E_{n-i} \cup \dots \cup E_{n-1})$ (see \cref{fig:lin-arr-to-pd} for an illustration).
	
	\begin{figure}[t]
        \centering
        \includegraphics[width=0.4\textwidth]{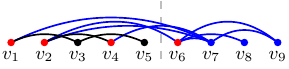}
        \caption{A linear arrangement of a graph $G$, the edges of $E_5$ are the edges crossing the gray vertical line. In red: the bag of $x_4$, in blue and red: the subgraph $G_{x_4}$, in black: the edges and vertices of $G$ outside $G_{x_4}$.}
        \label{fig:lin-arr-to-pd}
    \end{figure}
	
	Let $H = G_{x_i}$ and $v = v_{n+1-i}$ for shortness.
	For every vertex $v_j \in X_i \setminus \{v\} = Y_{n+1-i} \setminus \{v_{n+1-i}\}$ we have $j < (n+1)-i$, i.e., $j \leq n - i$.
	So for any two $u,w \in X_i \setminus \{v\}$ we have $uw \notin E(H)$---this is because the edge set of $H$ is $E_{n-i} \cup \dots \cup E_{n-1}$ and no edge of form $v_{j_1} v_{j_2}$ with $j_1, j_2 \leq n-i$ belongs to this set.
	Thus, for every edge $e$ of $H$ and every vertex $u \in X_i \setminus \{v\}$ incident with $e$, we have $e \in E_{((n+1)-i)-1} = E_{n-i}$ and $u$ is the left end-point of $e$.
		So we have 
		\begin{equation}\label{eq:bound}
			\sum_{u \in X_i \setminus \{v\}} \deg_H(u) = |E_{n-i}| \leq k.
		\end{equation}
		Thus, we get
		\begin{align*}
			&|B_1(x_i)| + 2 \cdot |B_2(x_i)| = \\
			&|\{u \in B_{x_i} \mid \deg_{H}(u) = 1\}| + 2\cdot |\{u \in B_{x_i} \mid \deg_{H}(u) \geq 2\}| \leq \\
			&|\{u \in X_i \mid \deg_{H}(u) = 1\}| + 2\cdot |\{u \in X_i \mid \deg_{H}(u) \geq 2\}| \leq \\
			&2 + |\{u \in X_i \setminus \{v\} \mid \deg_H(u) = 1\}| + 2\cdot|\{u \in X_i \setminus \{v\} \mid \deg_H(u) \geq 2\}| \leq \\
			&2 +\sum_{\substack{u \in X_i \setminus \{v\} \colon \\ \deg_H(u) = 1}} 1 + \sum_{\substack{u \in X_i \setminus \{v\} \colon \\ \deg_H(u) \geq 2}} \deg_H(u) = \\
			&2 + \sum_{u \in X_i} \deg_H(u) \stackrel{\eqref{eq:bound}}{\leq} \\
			&k + 2
		\end{align*}
		as desired.
		
		Now it remains to consider the ``intermediate'' nodes, say $x$, added in our construction of the very nice path decomposition and show that each of them satisfies 
		\begin{equation}\label{eq:desired-for-bags}
			\sum_{\substack{u \in B_x \colon \\ \deg_{G_x}(u) = 1}} 1 + \sum_{\substack{u \in B_x \colon \\ \deg_{G_x}(u) \geq 2}} 2  \leq k + 2
		\end{equation}
		which is then equivalent to $|B_1(x)| + 2 \cdot |B_2(x)| \leq k + 2$ by definition.
		We will now show that this result for the nodes $x_1, \dots, x_n$ implies the property for the remaining nodes as well:
		\begin{itemize}
		\item First, for every $i \in [n-1]$, we have $B_{x_i^f} \subseteq B_{x_i}$ and $G_{x_i^f} = G_{x_i}$.
		Therefore we also have $\deg_{G_{x_i^f}}(u) = \deg_{G_{x_i}}(u)$ for every $u \in B_{x_i^f}$ and
		\[
			\sum_{\substack{u \in B_{x_i^f} \colon \\ \deg_{G_{x_i^f}}(u) = 1}} 1 + \sum_{\substack{u \in B_{x_i^f} \colon \\ \deg_{G_{x_i^f}}(u) \geq 2}} 2 \leq \sum_{\substack{u \in B_{x_i} \colon \\ \deg_{G_{x_i}}(u) = 1}} 1 + \sum_{\substack{u \in B_{x_i} \colon \\ \deg_{G_{x_i}}(u) \geq 2}} 2 \leq k+2.
		\]
		So $x_i^f$ satisfies \eqref{eq:desired-for-bags} as $x_i$ does.
		
		\item Further, for every $i \in [n-1]_0$ such that the node $x_i^0$ exists, the path between $x_i^0$ and $x_{i+1}$ contains only introduce-vertex- and introduce-edge-nodes.
		So we have $B_{x_i^0} \subseteq B_{x_{i+1}}$ and $G_{x_i^0}$ is a subgraph of $G_{x_{i+1}}$.
		Therefore we also have $\deg_{G_{x_i^0}}(u) \leq \deg_{G_{x_{i+1}}}(u)$ for every $u \in B_{x_i^0}$ and
		\[
			\sum_{\substack{u \in B_{x_i^0} \colon \\ \deg_{G_{x_i^0}}(u) = 1}} 1 + \sum_{\substack{u \in B_{x_i^0} \colon \\ \deg_{G_{x_i^0}}(u) \geq 2}} 2 \leq \sum_{\substack{u \in B_{x_{i+1}} \colon \\ \deg_{G_{x_{i+1}}}(u) = 1}} 1 + \sum_{\substack{u \in B_{x_{i+1}} \colon \\ \deg_{G_{x_{i+1}}}(u) \geq 2}} 2 \leq k+2.
		\]
		So $x_i^0$ satisfies \eqref{eq:desired-for-bags} as $x_{i+1}$ does.
		
		\item And similarly, for every $i \in [n]$ and every $j \in [t_i]$, the path between $x_i^{0,j}$ and $x_{i+1}$ contains only introduce-vertex- and introduce-edge-nodes.
		So we have $x_i^{0,j} \subseteq x_{i+1}$ and $G_{x_i^{0,j}}$ is a subgraph of $G_{x_{i+1}}$.
		Therefore we also have $\deg_{G_{x_i^{0,j}}}(u) \leq \deg_{G_{x_{i+1}}}(u)$ for every $u \in B_{x_i^{0,j}}$, and $x_i^{0,j}$ satisfies \eqref{eq:desired-for-bags} as $x_{i+1}$ does.
		\item And finally, for every $i \in [n]$ and every $j \in [r_i]$ the path between $x_i^{1,j}$ and $x_{i+1}$ contains only introduce-vertex- and introduce-edge-nodes.
		Therefore, we have $B_{x_i^{1,j}} \subseteq B_{x_{i+1}}$ and $G_{x_i^{1,j}}$ is a subgraph of $G_{x_{i+1}}$.
		Therefore we also have $\deg_{G_{x_i^{1,j}}}(u) \leq \deg_{G_{x_{i+1}}}(u)$ for every $u \in B_{x_i^{1,j}}$, and $B_{x_i^{1,j}}$ satisfies \eqref{eq:desired-for-bags} as $x_{i+1}$ does. 
		\end{itemize}
		Altogether, every node $x$ of the constructed path decomposition satisfies $|B_1(x)|+2 \cdot |B_2(x)| \leq k + 2$.
		
		So it remains to ensure that the last three bags of the decomposition are as claimed in the lemma.
		Recall that the node $x_{n+1}$ is the (unique) forget-vertex-$v_n$-node.
		So we first take the (unique) forget-vertex-$v^*$-node and shift it to be placed right before $x_{n+1}$. 
		And second, we take the (unique) introduce-edge-$v_1 v^*$-node and move it right before the forget-vertex-$v^*$-node.
		Clearly, this remains a very nice path decomposition of $G$.
		To see that the bag of every node $x$ still satisfies $|B_1(x)|+2 \cdot |B_2(x)| \leq k + \bigoh(1)$, observe the following two properties.
		First, every bag gained at most one vertex due to these changes.
		Second, for every node $x$, the set of edges introduced so far (i.e., the set $E_x$) either remained unchanged, or it has lost precisely the edge $v_1 v^*$.
		Altogether, this implies that for every node $x$ the sum $|B_1(x)| + 2 \cdot |B_2(x)|$ increased by at most two, and is therefore bounded by $k + 4$.
\end{proof}
We can then run the dynamic-programming algorithm from \cref{thm:single-entry-computation} restricted to relevant footprints only to find out, with high probability, if the graph $G$ admits a Hamiltonian cycle. 
An important result implied by the paper by Cygan et al.\ can be stated as follows:
\begin{theorem}[\cite{DBLP:journals/jacm/CyganKN18}]\label{thm:cygan-et-al}
	If there exists a deterministic algorithm $\mathcal{A}$ that takes a graph $G$, a linear arrangement $v_1, \dots, v_n$ of $G$, a vertex $v^*$ adjacent to $v_1$, a weight function $\omega \colon E(G) \to [2\cdot|E(G)|]$, and a value $w^* \in [2 \cdot |V(G)| \cdot |E(G)|]_0$, runs in time $\alpha$, and outputs the number, modulo 2, of Hamiltonian paths of $G$ with end-vertices $v_1$ and $v^*$ of weight $w^*$, then there also exists an algorithm $\mathcal{B}$ with the following properties. 
	The algorithm $\mathcal{B}$ is a Monte-Carlo algorithm that takes a graph $G$ and a linear arrangement of $G$, runs in time $\ostar(\alpha)$, and solves the \textsc{Hamiltonian Cycle} problem.
	The algorithm $\mathcal{B}$ cannot give false positives and may give false negatives with probability at most $1/2$. 
\end{theorem}
Let us remark that even though the paper by Cygan et al.~\cite{DBLP:journals/jacm/CyganKN18} does not speak about linear arrangements and cutwidth at all, this theorem is implied from their results since it merely states that counting certain Hamiltonian paths modulo 2 suffices to solve the \textsc{Hamiltonian Cycle} problem.

We will now show how to use the path decomposition constructed in the \cref{lem:useful-path-decomposition-hc-ub} to obtain the desired $\bigoh^*((1+\sqrt{2})^k)$ algorithm: we will run the dynamic-programming from \cref{thm:single-entry-computation} on it but make use of \cref{lem:number-of-footprints} by only considering relevant degree sequences. 
Now we are ready to put everything together and prove the main result of this section:
\begin{theorem}
	There exists a one-sided error Monte-Carlo algorithm that takes a graph $G$ together with a linear arrangement 
	$v_1, \dots, v_n$ 
	of $G$ of cutwidth at most $k$, runs in time $\bigoh^*((1+\sqrt{2})^k)$, and solves the {\textsc{Hamiltonian Cycle}} problem.
	The algorithm cannot give false positives and may give false negatives with probability at most $1/2$.
\end{theorem}

\begin{proof}
	Let $v^*$ be an arbitrary neighbor of $v_1$, let $\omega \colon E(G) \to [2\cdot|E(G)|]$ be a weight function, and let $w^* \in [2 \cdot |V(G)| \cdot |E(G)|]_0$ be an integer.
	Let $e^* = \{v_1, v^*\}$.
	By \cref{thm:cygan-et-al}, it suffices to provide a deterministic algorithm that runs in time $\bigoh^*((1+\sqrt{2})^k)$ and computes, modulo 2, the number of Hamiltonian cycles of $G$ of weight $w^*$ with respect to $\omega$.
	
	First, in polynomial time we compute a very nice path decomposition, say $\mP$, of $G$ in which each node $x$ satisfies $|B_1(x)| + 2 \cdot |B_2(x)| \leq k + \bigoh(1)$, and the last three nodes of this decomposition are introduce-edge-$v_1 v^*$, forget-vertex-$v_1$, and forget-vertex-$v^*$ (recall \cref{lem:useful-path-decomposition-hc-ub}).
	Now we run the following dynamic-programming algorithm on $\mP$. 
	Informally speaking, we simply run the algorithm by Cygan et al.~\cite{DBLP:journals/jacm/CyganKN18} (see \cref{thm:single-entry-computation})---however, now by \cref{lem:number-of-footprints}, we know that some table entries are certainly zero, and we can save time by not storing them in the dynamic-programming table and instead, answering the corresponding queries by zero directly.
	Now we provide the details of this approach.
	
	For a node $x$ of $\mP$, we define the table $R_x$ as the restriction of $T_x$ to relevant mappings $s$ and $w \in [|V(G)| \cdot 2 |E(G)|]$, i.e., for every relevant mapping $s \colon B_x \to \{0, 1, 2\}$, every $M \in \mX(s^{-1}(1))$, and every $w \in [2 \cdot |V(G)| \cdot |E(G)|]_0$ we define $R_x[s, M, w] = T_x[s, M, w]$ (see \cref{def:dp-table}).
	Our algorithm processes the bags in the order they occur in $\mP$. 
	First, since the decomposition is very nice, the first bag, say $x_0$, is empty it is easy to see (and also argued by by Cygan et al.~\cite{DBLP:journals/jacm/CyganKN18}) that 
	$R_{x_0}[\emptyset, \emptyset, 0] = 1$
	and
	$R_{x_0}[\emptyset, \emptyset, q] = 0$
	hold for any $q \in [2 \cdot |V(G)| \cdot |E(G)|]$.
	
	Now we process the remaining nodes in the order they occur in $\mP$.
	So let $x$ be a non-first node of $\mP$, let $y$ be its predecessor, and let the table $R_y$ be already computed. 
	We now describe how to compute the table $R_x$ from it.
	For this, we iterate through all relevant $s \colon B_x \to \{0, 1, 2\}$, all $M \in \mX(s^{-1}(1))$, and all $w \in [2 \cdot |V(G)| \cdot |E(G)|]_0$ and proceed as follows.
	By \cref{thm:single-entry-computation} it is possible to compute the entry $T_x[s, M, w]$ in time $\bigoh^*(1)$ given the table $T_y$. 
	To compute all entries of $R_x$, we use the algorithm given by this theorem by answering the queries of form $T_y[s', M', w']$ for $s' \colon Y \to \{0, 1, 2\}$, $M' \in \mX(s'^{-1}(1))$, and $w' \in [2 \cdot |V(G)| \cdot |E(G)|]_0$ as follows:
	\begin{itemize}
	 \item If $s'$ is relevant for $y$, then $R_y[s', M', w']$ is well-defined and equal to $T_y[s', M', w']$. So we correctly answer the query with $R_y[s', M', w']$.
	 \item Otherwise, $s'$ is not relevant for $y$. By \cref{lem:number-of-footprints}, no partial cycle cover of $G_y$ has footprint $s'$ and therefore, we correctly answer the query $T_y[s', M', w']$ with $0$.
	\end{itemize}
	Recall that the node $x$ of $\mP$ satisfies $|B_1(x)| + 2 \cdot |B_2(x)| \leq k + \bigoh(1)$ (recall \cref{lem:useful-path-decomposition-hc-ub}). 
	Then by \cref{lem:number-of-footprints}, there exist $\bigoh((\sqrt{2}+1)^k)$ pairs $(s, M)$ where $s$ is relevant for $x$ and $M \in \mX(s^{-1}(1))$.
	Altogether, the table $R_x$ can therefore be computed from $R_y$ in time
	\[
		\bigoh((\sqrt{2}+1)^k) \cdot \bigoh(|V(G)| \cdot |E(G)|) \cdot \bigoh^*(1) = \bigoh^*((\sqrt{2}+1)^k).
	\]
	Recall that the $\mP$ was computed in polynomial time, hence the number of bags is polynomial in the size of $G$. 
	Altogether, we can compute the table $R_x$ for every node $X$ in time $\bigoh^*((\sqrt{2}+1)^k)$.
	
	Now let $z$ denote the node of $\mP$ preceding the introduce-edge-$v_1 v^*$-node. 
	The property on the last three bags of the very nice path decomposition $\mP$ implies that we have $B_z = \{v_1, v^*\}$ and $G_z = G - e^*$.
	By definition, the number of Hamiltonian paths between $v_1$ and $v^*$ of weight $w^*$ is precisely the number of partial cycle covers $P$ of $G_z$ with the footprint $s \colon \{v_1, v^*\} \to \{0,1,2\}$ such that $s(v_1) = s(v^*) = 1$ and such that for the perfect matching $M = \{\{v_1, v^*\}\} \in \mX(\{v_1, v^*\})$ the union $P \cup M$ forms a single cycle.
	In other words, the value $T_z[s, M, w]$ is the number (modulo 2) of Hamiltonian paths of $G$ with end-points $v_1$ and $v^*$.
	This value is equal to $0$ if $s$ is not relevant for $z$, and it is equal to $R_z[s, M, w^*]$ otherwise.
	Hence, the number, modulo 2, of Hamiltonian paths of $G$ with end-vertices $v_1$ and $v^*$ of weight $w^*$ can be computed in time $\bigoh^*((1+\sqrt{2})^k)$.
	\cref{thm:cygan-et-al} then implies the claim.
\end{proof}

  \subsection{Lower Bound}

In this section we provide a lower bound construction showing that unless SETH fails, there is no algorithm solving the \textsc{Hamiltonian Cycle} problem in time $\bigoh^*((1 + \sqrt{2} - \varepsilon)^{\operatorname{ctw}})$ for any $\varepsilon > 0$ even optimal linear arrangements are provided as part of the input.
This implies that our algorithm from the previous section is optimal under SETH.

As with the algorithm, we will use the ideas by Cygan et al.~\cite{DBLP:journals/jacm/CyganKN18} who showed that unless SETH fails, there is no algorithm solving the \textsc{Hamiltonian Cycle} problem in time $\bigoh^*((2 + \sqrt{2} - \varepsilon)^{\operatorname{pw}})$ for any $\varepsilon > 0$ even when optimal path decompositions are provided as part of the input.
Their construction has too large cutwidth to be used directly---this is because the path gadgets they use in their construction are, basically, cliques and although a clique of size $\beta$ ``contributes'' only $\beta$ to pathwidth, it contributes $\Theta(\beta^2)$ to cutwidth which is too large for the desired lower-bound construction.
We will adapt this construction in such a way that between two consecutive path gadgets with ``boundary size'' $\beta$ we have a matching, i.e., ``sufficiently small'' contribution of only $\beta$ to cutwidth. 
We will also adapt the clause-gadget construction in such a way that the cutwidth remains bounded.

\begin{definition}
	Let $G$ be a graph and let $A, B \subseteq V(G)$.
	The set $E(A)$ is defined as $E(A) = \{uv \in E(G) \mid u, v \in A\}$ and the set $E(A, B)$ is defined as $E(A, B) = \{uv \in E(G) \mid u \in A, v \in B\}$.
	Further let $\mF \subseteq 2^{E(A)}$.
	A Hamiltonian cycle $C$ of $G$ is \emph{consistent} with $(A, \mF)$ if $C \cap E(A) \in \mF$ holds.
\end{definition}

Cygan et al.\ have proven the following useful technical result:
\begin{lemma}[Induced Subgraph Gadget \cite{DBLP:journals/jacm/CyganKN18}]\label{lem:induced-subgraph-gadget}
	Let $G$ be a graph and let $ab$ be an edge of $G$ such that every Hamiltonian cycle of $G$ uses the edge $ab$.
	Further let $A \subseteq V(G)$ with $a, b \notin A$ and let $\mF \subseteq 2^{E(A)}$. 
	Then there is a computable function $f$ such that in time $f(|A|) \cdot n^{\bigoh(1)}$ we can compute a graph $G'$ with the following properties.
	Let $S = V(G') \setminus V(G)$.
	\begin{enumerate}
	 \item It holds that $|S| \leq f(|A|)$.
	 \item The graph $G'$ is obtained from $G$ by taking the graph $G$, adding the set $S$ of fresh isolated vertices, making $A \cup S \cup \{a, b\}$ to a clique, and then removing some edges with both endpoints in $A \cup S \cup \{a, b\}$.
	 \item The graph $G$ admits a Hamiltonian cycle consistent with $(A, \mF)$ if and only $G'$ admits a Hamiltonian cycle.
	\end{enumerate}
	In this context, we say that the vertices $A \cup S \cup \{a, b\}$ \emph{belong to the induced subgraph gadget} for $(A, \mF)$ in $G'$.
\end{lemma}

\subparagraph*{Reduction}

Now we are ready to describe the lower-bound construction itself.
Let $d \geq 3$ be an arbitrary but fixed integer treated as a constant in the following.
And let $I$ be an arbitrary instance of $d$-\textsc{SAT}.

To simplify the arguments the proof is structured as follows.
We will first construct a so-called \emph{constrained} graph $(G, W)$ where the set $W$ consists of tuples of form $(A, \mF, a, b)$ with $A \subseteq V(G)$, $\mF \subseteq 2^{E(A)}$, and $a,b \in V(G) \setminus A$ are adjacent vertices in $G$.
The constrained graph $(G, W)$ will, first, be such that every Hamiltonian cycle of $G$ uses the edge $ab$ for every $(A, \mF, a, b) \in W$.
Second, it will have the property that $I$ is satisfiable if and only if $G$ admits a Hamiltonian cycle $C$ consistent with $(A, \mF)$ for every $(A, \mF, a, b) \in W$---in this case we will also say that $C$ is \emph{consistent} with $W$.
Furthermore, the sets $A \cup \{a, b\}$ and $A' \cup \{a', b'\}$ will be disjoint for every $(A, \mF, a, b) \neq (A', \mF', a', b') \in W$. 
After that, the graph $G'$ will be obtained from $G$ by applying the induced subgraph gadget (see \cref{lem:induced-subgraph-gadget}) to every tuple $(A, \mF, a, b) \in W$.
Then the correctness of this gadget directly implies that $G'$ admits a Hamiltonian cycle if and only if $I$ is satisfiable.
After that, we will show that the cutwidth of $G'$ is bounded.

So we start with the description of the graph $G$.
Let $\beta$ and $\gamma$ be two constants chosen later in such a way that 
\begin{equation}\label{eq:desired-ineq-hc-lb}
	\sum_{\substack{i_1 + i_2 = \beta \colon \\ i_1 \text{ is even, } i_1, i_2 > 0}} {\beta - 1 \choose i_1 - 1} \sqrt{2}^{i_1/2 - 1} \geq 2^\gamma
\end{equation}
holds.
Further, let $X$ and $C = \{C_1, \dots, C_m\}$ denote the sets of variables and clauses of $I$, respectively, and let $n = |X|$.
We may assume, without loss of generality, that $m \geq 3$ as otherwise, the instance can be solved in polynomial time.
For every $j \in [m]$ let $b^j_1, \dots, b^j_{t(j)}$ denote the set of variables occurring in the clause $C_j$, in particular, we then have $0 < t(j) \leq d$.
We may assume that $n$ is divisible by $\gamma$ by adding at most $\gamma - 1$ variables not occurring in any clause.
Let $q = n / \gamma$.
We partition the set $X$ into $q$ blocks $X_1, \dots, X_q$, each of size $\gamma$.
For every $i \in [q]$, the $2^\gamma$ truth-value assignments of the block $X_i$ will be represented by different ways a Hamiltonian cycle of the arising graph can visit a certain set of vertices (making the graph constrained) and this set of vertices will ``contribute'' $\beta$ to the cutwidth yielding the total cutwidth of $\beta q + \bigoh_{\beta, d}(1)$.

We now proceed with a formal description of the construction of $G$ and $W$.
\begin{enumerate}
 \item Initialize the graph $G$ and the set $W$ to be empty.
 \item Let 
 \[
	\mA = \{(s, M) \mid s \in \{1,2\}^\beta, s(1) = 1, |s^{-1}(2)| > 0, |s^{-1}(1)| \text{ is even}, M \in \mX(s^{-1}(1))\}.
 \]
 Here $\mX(s^{-1}(1))$ is the set of matchings from \cref{lem:number-of-base-matchings}.
 Fix an injective mapping 
 \[
		\psi \colon \{0, 1\}^\gamma \to \mA;
 \]
 recall that $\beta$ and $\gamma$ satisfy \eqref{eq:desired-ineq-hc-lb} thus such a mapping exists. 
 \item For every $i \in [q]$, $j \in [m]$, $k \in [\beta]$, and $r \in \{0,1\}$, add the vertex $v^r_{i,j,k}$. 
 For $i \in [q]$ and $j \in [m]$, denote $V_{i,j} = \{v_{i,j,k}^r \mid k \in [\beta], r \in \{0, 1\}\}$, $V'_{i,j} = \{v_{i,j,k}^r \mid k \in [\beta] \setminus \{1\}, r \in \{0, 1\}\}$, $V'_i = \{v_{i,j,k}^r \mid j \in [m], k \in [\beta] \setminus \{1\}, r \in \{0, 1\}\}, $ and for $r \in \{0, 1\}$ denote $V_{i,j}^r = \{v_{i,j,k}^r \mid k \in [\beta]\}$. 
 \item Each set $V'_{i,j}$ forms a clique.
 \item For every $i \in [q]$ and $j \in [m]$ we also add a path $v_{i,j,1}^0 a_{i,j} b_{i,j} c_{i,j} d_{i,j} v_{i,j,1}^1$ where $a_{i,j}$, $b_{i,j}$, $c_{i,j}$, $d_{i,j}$ are fresh vertices. 
In particular, the vertices $b_{i,j}$ and $c_{i,j}$ are of degree 2 and therefore, any Hamiltonian path of the arising graph $G$ will necessarily contain the subpath $a_{i,j}$, $b_{i,j}$, $c_{i,j}$, $d_{i,j}$.
 \item For every $i \in [q]$, every $j \in [m]$, and every $k \in [\beta]$, we add the edge $v_{i,j,k}^1 v_{i,j+1,k}^0$. We call these edges \emph{horizontal}.
 \item For every $i \in [q]$, add a clique on $\beta$ fresh vertices denoted $K_\beta^i$, and add an edge between all vertices
of the clique $K_\beta^i$ and all vertices from $V_{i,1}^0$.
 \item For every $i \in [q]$, add a clique on $\beta$ fresh vertices denoted $\hat{K}_\beta^i$, and add an edge between all vertices of the clique $\hat{K}_\beta^i$ and all vertices from $V_{i,m}^1$.
 \item For every $i \in [q-1]$, add a vertex $y'_i$ and make it adjacent to all vertices of $K^i_\beta$ and $K^{i+1}_\beta$. Furthermore, a vertex $y'_q$ is added and made adjacent with all vertices of $K_\beta^q$ and $K^1_\beta$.
 \item For every $i \in [q]$, $j \in [m]$, and $\phi \in \{0,1\}^\gamma$ we define the set $\eta_{i,j}(\phi)$ as follows. 
 \begin{enumerate}
		\item Let $(s, M) = \psi(\phi)$. For simplicity we will sometimes also use $\eta_{i,j}(s, M)$ as a shortcut for $\eta_{i,j}(\phi)$.
		\item Initially, $\eta_{i,j}(\phi)$ is empty.
		\item Let $\{\alpha_1, \dots, \alpha_\ell\} = s^{-1}(2)$ where $\ell = |s^{-1}(2)|$. We add the edges $v_{i,j,\alpha_x}^0 v_{i,j,\alpha_{x+1}}^0$ and $v_{i,j,\alpha_x}^1 v_{i,j,\alpha_{x+1}}^1$ to $\eta_{i,j}(\phi)$ for every $x \in [\ell-1]$. And we also add the edge $v_{i,j,\alpha_\ell}^0 v_{i,j,\alpha_\ell}^1$ to $\eta_{i,j}(\phi)$.
		\item Now let $M' \in \mX(s^{-1}(1))$ be unique such that $M \cup M'$ forms a Hamiltonian cycle. 
		And let $\rho^0_1 \tau^0_1, \dots, \rho^0_u \tau^0_u$ be the edges of $M$ not incident with $1$ while $\rho^1_1 \tau^1_1, \dots, \rho^1_u \tau^1_u$ are the edges of $M'$ not incident with $1$ where $u = (|s^{-1}(1)| - 2)/2$.
		Further, let $y_2$ be such that $\{1, y_2\} \in M$ and let $z_2$ be such that $\{1, z_2\} \in M'$.
		For every $x \in [u]$ and $r \in \{0, 1\}$ we add the edge $v^r_{i,j,\rho^r_x} v^r_{i,j,\tau^r_x}$ to $\eta_{i,j}(\phi)$.
		\item We also add the edges $v_{i,j,y_2}^0 v_{i,j,\alpha_1}^0$ and $v_{i,j,z_2}^1 v_{i,j,\alpha_1}^1$ to $\eta_{i,j}(\phi)$.
 \end{enumerate}
 \item For every $j \in [m]$ let 
 \[
 	P_j = \{i \in [q] \mid X_i \cap \{b_1^j, \dots, b_{t(j)}^j\} \neq \emptyset \},
\]
i.e., the indices $i \in [q]$  such that some variable in $X_i$ occurs in $C_j$. Since $I$ is an instance of $d$-\textsc{SAT}, we have $|P_j| \leq d$. 
\begin{figure}[t]
	\centering
	\includegraphics[width=0.6\textwidth]{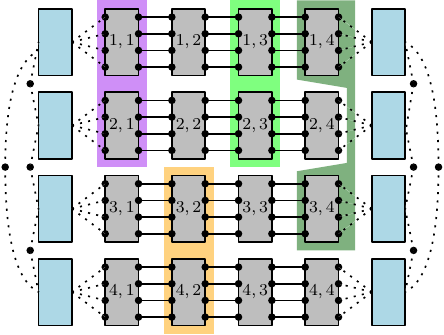}
	\caption{A sketch of the lower bound construction for $\beta = 4$ and $q = 4$. Cliques are blue, the sets $V_{i,j}$ are gray, for every column $j$, with a color we sketch the constraint $W_{\star, j}$. Dotted lines represent bicliques.}
	\label{fig:hc-lb}
\end{figure}
\begin{enumerate}
 \item\label{item:constraints-no-clause} For every $j \in [m]$ and every $i \in [q] \setminus P_j$ we add to the set $W$ the tuple $W_{i,j} = (V'_{i,j}, \mF_{i,j}, b_{i,j}, c_{i,j})$ where 
 \[
 	\mF_{i,j} = \{\eta_{i,j}(\phi) \mid \phi \in \{0,1\}^\gamma\}.
 \]
\item\label{item:clause-gadget} Let $p_1 \in P_j$ be arbitrary but fixed.
For every $j \in [m]$ we add to the set $W$ the tuple $W_{\star, j} = (V'_{\star,j}, \mF_j, b_{p_1, j}, c_{p_1, j})$ where
\[
	V'_{\star,j} = \bigcup_{p \in P_j} V'_{p, j}
\] 
and
\[
	\mF_j = \{\bigcup_{p \in P_j} \eta_{i,p}(\phi_p) \mid \forall p \in P_j \colon \phi_p \in \{0,1\}^\gamma, \exists p \in P_j \colon \phi_p \text{ satisfies } C_j\};
\]
Here ``$\phi_p \text{ satisfies } C_j$'' if and only if the truth-value assignment of $X_p$ obtained from $\phi_p$ by identifying $X_p$ with $\gamma$ satisfies $C_j$.
\end{enumerate}
\item For $i \in [q]$ and $j \in [m]$ we define $U_{i,j} = V_{i,j} \cup \{a_{i,j}, b_{i,j}, c_{i,j}, d_{i,j}\}$.
For $j \in [m]$ we also define $U_{\star, j} = \cup_{i \in P_j} U_{i,j}$ and for $i \in [q]$ we define $U_i = \cup_{j \in [m]} U_{i,j}$.
\end{enumerate}
This concludes the construction of $(G, W)$~(see \cref{fig:hc-lb} for an illustration).

Before moving on proving the correctness, let us make clear which parts of the construction are new. 
As for the ``path gadgets'', i.e., the graph induced on 
$\bigcup_{i,j} V_{i,j}$, we reuse them with the following difference.
In the construction of Cygan et al.~\cite{DBLP:journals/jacm/CyganKN18} the vertices $v_{i,j,k}^1$ and $v_{i,j+1,k}^0$ are merged into the same vertex: this way the pathwidth of the subgraph induced by these vertices is bounded by $\beta \cdot q + \bigoh(1)$ but the cutwidth is larger due to cliques induced by each of $V_{i,j}$.  
By having $v_{i,j,k}^1$ and $v_{i,j+1,k}^0$ being two different vertices with a single edge between them we ensure that the cutwidth of the arising subgraph is bounded by the desired value $\beta \cdot q + \bigoh(1)$.
However, this happens at the cost of decreasing the number of possible pairs $(s, M)$ we ``encode'' into a fixed $V_{i,j}$: this is because the $v_{i,j,k}^1$ and $v_{i,j+1,k}^0$ cannot be visited by a Hamiltonian cycle completely independently from each other (a formal proof will be provided later).
As for the clause gadgets we proceed differently.
There are two main reasons for this.
First, we doubt the correctness of the clause gadget by Cygan et al.~\cite{DBLP:journals/jacm/CyganKN18}: it is not a major issue and admits a quite simple fix.
Second, more crucially, the vertex being the clause gadget in their work can have unbounded degree which is inacceptable for us as that would make the cutwidth of the arising graph too large.
Instead, in our construction, the clause gadgets are the constraints added in \cref{item:clause-gadget}: crucially, since we reduce from $d$-SAT for a fixed $d$, the clause gadget for clause $C_j$ ``combines'' only a constant number of sets $V_{i,j}$ so that each clause gadget only contributes a constant number of edges to the total cutwidth. 
We will ensure that such edges do not overlap for different clauses thus implying that all clause gadgets totally increase the cutwidth of the construction by a constant only.

\begin{observation}
	The pair $(G, W)$ can be constructed from $I$ in polynomial time.
\end{observation}

\begin{claim}
	If $I$ is satisfiable, then the graph $G$ admits a Hamiltonian cycle consistent with~$W$.
\end{claim}
\begin{proof}
	Let $\phi \colon \{v_1, \dots, v_n\} \to \{0, 1\}$ be a truth-value assignment satisfying $I$.
	We now show how to construct a Hamiltonian cycle of $G$ satisfying $W$.
	First, let $i \in [q]$ be fixed. 
	Now let $(s, M) = \psi(\phi_{|_{X_i}})$.
	And let $M'$ be the unique matching in $\mX(s^{-1}(1))$ such that $M \cup M'$ form a Hamiltonian cycle.
	For every $j \in [m]$ we define the set $\mE_{i,j} = \eta_{i,j}(s, M)$.
	and 
	\[
		\mE_i = \bigcup_{j \in [m]} \mE_{i,j} \cup \bigcup_{j \in [m-1], k \in s^{-1}(1) \setminus \{1\}} \{v_{i,j,k}^1 v_{i,j+1,k}^0\}. 
	\]

	We now show that set $\mE_i$ induces a graph with a very particular structure.
	First, recall that by definition of $\mE_{i,j} = \eta_{i,j}(\psi_{X_i})$ a vertex $v_{i, j, k}^r$ (for $j \in [m]$, $k \in [\beta] \setminus \{1\}$, and $r \in \{0, 1\}$) has the degree of $s(k)$ in $\mE_{i,j}$.
	Therefore, $v_{i,j,k}^r$ has the degree of at most $2$ in $\mE_i$.
	Next we show that for every $j \in [m-1]$ the vertices $\{v_{i,j,k}^1, v_{i,j+1,k}^0 \mid k \in s^{-1}(1) \setminus \{1\}\}$ are connected by the path, denoted by $P_{i,j}$, with end-points $v_{i,j,z_2}^1$ and $v_{i,j+1,y_2}^0$. 
	This is true for the following reason.
	First, recall that $M \cup M'$ form a Hamiltonian cycle in the clique on the vertex set $s^{-1}(1)$. 
	So removing the vertex $1$ from this cycle, we obtain a path with end-points $y_2$ and $z_2$ visiting all elements of $s^{-1}(1) \setminus \{1\}$.
	The definitions of $\eta_{i,j}(s,M)$ and $\eta_{i,j+1}(s,M)$ then imply that in $\mE(i)$ we obtain the desired path $P_{i,j}$ from $v_{i,j,z_2}^1$ to $v_{i,j+1,y_2}^0$ that alternatingly visits the edges corresponding to $M'$ in $\eta_{i,j}(s,M')$, horizontal edges from $E(V_{i,j}, V_{i,j+1})$, the edges in $\eta_{i,j+1}(s,M)$ corresponding to $M$, and then horizontal edges from $E(V_{i,j}, V_{i,j+1})$ again.
	Now we are ready to describe the connected components of the graph induced by $\mE_i$.
	\begin{enumerate}
	 \item First, for every $x \in [u]$, there is a path on two vertices $v_{i,1,\rho_x^0}^0 v_{i,1,\tau_x^0}^0$.
	 \item Similarly, for every $x \in [u]$, there is a path on two vertices $v_{i,m,\rho_x^1}^1 v_{i,1,\tau_x^1}^1$.
	 \item And finally, the path 
	 \begin{align*}
	  &(v_{i,j,y_2}^0, v_{i,j,\alpha_1}^0, \dots, v_{i,j,\alpha_\ell}^0, v_{i,j,\alpha_\ell}^1, \dots, v_{i,j,\alpha_1}^1, v_{i,j,z_2}^1, P_{i,j}, v_{i,j+1,y_2}^0)_{j = 1, \dots, m-1}, \\
	  &v_{i,m,y_2}^0, v_{i,m,\alpha_1}^0, \dots, v_{i,m,\alpha_\ell}^0, v_{i,m,\alpha_\ell}^1, \dots, v_{i,m,\alpha_1}^1, v_{i,m,z_2}
	 \end{align*}
	 In simple words, this path repeats the following pattern for every $j \in [m]$ to visit all vertices of $V'_i$: 
	 First, on the side $r = 0$, start with the vertex with index $y_2$ and from it visit all vertices on the same side with indices from $s^{-1}(2)$.
	 After that we change the side to $r=1$ and visit the vertices with indices from $s^{-1}(2)$.
	 After that, we use the edge to the vertex with index $z_2$ on the same side, and finally, use the path $P_{i,j}$ to visit all the vertices with indices from $s^{-1}(1)\setminus \{1\}$ on that side as well as the vertices of $V_{i,j+1}^0$ with indices from $s^{-1}(1)\setminus \{1\}$.
	 This path then ends in $v_{i,j+1,y_2}^0$.
	\end{enumerate}
	
	Next, we define the edge set $\mJ_i$ as the edges of the path
	\[
		(v_{i,j,1}^0, a_{i,j}, b_{i,j}, c_{i,j}, d_{i,j}, v_{i,j,1}^1)_{j = 1, \dots, m}.
	\]
	Now the set $\mE_i \cup \mJ_i$ induces an even number of vertex-disjoint paths such that every vertex of $U_i$ is visited and moreover, we know, precisely, what the end-points of these paths are.
	Next, let $\{\xi^i_1, \dots, \xi^i_\beta\}$ be the elements of the clique $K_\beta^i$ and let $\{\hat\xi^i_1, \dots, \hat\xi^i_\beta\}$ be the elements of the clique $\hat K_\beta^i$.
	Now we define the set $\mK_i$ as follows:
	\begin{align*}
		\mK_i = &\{\xi^i_1 v^0_{i,1,1}\} \cup 
		\{\xi^i_\beta v_{i,1,\rho_1^0}^0\} \cup
		\{\xi^i_{x+1} v_{i,1,\tau_x^0}^0, \xi_i^{x+1}v_{i,1,\rho_{x+1}^0}^0 \mid x \in [u-1]\} \cup \\ 
		&\{v_{i,1,\tau_x^0}^0 \xi^i_{u+1}\} \cup 
		\{\xi^i_x \xi^i_{x+1} \mid x \in \{u+1, \dots, \beta-2\}\} \cup 
		\{\xi^i_{\beta-1} v_{i,1,y_2}^0\} \cup \\
		&\{v_{i,m,1}^1 \hat\xi^i_1\} \cup \{\hat\xi^i_1 v_{i,m,\rho_1^1}^1\} \cup \\
		&\{\hat\xi^i_{x+1} v_{i,m,\tau_{x}^1}^1, \hat \xi^i_{x+1} v_{i,m,\rho_{x+1}^1}^1 \mid x \in [u-1]\} \cup\\
		&\{v_{i,m,\tau_u^1}^1 \hat \xi^i_{u+1}\} \cup 
		\{\hat\xi^i_{x} \hat\xi^i_{x+1} \mid x \in \{u+1, \dots, \beta-1\}\} \cup 
		\{\hat\xi^i_\beta v_{i,m,z_2}^1\}
	\end{align*}
	Now the set $\mK_i \cup \mE_i \cup \mJ_i$ induces a single path visiting all vertices of $K_\beta^i$, $\hat K_\beta^i$, and $U_i$ and has end-points $\xi_i^1$ and $\xi_i^\beta$---let $\mP_i$ denote this path.
	Recall that by construction of $G$, first, the sets $\mP_i$ and $\mP_j$ are vertex-disjoint for any pair $i \neq j \in [q]$.
	And furthermore, the vertices of $G$ not visited by any of $\mP_1, \dots, \mP_q$ are precisely $y_1', \dots, y_q'$.
	We finally obtain a Hamiltonian cycle $C$ of $G$ as
	\[
		C = (\xi_1^i, \mP_i, \xi_\beta^i, y_i')_{i = 1, \dots, n}.
	\]
	By construction, it indeed visits every vertex of $G$ precisely once.
	
	It remains to argue that $C$ is consistent $W$.
	For this we show that all constraints in $W$ are satisfied.
	Let $j \in [m]$ be arbitrary.
	First, let $i \in [q] \setminus P_j$. 
	The intersection of $C$ with $E(V'_{i,j})$ is, by construction, precisely $\eta_{i,j}(\phi_{|_{X_i}}) \in \mF_{i,j}$, thus $C$ is consistent with $W_{i,j}$.

	Next recall that for any $i \neq i' \in [q]$, the graph $G$ contains no edges between $V'_{i,j}$ and $V_{i',j}$.
	Therefore we have
	\[
		C \cap E(V'_{\star, j}) = C \cap (\bigcup_{p \in P_j} E(V'_{p,j})) = \bigcup_{p \in P_j} (C \cap E(V'_{p, j})) = \bigcup_{p \in P_j} \eta_{p, j}(\phi_{|_{X_p}})
	\]
	where the last equality holds by construction of $C$.
	Since $\phi$ is a satisfying assignment of $I$, there exists an index $p \in P_j$ such that $\phi_{|_{X_p}}$ satisfies $C_j$.
	Therefore, the above intersection is an element of $\mF_j$ and $C$ is also consistent with $W_j$.
	Altogether, $C$ is a Hamiltonian cycle of $G$ consistent with $W$.
\end{proof}

\begin{claim}
	If the graph $G$ admits a Hamiltonian cycle consistent with $W$, then $I$ is satisfiable.
\end{claim}

\begin{proof}
	Let $C$ be a Hamiltonian cycle of $G$ consistent with $W$.
	First, we show that for every $i \in [q]$ and every $j \in [m]$, there exists a truth-value assignment $\phi_{i,j} \in \{0, 1\}^\gamma$ such that $C \cap E(V'_{i,j}) = \eta_{i,j}(\phi_{i,j})$: 
	If $i \notin P_j$, this follows since $C$ is consistent with $W_{i,j}$ (see \cref{item:constraints-no-clause}).
	Since $C$ is consistent with $W_{\star, j}$ (see \cref{item:clause-gadget}) we have 
	\[
		C \cap E(V'_{\star, j}) = \bigcup_{p \in P_j} \eta_{p, j}(\phi_{p, j})
	\]
	for some $(\phi_{p, j} \in \{0,1\}^\gamma)_{p \in P_j}$ such that there exists $p \in P_j$ such that $\phi_{p,j}$ as an assignment to $X_p$ satisfies $C_j$.
	Recall that first, by construction, the set $\eta_{p, j}(\phi_p)$ contains only edges with both end-points in $V'_{p, j}$ and second, the set $E(V'_{p, j}, V'_{p', j})$ is empty for any $p \neq p' \in [q]$.
	Thus for every $p \in P_j$ we get 
	\[
		C \cap E(V'_{p, j}) = \eta_{p, j}(\phi_{p, j})
	\]
	as claimed.
	
	Let now $i \in [q]$ be arbitrary but fixed.
	We will now show that for any $j \in [m-1]$ it holds that $\phi_{i,j} = \phi_{i,j+1}$.
	For this let $(s_j, M_j) = \psi(\phi_{i,j})$ and let $(s_{i,j+1}, M_{i,j+1}) = \psi(\phi_{i,j})$.
	Let $k \in [\beta]$ be arbitrary.
	Observe that by construction the vertex $v_{i,j,k}^1$ has precisely $s_{i,j}(k)$ incident edges in $C \cap E(V'_{i,j}) = \eta_{i,j}(s_{i,j}, M_{i,j})$ and similarly, the vertex $v_{i,j+1,k}^0$ has precisely $s_{j+1}(k)$ incident edges in $C \cap E(V'_{i,j+1}) = \eta_{i,j+1}(s_{i,j+1}, M_{i,j+1})$.
	Further, the edge $v_{i,j,k}^1 v_{i,j+1,k}^0$ is the only edge incident to $v_{i,j,k}^1$ with the other end-point outside $V'_{i,j}$ and similarly, it is the only edge incident to $v_{i,j+1,k}^0$ with the other end-point outside $V'_{i,j+1}$.
	So if we have $s_{i,j}(k) = 1$, then the Hamiltonian cycle $C$ uses the edge $v_{i,j,k}^1 v_{i,j+1,k}^0$ and we thus have $s_{j+1}(k) \neq 2$.
	Therefore, $s_{i,j}(k) = 1$ implies $s_{i,j+1}(k) = 1$. 
	A symmetric argument shows that $s_{i,j+1}(k) = 1$ implies $s_{i,j}(k) = 1$. 
	Hence, $s_{i,j} = s_{i,j+1}$ holds.
	
	Now we show that $M_{i,j} = M_{i,j+1}$ holds too.
	Let $M'_j$ and $M'_{j+1}$ be the unique matchings from $\mX(s^{-1}_{i,j}(1)) = \mX(s^{-1}_{i,j+1}(1))$ such that both $M_{i,j} \cup M'_{i,j}$ and $M_{i,j+1} \cup M'_{i,j+1}$ are Hamiltonian cycles.
	Now suppose that $M'_{i,j} \neq M_{i,j+1}$ is not a Hamiltonian cycle.
	Since every vertex in $\mX(s^{-1}_j(1))$ has the degree of two in a union $M'_{i,j} \cup M_{i,j+1}$ of two matchings, the set $M'_{i,j} \cup M_{i,j+1}$ induces a union of at least two cycles.
	So let $D$ be a cycle in $M'_{i,j} \neq M_{i,j+1}$ that does not contain the vertex $1$.
	Let $d_1, \dots, d_t$ be the vertices of $D$ in the order they occur on this cycle.
	For simplicity of notation, let $d_{t+1} = d_1$.
	This cycle uses the edges of $M'_{i,j}$ and $M_{i,j+1}$ alternatingly (in particular, $t$ is even), without loss of generality we may assume that the edge $d_1 d_2$ belongs to $M'_{i,j}$.
	Consider the sequence
	\begin{equation}\label{eq:smaller-cycle}
		(v_{i,j,d_k}^1, v_{i, j, d_{k+1}}^1, v_{i, j+1, d_{k+1}}^0, v_{i, j+1, d_{k+2}}^0)_{k = 1, 3, \dots, t - 3, t - 1} 
	\end{equation}
	forming a cycle in $G$.
	Let $k \in \{1, 3, \dots, t - 3, t - 1\}$ be arbitrary.
	The edge $v_{i,j,d_k}^1 v_{i, j, d_{k+1}}^1$ belongs to $C$ by definition of $\eta_{i,j}(s_{i,j}, M_{i,j})$.
	Similarly, the edge $v_{i, j+1, d_{k+1}}^0 v_{i, j+1, d_{k+2}}^0$ belongs to $C$ by definition of $\eta_{i,j+1}(s_{i,j+1}, M_{i,j+1})$. 
	Finally, the edges $v_{i, j, d_{k+1}}^1 v_{i, j+1, d_{k+1}}^0$ and $v_{i, j+1, d_{k+2}}^0 v_{i, j, d_{k+2}}^1$ belong to $C$ as it holds that $s_{i,j}(d_{k+1}) = s_{i,j}(d_{k+2}) = 1$ and we argued above that in this case such edges belong to $C$. 
	Thus, the sequence \eqref{eq:smaller-cycle} witnesses that in the graph induced by $C$ there is a cycle that, for example, does not visit the vertex $v_{1,1,1}$---this contradicts the fact that $C$ is a Hamiltonian cycle of $G$.
	Hence, $M'_{i,j} \cup M_{i,j+1}$ forms a Hamiltonian cycle.
	Recall that for every matching in $\mX(s_{i,j}^{-1}(1)) = \mX(s_{i,j+1}^{-1}(1))$ there is a unique matching in this set such that their union forms a Hamiltonian cycle.
	So we have $M_{i,j+1} = M_{i,j}$ as claimed.
	Recall that $\psi$ is injective.
	Then $\psi(\phi_{i,j}) = (s_{i,j}, M_{i,j}) = (s_{i,j+1}, M_{i,j+1}) = \psi(\phi_{i,j+1})$ implies $\phi_{i,j} = \phi_{i,j+1}$ as claimed.
	
	With this we define the truth value assignment $\phi \colon \{v_1, \dots, v_n\} \to \{0,1\}$ as follows.
	Recall that $X_1, \dots, X_q$ partition the set of variables of $I$.
	Thus it suffices to define the restriction of $\phi$ to $X_i$ for every $i \in [q]$.
	So we let $\phi_{|_{X_i}} = \phi_{i,1}$.
	As shown above it then holds that $\phi_{|_{X_i}} = \phi_{i,j}$ for every $j \in [m]$.
	It remains to argue that $\phi$ satisfies $I$.
	So let $j \in [m]$ be arbitrary.
	Recall that by definition of $\phi_{i,j}$ for $i \in [q]$ we have
	\begin{align*}
		&C \cap E(V'_{\star, j}, V'_{\star, j}) \stackrel{E(V'_{p, j}, V'_{p', j}) \text { for } p \neq p' \in [q]}{=} C \cap \bigcup_{p \in P_j} E(V'_{p, j}, V'_{p, j}) = \\
		&\bigcup_{p \in P_j} C \cap E(V'_{p, j}, V'_{p, j}) = \bigcup_{p \in P_j} \eta_{p,j}(\phi_{p, j}) 
	\end{align*}
	Since $C$ is consistent with $W$, we have $\bigcup_{p \in P_j} \eta_{p,j}(\phi_{p, j}) \in \mF_j$.
	So by definition of $\mF_j$ there exists an index $p \in P_j$ such that $\phi_{p, j} = \phi_{|_{X_p}}$ satisfies $C_j$, i.e., $\phi$ satisfies $C_j$.
\end{proof}

Now we are ready to conclude the description of the reduction.
For this, we first recall that in the graph $G$ for every $i \in [q]$ and $j \in [m]$ the vertices $b_{i,j}$ and $c_{i,j}$ have the degree of two and are adjacent.
Thus any Hamiltonian cycle of $G$ uses the edge between these vertices.
Further, for every $(A, \mF, a, b) \in W$ we have $a, b \notin A$ by construction.
And finally, for every $(A, \mF, a, b) \neq (A', \mF', a', b') \in W$ we the sets $A \cup \{a, b\}$ and $A' \cup \{a', b'\}$ are disjoint.
Thus, the graph $G'$ obtained from $G$ by replacing every constraint $U = (A, \mF, a, b) \in W$ by an induced subgraph gadget $(A, a, b, \mF)$ from \cref{lem:induced-subgraph-gadget} is well-defined.
For every $j \in [m]$ and $i \in ([q] \cup \{\star\}) \setminus P_j$, let $S_{i,j}$ be the union of vertices $U_{i,j}$ together with the set of all vertices of $G'$ that belong to the induced subgraph gadget for the element $W_{i,j} \in W$.
Recall that the size of $U_{i,j}$ is linear in $\beta$ for every $i \in [q]$.
Furthermore, $I$ is an instance of $d$-SAT, thus the size of $U_{\star, j}$ is linear in $d \cdot \beta$.
Therefore, by \cref{lem:induced-subgraph-gadget}, there exists a computable function $g \colon \mathbb{N} \times \mathbb{N} \to \mathbb{N}$ such that for every $j \in [m]$ and $i \in ([q] \cup \{\star\}) \setminus P_j$ we have
\begin{equation}\label{eq:bound-sij}
	|S_{i,j}| \leq g(\beta, d).
\end{equation}
The correctness of the induced subgraph gadget (\cref{lem:induced-subgraph-gadget}) together with above claims immediately implies the following:
\begin{corollary}
	The instance $I$ is satisfiable if and only if the graph $G'$ admits a Hamiltonian cycle.  Furthermore, the graph $G'$ can be computed from $I$ in polynomial time.
\end{corollary}

The last missing piece of the proof is the following claim providing an upper bound on the cutwidth of $G'$:
\begin{claim}
	The cutwidth of $G'$ is bounded by $\beta \cdot q + \bigoh_{\beta, d}(1)$.
	Furthermore, a linear arrangement of $G'$ of such cutwidth can be computed from $I$ in polynomial time.
\end{claim}

\begin{proof}
	To prove the claim we provide the desired linear arrangement. 
	By construction, the vertex set of $G'$ consists of the vertices $y_1', \dots, y_q'$, the vertices of the cliques $K_\beta^1, \dots, K_\beta^q$ and $\hat K_\beta^1, \dots, \hat K_\beta^q$ and the sets $S_{i,j}$ for every $j \in [m]$ and $i \in ([q] \cup \{\star\}) \setminus P_j$.
	All of these sets are pairwise disjoint, i.e., they partition $V(G')$. 
	We define the linear arrangement, denoted by $\ell$, in which the vertices of $G'$ occur in the following ordering from left to right:
	\begin{enumerate}
		\item First, for every $i = 1, \dots, q$ with $i \notin P_1$ we place the vertices of $K_\beta^i$ (in some arbitrary ordering), after that the vertices of $S_{i,1}$ (in some arbitrary ordering) followed by the vertex $y_i'$, and then we move to the next $i$.
		\item Now we place the vertices of $(\bigcup_{i \in P_1} K_\beta^{i})$ (in some ordering), after that the vertices of $S_{\star,1}$ (in some ordering) followed by the vertices $y_i'$ with $i \in P_1$.
		\item Next, for $j = 2, \dots, m-1$ we place the following vertices (i.e., we only process $j+1$ after processing $j$):
		\begin{enumerate}
			\item First, for every $i \in [q] \setminus P_j$ we add the vertices from $S_{i,j}$ in an arbitrary ordering. 
			\item After that, we place all vertices of the set $S_{\star, j}$ in an arbitrary way. 
		\end{enumerate}
		\item Now for $i = 1, \dots, q$ with $i \notin S_m$ we place the vertices of $\hat K_\beta^i \cup S_{i,m}$ in an arbitrary ordering 
		and move to the next $i$.
		\item Finally, we place all vertices of $(\bigcup_{i \in P_m} \hat K^i_\beta) \cup S_{\star,m}$.
	\end{enumerate}
	We refer to \cref{fig:hc-lb-linear-arrangement} for an illustration.
	Observe that this is a valid linear arrangement of $G'$ as every vertex was placed precisely once, and it can also be computed in polynomial time from $I$.
	
	\begin{figure}[t]
		\centering
		\includegraphics[width=0.6\textwidth]{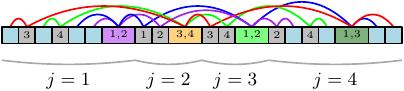}
		\caption{The constructed linear arrangement of the graph $G$ from \cref{fig:hc-lb}. The small numbers inside each ``block'' reflect the value of $i$. 
		For $i = 1, 2, 3, 4$, we fix the color blue, purple, red, and green, respectively, and represent with one line of this color all edges between $V_{i,j}$ and $V_{i,j+1}$ for every $j$. 
		The lines of the same color do not overlap for different values of $j$, and for a fixed value of $j$, there are $\beta$ of them.
		Every vertical line is crossed by at most one line of each color yielding the cutwidth of essentially $q \cdot \beta$.}
		\label{fig:hc-lb-linear-arrangement}
	\end{figure}

	Now we argue that every cut of this linear arrangement is crossed by at most $q \cdot \beta + \bigoh_{\beta, d}(1)$ edges.
	We say that two edges are \emph{non-overlapping} if there is no cut in this linear arrangement crossed by both of these edges.
	Let us recall that $I$ is an instance of $d$-\textsc{SAT} and therefore, we have $|P_j| \leq d$ and thus, $|P_j| \in \bigoh_d(1)$ for every $j \in [m]$---this will be crucial for the estimation.
	We will now cover the edges of $G'$ by several groups and provide an upper-bound on the value that each group contributes to any cut:
	\begin{itemize}
		\item First of all, the edges incident to $y_q'$ contribute at most $2\beta = \bigoh_\beta(1)$ to each cut.
		\item Next we consider the edges incident with $\{y_1', \dots, y_{q-1}'\}$. 
		If the vertices of $\{y_1', \dots, y_{q-1}'\} \cup \bigcup_{i \in [q]} V(K_\beta^i)$ would be ordered as 
		\begin{equation}\label{eq:ideal-clique}
			V(K_\beta^1), y_1', V(K_\beta^2), y_2', \dots, V(K_\beta^{q-1}), y_{q-1}', V(K_\beta^q), 
		\end{equation}
		then every cut would be crossed by at most $2 \beta$ such edges: this is because an edge incident with some $y_i'$ and an edge incident with some $y_{i'}'$ are non-overlapping in this ordering whenever $i, i' \in [q-1]$ and $|i - i'| > 1$.
		The ordering of $\{y_1', \dots, y_{q-1}'\} \cup \bigcup_{i \in [q]} V(K_\beta^i)$ in $\ell$ is obtained from \eqref{eq:ideal-clique} by taking at most $|P_1|$ indices $i$ and taking the vertices of the clique $V(K_\beta^i)$ as well as the vertex $y_i'$ and moving them elsewhere.
		Every vertex in $V(K_\beta^i) \cup \{y_i'\}$ has at most $2\beta$ incident edges with an end-point in $\{y_1', \dots, y_{q-1}'\}$.
		Hence, moving them around increases the cutwidth by at most $|P_1| \cdot (\beta + 1) \cdot 2 \beta$.
		Altogether, such edges thus contribute at most $2\beta + |P_1| \cdot (\beta + 1) \cdot 2 \beta = \bigoh_{\beta, d}(1)$ to any cut.
		\item Again, consider an arbitrary but fixed cut. 
		Similarly, we now show that there exist at most $|P_1|$ indices $i \in [q]$ such that the edges with one end-point in $K_\beta^i$ and the other end-point not equal to $y_i'$ cross this cut.
		To see this we consider the sets 
		\[
			\{u v \in E(G') \mid u \in V(K_\beta^i), v \neq y_i'\}_{i \in [q] \setminus P_1}, \{u v \mid u \in \bigcup_{i \in P_1} V(K_\beta^i), v \notin \{y_i' \mid i \in P_1\}\}.
		\]
		Now any two edges that belong to two different of these sets are non-overlapping.
		Thus there indeed exist at most $|P_1|$ indices with the claimed property.
		Now recall that for any fixed $i \in [q]$ the clique $K_{\beta}^i$ consists of $\beta$ vertices and each of them has the degree $2 \beta$.
		Thus, any cut is crossed by at most $|P_1| \beta^2 \in \bigoh_{\beta, d}(1)$ such edges.
		\item Analogously, for any fixed cut, there exist at most $|P_m|$ indices $i \in [q]$ such that the cut is crossed by edges with one end-point in $\hat{K}_\beta^i$.
		Any vertex in $\hat{K}_\beta^i$ has the degree of $2 \beta - 1$ and therefore, the cut is crossed by at most $|P_m| (2 \beta - 1) = \bigoh_{\beta, d}(1)$ such edges. 
		Note that the case described so far already cover all edges apart from those that have both end-points in $\bigcup_{j \in [m]} \bigcup_{i \in ([q] \cup \{\star\}) \setminus P_j} S_{i,j}$.
		\item Now we analyze the horizontal edges.
		Consider an arbitrary cut and let $j^* \in [m]$ be maximal such that for every $j \leq j^*$ all vertices of $\bigcup_{i \in ([q] \cup \{\star\}) \setminus P_j} S_{i, j}$ are to the left of this cut.
		The maximality implies that by construction of $\ell$, for every $j \in [m]$ with $j \geq j^*+2$ all vertices of $S_{i,j}$ are to the right of this cut.
		The choice of $j^*$ thus already implies that no edge of form $v_{i, j, k}^1 v_{i, j+1, k}^0$ crosses this cut for all $j \in [j^*-1] \cup \{j^*+2, \dots, m\}$ so we remain with the values $j \in \{j^*, j^*+1\}$.
		To simplify the upcoming definition we define the total ordering on the set $([q]  \cup \{\star\}) \setminus P_{j^*}$ such that $\star$ is maximal and $[q] \setminus P_{j^*}$ are ordered as natural numbers.
		So now let $i^* \in ([q] \cup \{\star\}) \setminus P_{j^*}$ be maximal such that for all $i \leq i^* \in ([q] \cup \{\star\}) \setminus P_{j^*}$ all vertices of $S_{i, j^*}$ are to the left of this cut.
		If $i^* = \star$ and $j^* = m$ hold, then no horizontal edges cross this cut.
		If $i^* = \star$ and $j^* \neq m$ hold, then 
		$q \cdot \beta$ horizontal edges cross this cut: those are precisely all horizontal edges of form $v_{i,j^*,k}^1 v_{i,j^*+1,k}$ where $i \in [q]$ and $k \in [\beta]$. 
		So we may now assume that $i^* \neq \star$ holds.
		Then:
		\begin{itemize}
			\item First, no edge of form $v_{i,j^*,k}^1 v_{i,j^*+1,k}^0$ crosses the cut for $i \leq i^*$.
			\item And second, no edge of form $v_{i,j^*+1,k}^1 v_{i,j^*+2,k}^0$ crosses the cut for $i > i^* + 1$.
		\end{itemize}
		Thus, only the following horizontal edges may cross the currently considered cut:
		\begin{itemize}
			\item The edges of form $v_{i, j^*, k}^1 v_{i, j^*+1, k}^0$ for $i > i^*$---there are at most $(q - i^*) \beta + |P_{j^*}| \beta$ such edges where the second addend is for $i = \star$.
			\item And the edges of form $v_{i, j^*+1, k}^1 v_{i, j^*+2, k}^0$ (if $j^*+2$ exists) for $i \leq i^*+1$---there are at most $(i^*+1) \beta + |P_{j^*}| \beta$ such edges where the second addend is for $i = \star$.
		\end{itemize}
		Altogether, we obtain that any fixed cut is crossed by at most $q \beta + \bigoh_{\beta, d}(1)$ horizontal edges.
		
		\item Finally, it remains to consider the edges with both endpoints in $S_{i,j}$ for some $j \in [m]$ and $i \in ([q] \cup \{\star\}) \setminus P_j$---for this recall that the vertices of $S_{i,j}$ are consecutive in $\ell$.
		Thus for any $(i, j) \neq (i', j')$ with $j, j' \in [m]$, $i \in ([q] \cup \{\star\}) \setminus P_j$, and $i' \in ([q] \cup \{\star\}) \setminus P_{j'}$ any edge with both end-points in $S_{i,j}$ is non-overlapping with any edge with both end-points in $S_{i',j'}$.
		Now recall that for any $j \in [m]$ and $i \in ([q] \cup \{\star\}) \setminus P_j$, the size of $S_{i, j}$ is bounded by $g(\beta, d)$ (see \eqref{eq:bound-sij}) and thus, the number of edges with both end-points in this set is bounded by $g(\beta, d)^2$.
		Altogether, any cut is crossed by $\bigoh_{\beta, d}(1)$ such edges.
	\end{itemize}
	By summing up the upper bounds over all edge types, we obtain that every cut of $\ell$ is crossed by at most $q \cdot \beta + \bigoh_{\beta, d}(1)$ edges as claimed.
\end{proof}
Now we are ready to put everything together to prove the desired lower bound.
\begin{theorem}
	Unless SETH fails, there is no algorithm solving the {\textsc{Hamiltonian Cycle}} problem in time $\bigoh^*((1+\sqrt{2}-\varepsilon)^{\operatorname{ctw}})$ for any $\varepsilon > 0$ even if a linear arrangement of the input graph of optimal width is provided with the input.
\end{theorem}

\begin{proof}
	Suppose such $\varepsilon > 0$ exists.
	We now show that there exists $\varepsilon' > 0$ such that for every $d \in \mathbb{N}$ the $d$-\textsc{SAT} problem can be solved in time $\bigoh((2-\varepsilon')^{n})$ where $n$ denotes the number of variables in the instance---this would contradict SETH.
	So let $d$ be fixed and let $I$ be an instance of $d$-SAT.
	The claims of this section imply that for every fixed pair $\gamma$ and $\beta$ of constants, we can in polynomial time construct a graph $G'$ of cutwidth $n / \gamma \cdot \beta + \bigoh_{\beta, \gamma, d}(1)$ such that $G$ admits a Hamiltonian cycle if and only if $I$ is satisfiable.
	Now we choose the values $\beta$ and $\gamma$ to show the desired claim.
	
	First, there exist a positive constant $C$ such that for sufficiently large $\beta$, we have
	\[
		\sum_{\substack{i_1 + i_2 = \beta \colon \\ i_1 \text{ is even, } i_1, i_2 > 0}} {\beta - 1 \choose i_1 - 1} 2^{i_1/2 - 1} \geq \sum_{i_1 \in [\beta]_0} {\beta \choose i_1} \sqrt{2}^{i_1}/C = (1+\sqrt{2})^\beta / C 
	\]
	where the inequality follows from elementary analysis.
	Thus to satisfy \eqref{eq:desired-ineq-hc-lb} we can use $\beta$ with 
	\[
		\beta = \frac{\gamma}{\log(1+\sqrt{2})} + \bigoh(1)
	\]
	and it remains to fix $\gamma$.
	Running the $\bigoh^*((1+\sqrt{2}-\varepsilon)^{\operatorname{ctw}})$ algorithm on the graph $G'$ of cutwidth $n/\gamma \cdot \beta + \bigoh_{\beta,d}(1)$ we obtain the running time of
	\begin{align*}
		&(1+\sqrt{2}-\varepsilon)^{n/\gamma \cdot \beta + \bigoh_{\beta,d}(1)} = (1+\sqrt{2}-\varepsilon)^{n/\gamma \cdot (\frac{\gamma}{\log(1+\sqrt{2})} + \bigoh(1)) + \bigoh_{\beta,d}(1)} = \\
		&\bigoh_{\beta,d,\varepsilon} ((1+\sqrt{2}-\varepsilon)^{n/\gamma \cdot (\frac{\gamma}{\log(1+\sqrt{2})} + \bigoh(1))}) = 
		\bigoh_{\beta,d,\varepsilon} (2^{\log (1+\sqrt{2}-\varepsilon) (n/\gamma \cdot (\frac{\gamma}{\log(1+\sqrt{2})} + \bigoh(1)))}) = \\
		&\bigoh_{\beta,d,\varepsilon} (2^{n(\frac{\log (1+\sqrt{2}-\varepsilon)}{\log(1+\sqrt{2})} + \frac{\log (1+\sqrt{2}-\varepsilon) \bigoh(1)}{\gamma})}) = \\
		&\bigoh_{\beta,d,\varepsilon} (2^{n(\frac{\log (1+\sqrt{2}-\varepsilon)}{\log(1+\sqrt{2})} + \frac{\log ((1+\sqrt{2}-\varepsilon)^{\bigoh(1)})}{\gamma})})
	\end{align*} 
	modulo factors polynomial in the size of the input instance of $d$-\textsc{SAT}.
	Thus choosing $\gamma$ large enough such that $\frac{\log (1+\sqrt{2}-\varepsilon)}{\log(1+\sqrt{2})} + \frac{\log ((1+\sqrt{2}-\varepsilon)^{\bigoh(1)})}{\gamma} < 1$ is satisfied yields an $\bigoh_d((2 - \varepsilon')^n)$ algorithm for $d$-SAT for every $d \in \mathbb{N}$ contradicting SETH.	
\end{proof}

\section{Maximum Induced Matching}\label{sec:mim}

An \emph{induced matching} in a graph $G=(V,E)$ is a set $M\subseteq E$ of edges such that $M$ is a matching, and $(V(M), M)$ is an induced subgraph of $G$, i.e., for any pair of distinct edges $e_1, e_2\in M$, it holds that $e_1\cap e_2 = \emptyset$ and there is no edge between an endpoint of $e_1$ and an endpoint of $e_2$ (which, in turn, can be written as $N(e) \cap V(M) = \emptyset$). 
We define the \Mimp problem as follows:
\probdef{\Mimp}{A graph $G=(V,E)$ and an integer $b\in\mathbb{N}$.}{Does $G$ contain an induced matching of cardinality $b$?}
We call an induced matching of $G$ \emph{maximum} if $G$ contains no induced matching of larger cardinality.

In this section, we show that, assuming SETH, the \Mimp problem cannot be solved in time $\ostar((3-\varepsilon)^{k})$ on graphs provided with linear arrangements of cutwidth $k$ for any positive value $\varepsilon$. 
This proves the tightness of the algorithm of Chaudhary and Zehavi~\cite{DBLP:conf/wg/ChaudharyZ23a} with running time $\ostar(3^k)$ for all parameters treewidth, pathwidth and cutwidth, where $k$ is the value of the parameter.

\begin{theorem}\label{thm:mim}
    Assuming SETH, there is no algorithm that solves the $\Mimp$ problem on graphs given with linear arrangements of cutwidth $k$ in time $\ostar\big((3-\varepsilon)^{k}\big)$ for any positive value $\varepsilon$.
\end{theorem}

By \cref{lb:theo:partition}, for every $\varepsilon > 0$, there exists an integer $q$ such that $q$-CSP-$3$ problem cannot be solved in time $\bigoh((3-\varepsilon)^n)$ where $n$ denotes the number of variables in the instance. 
In order to prove \cref{thm:mim}, we thus provide, for every integer value $q$, a reduction from the $q$-CSP-$3$ problem for $B=3$ to \Mimp such that the cutwidth of the arising graph is bounded by $n + \bigoh(1)$. 
So let $q$ be an arbitrary but fixed integer.
Let $I$ be an instance of the $q$-CSP-$3$ problem.
We now describe the instance $(G_I, b_I)$ of the \Mimp problem resulting from the reduction.  
Let $x_1,\dots, x_n$ denote the variables of $I$ and let $C_1,\dots, C_j$ be its constraints. For the sake of shortness, we denote $V_j = V_{C_j}$ and $R_j = R_{C_j}$ for each $j\in[m]$.

\subparagraph*{Construction.}

We start by describing the global structure of $G_I$ and the provide the details.
The graph $G_I$ consists of $n$ sequences of path-like structures which we will call \emph{paths} for simplicity.
We will denote the $i$.th path by $P_i$. 
Each path consists of $m\cdot(3n+1)$ copies of constant-size gadgets called \emph{path gadgets}. 
Each path gadget $X$ contains an designated \emph{entry vertex} $v_0(X)$ and a designated \emph{exit vertex} $v_2(X)$. The exit vertex of each path gadget is adjacent to the entry vertex of the following path gadget on the same path, giving these paths their path-like structure.
We divide each path $P_i$ into $3n+1$ segments of $m$ path gadgets each and we denote the $k$.th segment by $P_i^k$ for $k \in [3n+1]$. 
We denote by $X_i^{k,1},\dots X_i^{k, m}$ the path gadgets of $P_i^k$ in the order they occur on this path. 
We also add a designated vertex $g_i$ to each path $P_i$ so that it is adjacent to the exit vertex of the last path gadget on $P_i$, i.e., to the vertex $v_2(X_i^{3n+1,m})$.
For every $j \in [m]$ we also add $3n+1$ constraint gadgets corresponding to $C_j$ (we call them $\mathscr{C}_j^1, \dots \mathscr{C}_j^{3n+1}$). 
A constraint gadget $\mathscr{C}_j^k$ will only be adjacent to some path gadgets of form $X_i^{k, j}$. 

\begin{figure}[t]
    \centering
    \includegraphics[width=0.5\textwidth]{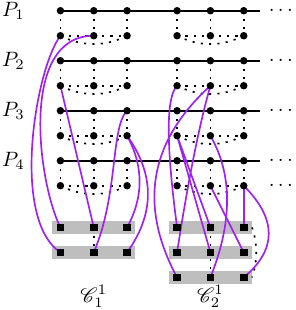}
    \caption{The first segment of the lower-bound construction for $n = 4$, $m = 2$ and the constraints $V_1 = (x_1, x_2, x_3)$, $R_1 = \{(1,0,2),(0,2,2)\}$ and $V_2 = (x_2, x_3, x_4)$, $R_2 = \{(0,0,2),(1,0,1),(1,1,2)\}$. Dotted lines sketch blocked edges. A dotted line between the gray boxed reflects that there is a blocked edge between any two vertices in different boxes. }  
    \label{fig:mim-lb}
\end{figure} 

The intersection of a maximum induced matching with a path gadget $X_i^{k,j}$ defines a so-called \emph{state} on this path gadget.
The construction will ensure that the states of this path gadget are in one-to-one correspondence with the $3$ assignments of the $i$.th variable $x_i$. 
The role of the paths $P_i$ is to force the existence of some value $k\in [3n+1]$, where for each path $P_i$, a solution defines the same state in all path gadgets on the $k$.th segment of this path.
A constraint gadget $\mathscr{C}$ corresponding to constraint $C$ then ensures that only states combinations are allowed that 
correspond to partial assignments satisfying this constraint. 
Hence, the existence of the segment $k$ with above properties guarantees the existence of an assignment satisfying all constraints $C_j$. To conclude the construction, it remains to define the above-mentioned gadgets.
For a gadget $X$ and a vertex $v$ of $X$, by $v(X)$ we refer to the vertex $v$ of this gadget---as we will use multiple copies of the same gadget, this will allow us to distinguish the vertices in different gadgets.

First, we define a \emph{blocked edge gadget} as a subgraph $\chi$ consisting of two special vertices $x$ and $y$ with an edge between them, and two vertex-disjoint paths on $4$ edges each, both having $x$ and $y$ as their endpoints and the internal vertices of these paths have no neighbors outside $\chi$. 
When we say we add a \emph{blocked edge} $\chi(u,v)$ between two vertices $u,v$ of $G_I$, that means we add a copy of $\chi$ and we identify $x$ with $u$ and $y$ with $v$ (see \cref{fig:mim-blocked-edge} for an illustration).

Now we define a \emph{path gadget}. 
A path gadget $X$ consists of a simple path on three vertices $v_0, v_1, v_2$, where $v_0$ is the \emph{entry vertex} of $X$ and $v_2$ is the \emph{exit vertex} of $X$. 
It also contains three vertices $u_0, u_1, u_2$ and add a blocked edge $\chi(v_c, u_c)$ for each $c\in[2]_0$. 
Finally, for each pair $c \neq c' \in [2]_0$, it contains a blocked edge $\chi(u_c,u_{c'})$.
For simplicity of the arguments in the following, we define the set 
\[
    U = \bigcup_{i \in [n], k \in [3n+1], j \in [m], p \in [2]_0} u_p(X_i^{j,k}).
\]

Finally, we describe a \emph{constraint gadget} $\mathscr{C}_j^k$ corresponding to the constraint $C_j$ for $j \in [m]$. 
Let $R_j = \{\alpha_1,\dots \alpha_{h_j}\}$. 
For every $r \in [h_j]$ the gadget $\mathscr{C}_j^k$ contains a set $A_r=\{a_r^1,\dots a_r^d\}$ of fresh vertices. 
For all $r\neq r' \in [h_j]$ and all $z,z'\in[d]$, it also contains the blocked edges $\chi(a_r^z, a_{r'}^{z'})$. 
For every $r\in[h_j]$ and $z\in[d]$, let $v_i = (V_j)_z$ be the $z$.th variable of the $j$.th constraint. Moreover, let $c = (\alpha_r)_z$. 
Then $\mathscr{C}_j^k$ also contains the edge $\{a_r^z, u_c(X_i^{k,j})\}$.

This concludes the construction of the graph $G_I$ (see \cref{fig:mim-lb} for an illustration). 
We also define the budget $b_I = 2 b_0 + (n+d)\cdot m(3n+1)$, where $b_0 = 6 \cdot nm(3n+1) + (3n+1)\sum_{j \in [m]} {h_j \choose 2} d^2$ is the number of blocked edges in $G_I$.
 
\begin{figure}[t]
    \centering
    \includegraphics[width=0.2\textwidth]{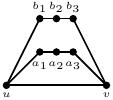}
    \caption{A blocked $\chi(u, v)$ edge between the vertices $u$ and $v$.}  
    \label{fig:mim-blocked-edge}
\end{figure} 

The useful property of the blocked edge gadget can now be stated as follows:
\begin{lemma}\label{mim::lem:blocked-edges}
    Let $G$ be a graph, let $u \neq v$ be two vertices of $G$, and let $G'$ be the graph arising from $G$ by adding a blocked edge $\chi(u,v)$ between two vertices $u \neq v$ of $G$. First, every induced matching of $G$ whose edges are incident with at most one of the vertices $u, v$ can be extended to an induced matching of $G'$ by adding two edges of $\chi(u, v)$ to it, and also in the arising matching, at most one of the vertices $u, v$ is incident with an edge of this matching. 
    Second, every maximum induced matching of $G'$ contains exactly two edges of $\chi(u, v)$ and the edges of $M$ are incident with at most one of the vertices $u,v$.  
\end{lemma}

\begin{proof}
    Let $u, a_1, a_2, a_3, v$ and $u, b_1, b_2, b_3, v$ be the two paths of length four between $u$ and $v$ in $\chi(u, v)$ (see \cref{fig:mim-blocked-edge} for an illustration). 

    First, let $M$ be an induced matching of $G$ that is incident with at most one of $u$ and $v$. 
    Without loss of generality, we assume that $v\notin V(M)$. 
    Then $M \cup \{\{a_2, a_3\}, \{b_2, b_3\}\}$ is also an induced matching arising from $M$ by adding two edges of $\chi(u, v)$ to it---this matching is induced because the neighbors of $a_2, a_3$ in $G'$ are $a_1$ and $v$ and we have $a_1, v \notin V(M)$ (and analogously for $b_2, b_3$).
    
    On the other hand, consider a maximum induced matching $M$ of $G'$.
    It is straight-forward to verify that $\chi(u,v)$ does not contain any induced matching of cardinality larger than two.
    Since $\chi(u, v)$ is an induced subgraph of $G'$, the induced matching $M$ contains at most two edges of $\chi(u, v)$.
    
    First, suppose that $M$ contains edges incident with both $u$ and $v$.
    If the edge incident with $u$ in $M$ would be not the same as the edge incident with $v$ in $M$, then the matching $M$ would not be induced as witnessed by the edge $\{u,v\}$ in $\chi(u,v)$.
    Therefore, in this case $M$ contains the edge $\{u, v\}$.
    The neighborhood of $\{u,v\}$ in $\chi(u, v)$ is $a_1, a_2, b_1, b_2$ so since $M$ is an induced matching, no further edge of $\chi(u, v)$ belongs to $M$.
    But then, for example, $(M \setminus \{\{u,v\}\}) \cup \{\{a_1, a_2\}, \{b_1, b_2\}\}$ is an induced matching of larger cardinality than $M$ contradicting the fact that $M$ is maximum.
    
    Thus, at most one of the vertices $u$ and $v$ is incident with the edges in $M$.
    Without loss of generality, we may assume that $v \notin V(M)$ holds.
    Let $t$ be the number of edges of $\chi(u,v)$ that belong to $M$ and $M^*$ denote the matching arising from $M$ by removing all edges from $\chi(u, v)$ from it.
    Then the set $M^* \cup \{\{a_2, a_3\}, \{b_2, b_3\}\}$ is also an induced matching of $G'$ of size $|M| - t + 2$.
    Hence, if $t \leq 1$ would hold, then $M$ would be not a maximum induced matching.
    So we have $t = 2$ and this concludes the proof of the second claim.
\end{proof}

\begin{lemma}\label{mim::lem:sat-to-sol}
    If $I$ is satisfiable, then $G_I$ admits an induced matching of size $b_I$.
\end{lemma}

\begin{proof}
    Let $\tau \colon \{x_1, \dots, x_n\} \to [B]_0$ be an assignment satisfying $I$. 
    We now construct an induced matching $M$ of cardinality $b_I$ in $G_I$.
    For this, we start with $M$ being an empty set.
    Next for every $i \in [n]$ we proceed as follows:
    \begin{itemize}
     \item If $\tau(x_i) = 0$, then for every $k \in [3n+1]$ and every $j \in [m]$, we add the edge $\{v_1(X_i^{k,j}), v_2(X_i^{k,j})\}$ to $M$. 
     \item If $\tau(x_i) = 2$, then for every $k \in [3n+1]$ and every $j \in [m]$, we add the edge $\{v_0(X_i^{k,j}), v_1(X_i^{k,j})\}$ to $M$.
     \item If $\tau(x_i) = 1$, the we add the following edges to $M$.
     First for every $k \in [3n+1]$ and every $j \in [m-1]$, we add the edge $\{v_2(X_i^{k,j}), v_0(X_i^{k,j+1})\}$ to $M$.
     Also for every $k \in [3n]$ we add the edge $\{v_2(X_i^{k,m}), v_0(X_i^{k+1,1})\}$ to $M$.
     And finally, we add the edge $\{v_2(X_i^{3n+1,m}), g_i\}$.
     \end{itemize}
     Observe that for the currently set $M$ and $i \in [n]$, if $\tau(x_i) = a$ for some $a \in [2]_0$, then for any $k \in [3n+1]$ and $j \in [m]$, neither the vertex $u_a(X_i^{k,j})$ nor any of its neighbors is incident with an edge of $M$.
     Also $M$ constructed so far is clearly an induced matching of $G_I$ of cardinality $nm(3n+1)$.
     
     Now for every $j \in [m]$ we proceed as follows.
     Let $R_j = \{\alpha_1, \dots, \alpha_{h_j}\}$.
     Since $\tau$ is a satisfying assignment of $I$, there exists an in index $r \in [h_j]$ so that $\tau(V_j) = \alpha_r$.
     For every $z = 1, \dots, d$ we proceed as follows.
     Let $i = (V_j)_p$.
     By construction, for every $k \in [3n+1]$, there exists the edge $\{u_{\tau(x_i)} (X_i^{j,k}), a_r^z(\mathcal{C}_j^k)\}$---we add this edge to $M$.
     
     For the set $M$ constructed so far, it is straight-forward to verify that $M$ is an induced matching of $G_I$ and it has cardinality $m(3n+1)n + m(3n+1)d$.
     Furthermore, for every blocked edge in $G_I$, at most one end-vertex of this blocked edge is incident with an edge of $M$.
     Recall that $b_0$ is, by definition, the number of blocked edges in $G_I$.
     Now by applying \cref{mim::lem:blocked-edges} to every blocked edge of $G_I$ one after another, we can add $2b_0$ edges to $M$ so that the arising edge set is an induced matching of $G_I$ of cardinality $m(3n+1)n + m(3n+1)d + 2b_0 = b_I$.
\end{proof}

\begin{lemma}\label{mim::lem:form-of-solution}
    The size of a maximum induced matching of $G_I$ is at most $b_I$. Moreover, if $G_I$ contains an induced matching $M$ of size $b_I$, then all of the following properties apply:
    \begin{itemize}
        \item For every $i \in [n]$, 
        let $P_i^0$ be the path in $G_I$ induced by the vertices 
        \[
         \{g_i\} \cup \bigcup_{j \in [m], k \in [3n+1]} \left\{v_0(X_i^{k,j}), v_1(X_i^{k,j}), v_2(X_i^{k,j})\right\}.
        \] 
        Then $M$ contains exactly $m(3n+1)$ edges of $P_i^0$.
        \item For every $k \in [3n+1]$ and $j \in [m]$, there exists a value $r\in[h_j]$, such that first, the set $M$ contains all edges between $A_r(\mathscr{C}_j^k)$ and $U$ present in $G_I$ and second, for every $r' \neq r \in [h_j]$, the set $M$ contains no edges between $A_{r'}(\mathscr{C}_j^k)$ and $U$ present in $G_I$.
        \item The matching $M$ contains exactly two edges in each blocked edge gadget $\chi$ of $G_I$.
        \item The matching $M$ does not contain any edges of $G_I$ other than the ones described in the previous items.
    \end{itemize}
\end{lemma}

\begin{proof}
    Let $M$ be a maximum induced matching of $G_I$. 
    First, we observe that the edges of $G_i$ can be partitioned into the sets
    \begin{equation}\label{eq:edge-partition-tp}
        \{E(P^0_i)\}_{i\in[n]} \cup \{E(\chi)\}_{\chi \text{ is a blocked edge in } G_I} \cup \{E_j^k\}_{j\in[m],k\in[3n+1]}
    \end{equation}
    where 
    $E_j^k$ denotes the set of all edges incident with the vertices of $\mathscr{C}_j^k$.
    
    For every $i\in[n]$, the path $P_i^0$ is a simple path consisting of $3m(3n+1)$ edges. 
    Observe that if two edges of this path belong to $M$, then there are at least two edges of this path between them---otherwise, $M$ would not be an induced matching.
    Therefore, of any three consecutive edges of the path $P_i^0$, the set $M$ contains at most one.
    And therefore, $M$ contains at most $m(3n+1)$ edges of $P_i^0$.
    
    Now recall that by \cref{mim::lem:blocked-edges}, every blocked edge gadget of $G_I$ contains at most two edges of $M$, and there are $b_0$ blocked edges in $G_I$.
    Then, $M$ contains at most $2b_0$ edges of blocked edge gadgets.

    Further, consider $k \in [3n+1]$ and $j \in [m]$. 
    Recall that for any $r \neq r' \in [h_j]$ and any $z, z' \in [d]$, the vertices $a_r^z$ and $a_{r'}^{z'}$ are end-points of a blocked edge so by \cref{mim::lem:blocked-edges}, at most one of them is incident with an edge of $M$.
    So there exists at most one index $r \in [h_j]$ such that the vertices of $A_r$ are incident with edges in $M$.
    Finally, recall that by construction, every vertex in $A_r$ has exactly one neighbor in $U$ and we have $|A_r| = d$. 
    Recall that \eqref{eq:edge-partition-tp} is a partition of the edge set of $G_I$. 
    Therefore, first $M$ has the cardinality of at most:
    \begin{align*}
        (m(3n+1)) \cdot n + 2 b_0 + m(3n+1) \cdot d = b_I.
    \end{align*}
    And furthermore, if an induced matching of $G_I$ has the cardinality of exactly $b_I$, then for every set in this partition, the intersection of the induced matching with this set is exactly the upper bound we argued above, i.e., the second claim of the lemma holds as well. 
\end{proof}

\begin{lemma}\label{mim::lem:sol-to-sat}
    If $G_I$ contains an induced matching $M$ of size $b_I$, then $I$ is satisfiable.
\end{lemma}

\begin{proof}
    For every $i \in [n]$ we call a vertex $v$ of $P_i^0$ \emph{$M$-free} if 
    $v$ is not incident with any edge of $M \cap E(P_i^0)$.
    Let $X$ be a path gadget of $G_I$.
    Note that at least one of the vertices $v_0(X), v_1(X), v_2(X)$ is $M$-free---otherwise, $G[M]$ would contain a path on two edges contradicting the definition of an induced matching.
    We define the value $s_M(X)$ as follows.
    If at least two vertices of $v_0(X), v_1(X), v_2(X)$ are $M$-free, we set $s_M(X)=\perp$.
    Otherwise, we let $s_M(X)$ be the unique value $c \in [2]_0$ such that $v_c(X)$ is $M$-free.
    We call $s_M(X)$ the \emph{state} of $X$ defined by $M$.
    
    Let $i \in [n]$ be arbitrary but fixed.
    By \cref{mim::lem:form-of-solution}, $M$ is a maximum induced matching of $G_I$, and $M$ contains $m(3n+1)$ edges of $P_i^0$. 
    The set $M$ is an induced matching and therefore, of any three consecutive edges of the path $P_i^0$, $M$ contains at most one.
    We have $|E(P_i^0)| = 3m(3n+1)$ and $|M \cap E(P_i^0)| = m(3n+1)$ so if we (uniquely) partition the path $P_i^0$ into the sets consisting of three consecutive edges each, and denote this partition by $\mQ$, then each of the sets in $\mQ$ contains exactly one edge of $M$.
    
    Now we make some observations to show that the sequence of states of the path gadgets on the path $P_i$ is very restricted. 
    First, let $X$ be a path gadget on $P_i$ and let $w$ denote the vertex following $v_2(X)$ on $P_i^0$.
    We define the set $E^X = \{\{v_0(X), v_1(X)\}, \{v_1(X), v_2(X)\}, \{v_2(X), w\}\}$.
    Since this set belongs to $\mQ$, exactly one of the edges in $E^X$ belongs to $M$.
    \begin{itemize}
        \item If we have $\{v_0(X), v_1(X)\} = E^X \cap M$, then by definition we have $s_M(X) = 2$.
        \item If we have $\{v_1(X), v_2(X)\} = E^X \cap M$, then by definition we have $s_M(X) = 0$.
        \item If we have $\{v_2(X), w\} = E^X \cap M$, then 
        we have $s_M(X) \in \{1, \perp\}$: this is because the vertex $v_1(X)$ is $M$-free by definition of an induced matching while the vertex $v_0(X)$ may be incident with an edge of $M$ whose other end-vertex belongs to the gadget preceding $X$ on $P_i$.
    \end{itemize}
    Also note that the above implications are actually if-and-only-if relations, as $\{2\}, \{0\}, \{1, \perp\}$ partition the set of possible states.
    
    Now suppose that $s_M(X) \in \{\perp, 1\}$ holds.
    We first claim that if there exists a gadget $Y$ directly following $X$ on $P_i$, then we have $s_M(Y) = 1$.
    If $Y$ would have the state $2$ or $0$, then by above observations $G[V(M)]$ would contain a path on two ($v_2(X), v_0(Y), v_1(Y)$) or three ($v_2(X), v_0(Y), v_1(Y), v_2(Y)$) edges, respectively, contradicting the definition of an induced matching.
    Thus it holds that $s_M(Y) \in \{\perp, 1\}$.
    Let now $w'$ be the vertex following $v_2(Y)$ on $P_i^0$.
    Above we have shown that then $E^X \cap M = \{\{v_2(X), w = v_0(X)\}\}$ and $E^Y \cap M = \{\{v_2(Y), w'\}\}$ hold. 
    Thus, none of the vertices $v_0(Y)$ and $v_2(Y)$ is $M$-free 
    and we indeed have $s_M(Y) = 1$. 
    
    Next, suppose that $s_M(X) = 0$ holds.
    In this case we claim that if there exists a gadget $Y$ directly following $X$ on $P_i$, then we have $s_M(Y) \neq 2$.
    So suppose that we have $s_M(Y) = 2$.
    Then it holds that $\{v_1(X), v_2(X)\} = E^X \cap M$ and $\{w = v_0(Y), v_1(Y)\} = E^X \cap M$ and therefore $G[V(M)]$ contains a path $v_1(X), v_2(X), v_0(Y), v_1(Y)$ on three edges contradicting the definition of an induced matching.
    
    So altogether, this implies that along the path $P_i$ we first see the path gadgets of state $2$, then the path gadgets of state $0$, then at most one path gadget of state $\perp$, and then the path gadgets of state $1$ where each of these subsequences may, in general, be empty.
    More formally, this can be stated as follows: on the path $P_i$, there exist at most three gadgets $X$ satisfying at least one of the following properties:
    \begin{itemize}
        \item $s_M(X) = \perp$,
        \item or there exists a path gadget $Y$ such that $Y$ directly follows $X$ on $P_i$ and we have $s_M(X) \neq s_M(Y)$.
    \end{itemize}
    Hence, on the paths $P_1, \dots, P_n$ at most $3n$ such gadgets exist.
    Therefore, there exists an index $k \in [3n+1]$ such that for every $i \in [n]$, we have
    \[
        s_M(X_i^{k,1}) = s_M(X_i^{k,2}) = \dots = s_M(X_i^{k,m}) \neq \perp.
    \]

    Now we are finally ready to define the assignment $\tau \colon \{x_1, \dots, x_n\} \to [2]_0$ such that $\tau(v_i)=s_M(X_i^{k,1})$ holds for every  $i\in[n]$. 
    It remains to show that $\tau$ satisfies $I$.
    For this let $j \in [m]$ be arbitrary.
    Let $V_j = (x_{q_1}, \dots, x_{q_d})$ be the variables occurring in the constraint $C_j$ and let $R_j = \{\alpha_1, \dots, \alpha_{h_j}\}$ be the set of assignments satisfying $C_j$.
    By \cref{mim::lem:form-of-solution}, there exists a value $r \in [h_j]$ such that $M$ contains all edges between $A_r(\mathcal{C}_j^k)$ und $U$ present in $G_I$.
    Recall that for every $z \in [d]$, the vertex $a_r^z(\mathcal{C}_j^k)$ is adjacent only to the vertex $u_{c_z}(X_{q_z}^{k,j})$ from the set $U$ where $(c_1, \dots, c_d) = \alpha_r$ is the $r$.th assignment satisfying $C_j$.
    I.e., we have $\{a_r^z(\mathcal{C}_j^k), u_{c_z}(X_{q_z}^{k,j})\} \in M$ for every $z \in [d]$.
    It remains to show that we have $\tau(x_{q_z}) = c_z$ for every $z \in [d]$.
    
    So suppose we have 
    \[
        c_z \neq \tau(x_{q_z}) = s_M(X_{q_z}^{k,1}) = s_M(X_{q_z}^{k,j}).
    \]
    First, the choice of $k$ implies that $s_M(X_{q_z}^{k,j}) \neq \perp$ holds, i.e., we have $s_M(X_{q_z}^{k,j}) \in [2]_0$.
    Then $\perp \neq s_M(X_{q_z}^{k,j}) \neq c_z$ implies that the vertex $v_{c_z}(X_{q_z}^{k,j})$ is not $M$-free.
    By definition of being not $M$-free, one of the edges, say $e$, of $P_i^0$ incident with $v_{c_z}(X_{q_z}^{k,j})$ belongs to $M$. 
    Recall that the edge $v_{c_z}(X_{q_z}^{k,j}) u_{c_z}(X_{q_z}^{k,j})$ belongs to $G_I$ (as it is contained in the blocked edge between these vertices).
    But then the subgraph induced on the end-vertices of the edges $e$ and $a_r^z(\mathcal{C}_j^k), u_{c_z}(X_{q_z}^{k,j})$, both in $M$, contains a path on three edges contradicting the fact that $M$ is an induced matching.
    Hence, we have 
    \[
        \tau(x_{q_1}, \dots, x_{q_d}) = (c_1, \dots, c_d) = \alpha_r \in R_j,
    \]
    i.e., $\tau$ satisfies the constraint $C_j$ as desired.
\end{proof}

\begin{lemma}\label{mim::lem:cutwidth}
    The cutwidth of $G_I$ is upper-bounded by $n+\bigoh_d(1)$. 
    Furthermore, the graph $G_I$ and a linear arrangement of $G_I$ of cutwidth at most $n+\bigoh_d(1)$ can be computed from $I$ in polynomial time.
\end{lemma}

\begin{proof}
    We describe a linear arrangement $\ell$ of $G$, and then we bound its cutwidth. 
    To construct $\ell$ we start with an empty linear arrangement and proceed as follows by always adding the vertices to the very right.
    First we iterate over $k = 1, \dots, 3n+1$.
    For a fixed value of $k$ we iterate over $j = 1, \dots, m$ as follows.
    For every $i = 1, \dots, n$, we place all vertices of $X_i^{j,k}$ consecutively.
    After processing all values of $i$, we add all vertices of $\mathcal{C}_j^k \setminus U$.
    In the very end, we add the vertices $g_1, \dots, g_n$.
    Observe that every vertex of $G_I$ was added exactly once to $\ell$---in particular, every blocked edge gadget belongs to exactly one path or constraint gadget.
    Clearly, $\ell$ and $G$ can be constructed from $I$ in polynomial time.

    For $i \in [n]$ and a non-last path gadget $X$ on $P_i$, by $f(X)$ we denote the path gadget directly following $X$ on $P_i$.
    Now we consider an arbitrary cut and bound the number of edges crossing this cut depending on the ``type'' of the edge:
    \begin{itemize}
     \item For every $i \in [n]$, at most one of the edges 
     \begin{align*}
        &\left\{\{v_2(X), v_0(f(X))\} \mid \exists (k,j) \in ([3n+1] \times [m]) \setminus \{(3n+1, m)\} \colon X = X_i^{j,k}\right\} \cup \\
        &\left\{\{v_2(X_i^{3n+1, m}), g_i\right\}
     \end{align*}
        crosses the cut.
        Therefore, totally over all $i \in [n]$, at most $n$ such edges cross the cut.
    \item There exists at most one path gadget $X$ such that the edges with both end-points in $X$ cross this cut.
    Since every path gadget contains $\bigoh(1)$ vertices, only $\bigoh(1)$ such edges cross this cut.
    \item There exists at most one pair $(k, j) \in [3n+1] \times [m]$ such that the edges with at least one end-vertex in $\mathcal{C}_j^k \setminus U$ cross this cut.
    Recall that we have $|\mathcal{C}_j^k| \in \bigoh_d(1)$ by construction and every vertex in $\mathcal{C}_j^k \setminus U$ has at most one neighbor outside $\mathcal{C}_j^k \setminus U$. 
    Therefore, at most $\bigoh_d(1)$ such edges cross this cut.
    \end{itemize}
    These cases cover all edges of $G_I$ and therefore, the cutwidth of $\ell$ is upper-bounded by $n+\bigoh_d(1)$.
\end{proof}

\begin{proof}[Proof of \cref{thm:mim}]
Assume that there exists an algorithm that solves the \Mimp problem in time $\ostar((3-\varepsilon)^{\ctw})$ for some $\varepsilon > 0$. 
Let $d$ be arbitrary.
We show that $d$-CSP-$3$ can be solved in time $\ostar((3-\varepsilon)^n)$ where $n$ denotes the number of variables contradicting SETH.
Given an instance $I$ of the $d$-CSP-$3$, 
we first construct an equivalent instance $(G_I, b_I)$ of the \Mimp problem together with a linear arrangement of $G_I$ of cutwidth at most $n+\bigoh_d(1)$ as described above. 
Then we run the above algorithm on $(G_I, b_I)$ in time $\ostar((3-\varepsilon)^{n+\bigoh_d(1)}) = \ostar((3-\varepsilon)^{n})$. 
The algorithm is correct as the instances are equivalent.
\end{proof}

\section{Max Cut}\label{sec:maxcut}

It is folklore that the \textsc{Max Cut} problem can be solved in time $\ostar(2^t)$ if the input graph is provided with a tree 
decomposition of pathwidth $p$ (see e.g., \cite{DBLP:books/sp/CyganFKLMPPS15}).
Lokshtanov et al.~\cite{DBLP:journals/talg/LokshtanovMS18} have shown that unless SETH fails, the problem cannot be solved in time $\ostar((2 - \varepsilon)^p)$ for any $\varepsilon > 0$ even if the optimal path decomposition is provided with the input.
It is well-known that for any graph, the inequality $\operatorname{pw} \leq \operatorname{ctw} + \mathcal{O}(1)$ holds (see e.g., \cite{DBLP:journals/tcs/Bodlaender98}).
In this section, we strengthen the result by Lokshtanov et al.\ by showing that, unless SETH fails, the problem cannot be solved in time $\ostar((2 - \varepsilon)^k)$ for any $\varepsilon > 0$ if the input graph is provided with a linear arrangement of cutwidth $k$.
This implies that the $\mathcal{O}^*(2^p)$ algorithm solving \textsc{Max Cut} for graphs of pathwidth $p$ is optimal for cutwidth as well.

Before providing the lower bound, let us briefly explain why the lower-bound construction by Lokshtanov et al.~\cite{DBLP:journals/talg/LokshtanovMS18} does not work as the desired lower bound for cutwidth.
In their reduction, they first define a weighted graph and after that, replace every edge $uv$ of weight $w$ by $w$ paths of length 3 between the vertices $u$ and $v$.
In their construction some weights are as large as $3n$ (where $n$ denotes the number of variables in the original instance of the SAT problem) and therefore, the arising unweighted graph has the cutwidth of at least $3n$ being already too large for the desired no $\mathcal{O}^*(2^{\operatorname{ctw} - \varepsilon})$ lower bound.

We will follow their idea but will have to carry out multiple adaptations in order to bound the cutwidth of the arising graph.
For the sake of completeness, we provide the whole construction and emphasize the points that divert from the construction by Lokshtanov el al.~\cite{DBLP:journals/talg/LokshtanovMS18}. 	

\begin{theorem}
    Assuming SETH, there is no algorithm that solves the {\textsc{Max Cut}} problem on graphs given with linear arrangements of cutwidth $k$ in time $\ostar\big((2-\varepsilon)^{k}\big)$ for any positive value $\varepsilon$.
\end{theorem}

\begin{proof}
    For a partition $(V_0, V_1)$ of the vertex set of a graph $G$, we say that an edge $e$ of $G$ is \emph{crossing this partition} (or simply \emph{is crossing}) if $e$ has one end-vertex in $V_0$ and the other in~$V_1$.
    
    Let $d \geq 3$ be an arbitrary but fixed integer treated as a constant in the following.
    We will show that if we can solve \textsc{Max Cut} in time $\mathcal{O}^*((2-\varepsilon)^{\operatorname{ctw}})$ for some $\varepsilon > 0$, then we can also solve $d$-\textsc{SAT} in time $\mathcal{O}^*((2-\varepsilon)^n)$ where $n$ denotes the number of variables in the instance.
    Since $d$ is chosen arbitrarily and independently of $\varepsilon$, this would contradict SETH.
    For this, let $I$ be an arbitrary instance of $d$-\textsc{SAT}.
    Let $U = \{v_1, \dots, v_n\}$ and $C = \{C_1, \dots, C_m\}$ denote the sets of variables and clauses of $I$, respectively.
    For every $j \in [m]$ let $b^j_1, \dots, b^j_{t(j)}$ denote the set of variables occurring in the clause $C_j$, in particular, we then have $0 < t(j) \leq d$.
    We may assume that no variable occurs both positively and negatively in $C_j$ as otherwise, this clause can be discarded without changing the satisfiability of the instance.
    In the remainder of the proof we construct an equivalent instance of \textsc{Max Cut} of cutwidth upper-bounded by $n + \mathcal{O}_d(1)$ together with a linear arrangement of this cutwidth.

    To simplify the arguments, we will first provide an instance of \textsc{Weighted Max Cut} equivalent to $I$. After that we will use a simple replacement used by Lokshtanov et al.~\cite{DBLP:journals/talg/LokshtanovMS18} to obtain an equivalent instance of (unweighted) \textsc{Max Cut} of cutwidth bounded by $n + \mathcal{O}_d(1)$.

    \begin{figure}[t]
        \centering
        \includegraphics[width=0.8\textwidth]{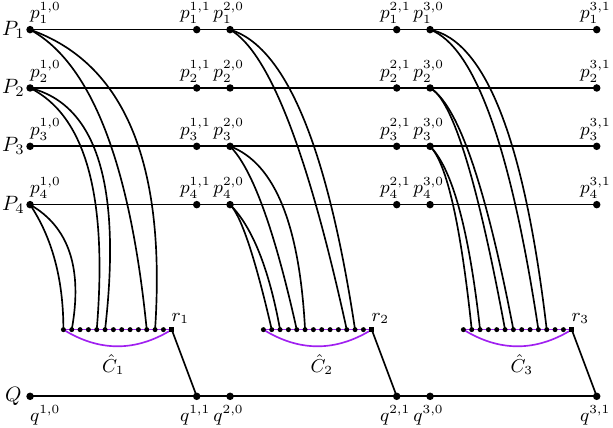}
        \caption{An example of the weighted graph for $d = 3$, $n = 4$, and clauses $\overline{v_1} \lor v_2 \lor v_4$, $\overline{v_1} \lor v_3 \lor \overline{v_4}$, and $v_1 \lor \overline{v_2} \lor \overline{v_3}$. The black edges have unit weight while a purple edge has the weight of 3 because every clause contains $d = 3$ variables.}
        \label{fig:maxcut-lb}
    \end{figure}
    
    So we now construct a weighted graph $(G, \omega \colon E(G) \to \bN^+)$ and refer to \cref{fig:maxcut-lb} for illustration. 
    Unless otherwise specified, the edges have unit weights---only the edges of the cycles $\hat C_1, \dots, \hat C_m$ defined later will possibly have non-unit weights.
    First, we create a so-called \emph{synchronization path} $Q = q^{1,0} q^{1,1} q^{2,0} q^{2,1} \dots q^{m,0} q^{m,1}$ on $2m$ vertices. 
    The target value of the number of crossing edges, also called the \emph{budget}, defined later will ensure that all edges of the synchronization path are crossing, i.e., every second vertex belongs to the same side of the partition.  
    Thus, the edges with one end-point on the synchronization path will later ensure that some vertices are assigned to a fixed side of the partition.
    Next, for every $i \in [n]$, we create a path $P_i = p_i^{1,0} p_i^{1,1} p_i^{2,0} p_i^{2,1} \dots p_i^{m,0} p_i^{m,1}$ on $2m$ fresh vertices.
    The budget will, again, ensure that all edges of $P_i$ are crossing and therefore, there are only two possible ways how the vertices of $P_i$ can be assigned to the sides of the partition---these two options will then correspond to two truth-value assignments of the variable $v_i$.
    Next for every $j \in [m]$, we create an odd cycle 
    \[  
        \hat C_j = c_j^{1,1}, c_j^{1,2}, c_j^{1,3}, c_j^{1,4}, c_j^{2,1}, c_j^{2,2}, c_j^{2,3}, c_j^{2,4}, \dots, c_j^{t(j),1}, c_j^{t(j),2}, c_j^{t(j),3}, c_j^{t(j),4}, r_j
    \]
    on $4t(j)+1$ fresh vertices such that every edge of this cycle has weight $3t(j)$. 
    For every $j \in [m]$, we also add an edge $q^{j,1} r_j$. 
    Finally, for every $j \in [m]$ and every $z \in [t(j)]$, we proceed as follows. 
    Let $i = i(j, z) \in [n]$ be such that $b^j_z = v_i$, i.e., $v_i$ is the $z$-th variable occurring in the clause $C_j$.
    We define the vertex $s(j, z) = p_i^{j, 0}$.
    If $v_i$ occurs positively in $C_j$, we add the edges $c_j^{z, 1} s(j, z)$ and $c_j^{z, 2} s(j, z)$ and we define the value $f(j,z) = 1$. 
    Otherwise, i.e., if $v_i$ occurs negatively in $C_j$, we add the edges $c_j^{z, 2} s(j, z)$ and $c_j^{z, 3} s(j, z)$ and we define the value $f(j,z) = 2$. 
    This concludes the construction of the weighted graph $(G, \omega)$.
    In only remains to define the total weight of the crossing edges in the sought solution.

    Before doing so, let us emphasize how much $(G, \omega)$ differs from the construction by Lokshtanov et al.~\cite{DBLP:journals/talg/LokshtanovMS18}.
    First, the paths $P_1, \dots, P_n$: Lokshtanov et al. use a single vertex for every $i \in [n]$ and then make it adjacent to, in worst-case, at least $2m$ vertices of the clause gadgets---leading to too large cutwidth.
    Similarly, they use a single synchronization-vertex instead of the path $Q$, and we had to eliminate it for the same reason.
    Further, in their reduction, for every $j \in [m]$, all edges of the odd cycle $\hat{C}_j$ have the weight of $3n$: to make the graph unweighted they replace each such edge by $3n$ paths of length $3$ with the same end-vertices---this leads to the cutwidth of at least $3n$ which is too large for our purposes. 
    
    To simplify the following arguments, we will now partition the edges of $G$ into the so-called \emph{groups} and define the so-called \emph{budgets} corresponding to each of the groups. 
    First, the set $E_1$ contains all edges on the synchronization path $Q$ as well as the edges with one end-vertex on $Q$ and the other in $\{r_1, \dots, r_m$\}.
    We define the budget $B_1 = 3m-1$.
    Second, the set $E_2$ contains all edges of the paths $P_1, \dots, P_n$.
    We define the budget $B_2 = n \cdot (2m-1)$. 
    Finally, for every $j \in [m]$, the set $E^j_3$ contains 
    all edges incident with $V(\hat C_j)$ apart from those in $E_1$.
    We define $B^j_3 = 12t(j)^2 + t(j) + 1$.
    With those we can define $E_3 = \bigcup_{j \in [m]} E_3^j$ and $B_3 = \sum_{j \in [m]} B_3^j$.
    Finally, we define $B = B_1 + B_2 + B_3$.
    Note that $E_1, E_2, E_3^1, E_3^2, \dots, E_3^m$ partition the edge set of $G$.
    Now we prove some useful properties of the partitions of the vertex set of $G$.

    \begin{observation}\label{obs:crossing-path}
        Let $x \in \mathbb{N}^+$, let $P = w^{1,0} w^{1,1} w^{2,0} w^{2,1} \dots w^{x,0} w^{x,1}$ be a path on $2x$ vertices in $G$, and let $(V_0, V_1)$ be a partition of $V(G)$.
        Then exactly $2x-1$ edges of $P$ are crossing $(V_0, V_1)$ if and only if there exists an index $y \in \{0, 1\}$ with $V_0 \cap V(P) = \{w^{\ell, y} \mid \ell \in [x]\}$.
    \end{observation}

    \begin{proof}
        The only way to make all $2x-1$ edges of the path $P$ crossing is to enforce that the vertices along this path alternate between $V_0$ and $V_1$, i.e., either to have $V_0 \cap V(Q) = \{w^{1,0}, w^{2,0}, \dots, w^{m,0}\}$ and $V_1 \cap V(Q) = \{w^{1,1}, w^{2,1}, \dots, w^{m,1}\}$, or to have $V_0 \cap V(Q) = \{w^{1,1}, w^{2,1}, \dots, w^{m,1}\}$ and $V_1 \cap V(Q) = \{w^{1,0}, w^{2,0}, \dots, w^{m,0}\}$.
        On the other hand, if the vertices alternate like this, then all edges of $P$ are crossing.
    \end{proof}    
    
    \begin{claim}\label{claim:e-1-edges}
        For any partition $(V_0, V_1)$ of the vertex set of $G$, at most $B_1$ edges of $E_1$ are crossing. 
        Moreover, if exactly $B_1$ edges of $E_1$ are crossing, then there exists an index $y \in \{0, 1\}$ with first, $V_0 \cap V(Q) = \{q^{\ell, y} \mid \ell \in [m]\}$, and second, for all $j \in [m]$ we have $r_j \in V_y$.
    \end{claim}

    \begin{proof}
        First, $|E_1| = B_1$ so the first part of the claim is obvious.
        Second, if all $B_1$ edges of $E_1$ are crossing, in particular, all edges of $Q$ are crossing so by \cref{obs:crossing-path}, the second part of the claim follows. 
        Finally, all edges of $E_1$ being crossing also implies that for every $j \in [m]$ we have $r_j \in V_y$ if and only if $q^{j,1} \in V_{1-y}$.
        This, in turn, is equivalent to $q^{j,0} \in V_{y}$ concluding the proof.
    \end{proof}

    \begin{claim}\label{claim:e-2-edges}
         For any partition $(V_0, V_1)$ of the vertex set of $G$, at most $B_2$ edges of $E_2$ are crossing. 
         Moreover, if exactly $B_2$ edges in $E_2$ are crossing, then for every $i \in [n]$ there exists an index $y(i) \in \{0, 1\}$ with $V_0 \cap V(P_i) = \{p_i^{\ell, y(i)} \mid \ell \in [m]\}$.
    \end{claim}

    \begin{proof}
        The set $E_2$ consists of $n$ vertex-disjoint paths consisting of $2m-1$ edges each and therefore, at most all $n \cdot (2m-1) = B_2$ edges in $E_2$ are crossing so the first part trivially holds.
        Further, if all $B_2$ edges in $E_2$ are crossing, in particular, for every $i \in [n]$, all edges of the path $P_i$ are crossing.       
        \cref{obs:crossing-path} implies that there exists an index $y(i) \in \{0,1\}$ with $V_0 \cap V(P_i) = \{p_i^{\ell, y(i)} \mid \ell \in [m]\}$
    \end{proof}

    \begin{claim}\label{claim:e-3-edges}
         Let $j \in [m]$. For any partition $(V_0, V_1)$ of the vertex set of $G$, the total weight of the crossing edges in $E^j_3$ is at most $B^j_3$. 
         Moreover, if the total weight of the crossing edges in $E^j_3$ is exactly $B^j_3$, then the following hold.
         First, there exists a unique not crossing edge $e$ on the cycle $\hat C_j$.
         Furthermore, there exists an index $z^*(j) \in [t(j)]$ such that $c_j^{z^*(j), f(j, z^*(j))}$ and $c_j^{z^*(j), f(j, z^*(j))+1}$ are the end-vertices of $e$.
         Finally, both edges $c_j^{z^*(j), f(j,z^*(j))} s(j,z^*(j))$ and $c_j^{z^*(j), f(j,z^*(j)+1} s(j,z^*(j))$ 
         are crossing.
    \end{claim}
    Before proving the claim, let us briefly explaining what it means less formally.
    The set $E_3^j$ consists of the edges of the odd cycle $\hat C_j$ together with all edges with one end-point on $\hat C_j$ and the other on $V(P_1) \cup \dots V(P_n)$.
    The claim states that if the total weight of the crossing edges in $E_3^j$ is $B_3^j$, then on the cycle $\hat C_j$ only one edge is not crossing.
    And the end-points of this edge (therefore, lying on the same side of the partition) form a triangle with a vertex in $V(P_1) \cup \dots V(P_n)$ on a different side of this partition---thus making the two edges of this triangle crossing.
    \begin{proof}
        Recall that $\hat C_j$ is an odd cycle, in particular, it is not bipartite so at least one edge of $\hat C_j$ is not crossing.
        First, consider a partition $(V_0, V_1)$ such that at least two edges of $\hat C_j$ are not crossing.
        In this case at most $4t(j) - 1$ edges of $\hat C_j$ and at most all of the $2t(j)$ edges between $\hat C_j$ and $V(P_1) \cup \dots \cup V(P_n)$ are crossing.
        Recall that each of the former edges has the weight of $3t(j)$ and each of the latter edges has a unit weight. 
        Then the total weight of crossing edges in $E^j_3$ is at most
        \[
            (4t(j) - 1) \cdot 3t(j) + 2t(j) = 12t(j)^2 - t(j) < 12t(j)^2 + t(j) + 1 = B_3^j.
        \]

        So in the following we may assume that for a partition $(V_0, V_1)$, exactly one of the edges of $\hat C_j$ is not crossing, i.e., exactly $4t(j)$ of these edges are crossing.
        Thus, the total weight of the edges in $\hat C_j$ that are crossing is exactly $4t(j) \cdot 3t(j) = 12t(j)^2$.
        Now consider the matching $M$ consisting of the $t(j)$ edges 
        \[
            c_j^{1, f(j,1)} c_j^{1, f(j,1)+1}, c_j^{2, f(j,2)} c_j^{2, f(j,2)+1}, \dots, c_j^{t(j), f(j,t(j))} c_j^{t(j), f(j,t(j))+1}
        \]
        that all belong to $\hat C_j$.
        Recall that the end-vertices of these edges are exactly the end-vertices of the $2t(j)$ edges with one end-vertex in $V(\hat C_j)$ and the other one in $V(P_1) \cup \dots \cup V(P_n)$.
        Furthermore, for every $z \in [t(j)]$, the vertices $c_j^{z, f(z,1)}$ and $c_j^{z, f(z,1)+1}$ have a common neighbor $s(j,z)$ in $V(P_1) \cup \dots \cup V(P_n)$.
        
        First, consider the case that all of the edges in $M$ are crossing.
        Then for every $z \in [t(j)]$, one of the vertices $c_j^{z, f(j,z)}$ and $c_j^{z, f(j,z)+1}$ belongs to $V_0$ and the other belongs to $V_1$ -- thus, exactly one of the edges $c_j^{z, f(j,z)} s(j,z)$ and $c_j^{z, f(j,z)+1} s(j,z)$ is crossing (no matter to which of the two sets $V_0$ and $V_1$ the vertex $s(j,z)$ belongs to).
        Therefore, the total weight of the crossing edges in $E^j_3$ is equal to $12t(j)^2 + t(j) < B^j_3$.

        Finally, we remain with the case where there exists a unique index $z^*(j) \in [t(j)]$ such that the edge $c_j^{z^*(j), f(j,z^*(j))}, c_j^{z^*(j), f(j,z^*(j))+1}$ is not crossing.
        First, for every $z \neq z^*(j) \in [t(j)]$, one of the vertices $c_j^{z, f(j,z)}$ and $c_j^{z, f(j,z)+1}$ then belongs to $V_0$ and the other belongs to $V_1$---thus, exactly one of the edges $c_j^{z, f(j,z)} s(j,z)$ and $c_j^{z, f(j,z)+1} s(j,z)$ is crossing.
        And further, either both edges $c_j^{z^*(j), f(j,z^*(j))} s(j,z^*(j))$ and $c_j^{z^*(j), f(j,z^*(j))+1} s(j,z^*(j))$ are crossing or none of them.
        First, this implies that the total weight of crossing edges in $E^j_3$ is at most
        \[
              12t(j)^2 + (t(j) - 1) + 2 = 12t(j)^2 + t(j) + 1 = B^j_3.
        \]
        And second, the upper bound $B_j^3$ can only be achieved if both edges $c_j^{z^*(j), f(j,z^*(j))} s(j,z^*(j))$ and $c_j^{z^*(j), f(j,z^*(j))+1} s(j,z^*(j))$ are crossing.
        This concludes the proof.
    \end{proof}

    With these claims in hand we are ready to prove that the instance $I$ is satisfiable if and only if the weighted graph $(G, \omega)$ admits a partition such that the total weight of the edges crossing this partition is at least $B$.

    \begin{claim}
        If $I$ is satisfiable, the weighted graph $(G, \omega)$ admits a partition such that the total weight of the edges crossing this partition is at least $B$.
    \end{claim}
    \begin{proof}
        Let $\phi \colon \{v_1, \dots, v_n\} \to \{0, 1\}$ be a satisfying assignment of $I$.
        We will now construct a partition $(V_0, V_1)$ of the vertex set of $G$ such that the total weight of the crossing edges is at least $B$.

        First, we put the vertices $q^{1,0}, q^{2,0}, \dots, q^{m,0}$ into $V_0$ and the vertices $q^{1,1}, q^{2,1}, \dots, q^{m,1}$ into $V_1$.
        Further, we put all vertices $r_1, r_2, \dots, r_m$ into $V_0$.
        Note that now all edges in $E_1$ are crossing, i.e., the total weight of the crossing edges in $E_1$ is equal to $(2m-1) + m = B_1$.

        Next, for every $i \in [n]$ we proceed as follows.
        We put the vertices $p_i^{1, \phi(i)}, p_i^{2, \phi(i)}, \dots, p_i^{m, \phi(i)}$ into $V_1$ and the vertices $p_i^{1, 1-\phi(i)}, p_i^{2, 1-\phi(i)}, \dots, p_i^{m, 1-\phi(i)}$ into $V_0$.
        Note that no all edges of the path $P_i$ are crossing.
        This implies that the total weight of crossing edges in $E_2$ is equal to $n \cdot (2m-1) = B_2$.

        Finally, for every $j \in [m]$, we partition the vertices of the cycle $\hat C_j$ as follows.
        The assignment $\phi$ satisfies the instance $i$ so there exists an index $z = z(j) \in [t(j)]$ such that we have $\phi(b^j_z) = 1$ if and only if $b^j_z$ occurs positively in $C_j$---note that there might be multiple $z \in [t(j)]$ with this property, we fix an arbitrary one.
        First, for all $1 \leq z' < z$ we put:
        \begin{itemize}
            \item the vertex $c_j^{z',1}$ into $V_1$,
            \item the vertex $c_j^{z',2}$ into $V_0$,
            \item the vertex $c_j^{z',3}$ into $V_1$,
            \item the vertex $c_j^{z',4}$ into $V_0$,
        \end{itemize}
        i.e., starting with $c_j^{1,1}$ we put vertices alternatingly into $V_1$ and $V_0$ until we reach $c_j^{z,1}$.
        Second, if $b^j_z = v_{i(j,z)}$ occurs positively in $C_j$, we put 
        \begin{itemize}
            \item $c_j^{z,1}$ into $V_1$,
            \item $c_j^{z,2}$ into $V_1$,
            \item $c_j^{z,3}$ into $V_0$,
            \item $c_j^{z,4}$ into $V_1$.
        \end{itemize}
        Otherwise, i.e., if $b^j_z = v_{i(j,z)}$ occurs negatively in $C_j$, we put 
        \begin{itemize}
            \item $c_j^{z,1}$ into $V_1$,
            \item $c_j^{z,2}$ into $V_0$,
            \item $c_j^{z,3}$ into $V_0$,
            \item $c_j^{z,4}$ into $V_1$.
        \end{itemize}
        Finally, for every $z < z' \leq t(j)$ we put 
        \begin{itemize}
            \item the vertex $c_j^{z',1}$ into $V_0$,
            \item the vertex $c_j^{z',2}$ into $V_1$,
            \item the vertex $c_j^{z',3}$ into $V_0$,
            \item the vertex $c_j^{z',4}$ into $V_1$,
        \end{itemize}
        i.e., starting with $c_j^{z+1, 1}$ we put the vertices alternatingly in $V_0$ and $V_1$ until we reach $r_j$.
        
        Recall that we already assigned the vertex $r_j$ to $V_0$.
        By construction, there exists exactly one edge on $\hat C_j$ that is not crossing, namely the edge $c_j^{z, f(j,z)} c_j^{z, f(j,z) + 1}$.
        So the total weight of the crossing edges on the cycle $\hat C_j$ is equal to $4t(j) \cdot 3t(j) = 12t(j)^2$.
        Also, for every $z' \neq z \in [t(j)]$, exactly one of the vertices $c_j^{z', f(j,z')}$ and $c_j^{z', f(j,z') + 1}$ belongs to $V_0$ so exactly one of the edges $c_j^{z', f(j,z')} s(j,z')$ and $c_j^{z', f(j,z')+1} s(j,z')$ is crossing.
        Furthermore, the following holds.

        If $v_{i(j,z)}$ occurs positively in $C_j$, then by the choice of the value $z$, we have $\phi(v_{i(j,z)}) = 1$ and therefore (by the construction of $V_0 \cap V(P_i)$ and $V_1 \cap V(P_i)$) we have $p_i^{j,0} \in V_0$.
        Recall that by construction above, we have $c_j^{z,1} \in V_1$ and $c_j^{z,2} \in V_1$, i.e., we have $c_j^{z,f(j,z)} \in V_1$ and $c_j^{z,f(j,z)+1} \in V_1$.
        Therefore, the two edges $c_j^{z,f(j,z)} p_i^{j,0}$ and $c_j^{z,f(j,z)+1} p_i^{j,0}$ are crossing.

        If $v_{i(j,z)}$ occurs negatively in $C_j$, then by the choice of the value $z$, we have $\phi(v_{i(j,z)}) = 0$ and therefore (by the construction of $V_0 \cap V(P_i)$ and $V_1 \cap V(P_i)$) we have $p_i^{j,0} \in V_1$.
        Recall that by construction above, we have $c_j^{z,2} \in V_0$ and $c_j^{z,3} \in V_0$, i.e., we have $c_j^{z,f(j,z)} \in V_0$ and $c_j^{z,f(j,z)+1} \in V_0$.
        Therefore, the two edges $c_j^{z,f(j,z)} p_i^{j,0}$ and $c_j^{z,f(j,z)+1} p_i^{j,0}$ are crossing.
        
        Altogether, the total weight of the crossing edges in $E_3^j$ is at least:
        \[
            12t(j)^2 + (t(j) - 1) + 2 = 12t(j)^2 + t(j) + 1 = B_j^3
        \]
        And therefore, the total weight of the crossing edges is at least
        \[
            B_1 + B_2 + \sum_{j \in [m]} B_3^j = B
        \]
        as desired.
    \end{proof}

    \begin{claim}
        If the weighted graph $(G, \omega)$ admits a partition such that the total weight of the edges crossing this partition is at least $B$, then $I$ is satisfiable.
    \end{claim}
    \begin{proof}
        Let $(V_0, V_1)$ be a partition of $V(G)$ such that the total weight of the edges crossing it is at least $B$.
        Without loss of generality we may assume that $q^{1,0} \in V_0$ holds.
        Recall that the sets $E_1, E_2, E_3^1, E_3^2, \dots, E_3^m$ partition the edge set of $G$.
        By \cref{claim:e-1-edges}, \cref{claim:e-2-edges}, and \cref{claim:e-3-edges} the total weight of crossing edges in $E_1, E_2, E_3^1, \dots, E_3^m$ is at upper-bounded by $B_1, B_2, B_3^1, \dots, B_3^m$, respectively.
        At the same time, $E_1, E_2, E_3^1, \dots, E_3^m$ partition the edge set of $G$ and we have $B_1 + B_2 + B_3^1 + \dots + B_3^m = B$.
        Hence, the total weight of crossing edges in $E_1, E_2, E_3^1, \dots, E_3^m$ is at exactly $B_1, B_2, B_3^1, \dots, B_3^m$, respectively.
        Then we know the following:
        \begin{itemize}
            \item 
            We have $q^{1,0}, q^{2,0}, \dots, q^{m,0} \in V_0$ and $q^{1,1}, q^{2,1}, \dots, q^{m,1} \in V_1$.
            Furthermore, for every $j \in [m]$, we have $r_j \in V_0$ (see \cref{claim:e-1-edges}).

            \item For every $i \in [n]$, there exists an index $y(i) \in \{0,1\}$ with $V_0 \cap V(P_i) = \{p_i^{\ell, y(i)} \mid \ell \in [m]\}$ (see \cref{claim:e-2-edges}).

            \item For every $j \in [m]$, there exists an index $z^*(j) \in [t(j)]$ such that the following holds.
            The edge $c_j^{z^*(j), f(j, z^*(j))} c_j^{z^*(j), f(j, z^*(j))+1}$ is the unique edge of $\hat{C}_j$ not crossing the partition. 
            Furthermore, both edges $c_j^{z^*(j), f(j,z^*(j))} s(j,z^*(j))$ and $c_j^{z^*(j), f(j,z^*(j))+1} s(j,z^*(j))$ are crossing (see \cref{claim:e-3-edges}).
        \end{itemize}
        
        Now we define the truth-value assignment $\phi \colon \{v_1, \dots, v_n\} \to \{0, 1\}$ via $\phi(v_i) = 1 - y(i)$ for every $i \in [n]$.
        In the remainder of this proof we will show that $\phi$ is a satisfying assignment for $I$.
        For this let $j \in [m]$ be arbitrary.
        Recall that the edges $c_j^{z^*(j), f(j,z^*(j))} s(j,z^*(j))$ and $c_j^{z^*(j), f(j,z^*(j))+1} s(j,z^*(j))$ are crossing.
        Let $i = i(j) \in [n]$ be such that $s(j,z^*(j)) = p_i^{j,0}$---by construction, the variable $u_i$ occurs in $C_j$.        
        Recall that on the path $\hat C_j - c_j^{z^*(j), f(j, z^*(j))} c_j^{z^*(j), f(j, z^*(j))+1}$ every edge is crossing, i.e., vertices alternate between $V_0$ and $V_1$ along this path, and we have $r_j \in V_0$.
        Now we make a case distinction.

        First, suppose that $v_i$ occurs positively in $C_j$.
        By construction, we then have $f(j,z^*(j)) = 1$.
        Thus, the distance between the vertices $r_j$ and $c_j^{z^*(j), 1}$ on the aforementioned path is odd and therefore, we have $c_j^{z^*(j), 1} \in V_1$.
        Since the edge $c_j^{z^*(j), 1} p_i^{j,0}$ is crossing, we then have $p_i^{j,0} \in V_0$ and hence, $\phi(v_i) = 1$ by construction.
        So $\phi$ satisfies the clause $C_j$.

        Now suppose that $v_i$ occurs negatively in $C_j$.
        By construction, we then have $f(j,z^*(j)) = 2$.
        Thus, the distance between the vertices $r_j$ and $c_j^{z^*(j), 2}$ on the aforementioned path is even and therefore, we have $c_j^{z^*(j), 2} \in V_0$.
        Since the edge $c_j^{z^*(j), 2} p_i^{j,0}$ is crossing, we then have $p_i^{j,0} \in V_1$ and hence, $\phi(v_i) = 0$ by construction.
        So $\phi$ satisfies the clause $C_j$.

        Altogether, we obtain that $\phi$ is indeed a satisfying assignment of $I$. 
    \end{proof}

    Now we show how to get rid of the weight function $\omega$ and obtain an equivalent insance $G'$ of (unweighted) \textsc{Max Cut}.
    To obtain it, we employ the following transformation by Lokshtanov et al.~\cite{DBLP:journals/talg/LokshtanovMS18}.
    We define the unweighted graph $G'$ which is obtained from $G$ by replacing every edge $e = \{u, v\}$ with $\omega(e)$ paths of length $3$ (i.e., 3 edges long) between $u$ and $v$, the internal vertices of these paths are then fresh vertices of degree $2$ in $G'$.

    \begin{lemma}[\cite{DBLP:journals/talg/LokshtanovMS18}]
        Let $W = \sum_{e \in E(G)} \omega(e)$.
        Then the graph $(G, \omega)$ admits a partition in which the total weight of the crossing edges is at least $B$ if and only if the graph $G'$ admits a partition in which at least $B' = 2W + B$ edges are crossing.
    \end{lemma}

    Finally, we show that the cutwidth of $G'$ is at most $n+\mathcal{O}_d(1)$.

    \begin{claim}
        The cutwidth of the graph $G'$ is at most $n+\mathcal{O}_d(1)$.
        Moreover, the graph $G'$ and a linear arrangement $\ell'$ of $G'$ of cutwidth at most $n+\mathcal{O}_d(1)$ can be computed from $I$ in polynomial time.
    \end{claim}

    \begin{proof}
        To provide the desired linear arrangement, we describe the ordering in which the vertices of the weighted graph $G$ are placed, then explain how to put the subdivision-vertices that arose in the construction of $G'$, and finally, provide an upper bound on the cutwidth of this linear arrangement.

        As usually done in similar lower-bound constructions, we arrange vertices ``columnwise''.
        We initialize $\ell$ to be empty and in the following, we always add vertices to the right of the last vertex placed so far.
        We iterate through $j = 1, \dots, m$ and place the vertices of the $j$. th ``column'', i.e., we add the vertices 
        \[
            q^{j, 0}, q^{j, 1}, p_1^{j, 0}, p_1^{j,1}, p_2^{j, 0}, p_2^{j,1}, \dots, p_n^{j, 0}, p_n^{j,1}
        \]
        in this ordering and after that, we add the vertices of the odd cycle $\hat C_j$ in the ordering they appear on this cycle starting with $c_j^{1,1}$.
        Observe that after we finish this procedure, every vertex of $G$ has been placed exactly once and the vertices of $V(G') \setminus V(G)$ are still missing.
        Let $\ell$ denote this linear arrangement of $G$.

        Now for every edge $uv$ of $G$ let $a_1, b_1, \dots, a_{\omega(uv)}, b_{\omega(uv)}$ be the vertices of $G'$ such that the paths $u a_1 b_1 v, \dots, u a_{\omega(uv)} b_{\omega(uv)} v$ in $G'$ arose by the subdivision of the edge $uv$.
        Without loss of generality, let $u$ occur before $v$ in the constructed linear arrangement---otherwise, we swap the roles of $u$ and $v$.
        To obtain $\ell'$, we insert into $\ell$, for every $uv \in E(G)$ and every $i \in [\omega(uv)]$, the vertices $a_i$ and $b_i$ into arbitrary positions between $u$ and $v$ in such a way that $a_i$ occurs before $b_i$.
        This concludes the construction of the linear arrangement $\ell$ of $G'$.
        
        It remains to bound the cutwidth of $\ell$.
        We say that two edges are \emph{non-overlapping} in some linear arrangement if no cut of this arrangement crosses both of these edges.
        For the sake of readability we first, state some properties of $\ell$ and then conclude with an upper bound on the cutwidth of $\ell'$:
        \begin{itemize}
            \item Any two different edges of $Q$ are non-overlapping in $\ell$.
            \item For every $i \in [n]$, any two different edges of $P_i$ are non-overlapping in $\ell$.
            \item For every $j \neq j' \in [m]$, any edge incident with a vertex of $\hat C_j$ and any edge incident with a vertex of $\hat C_{j'}$ are non-overlapping.
        \end{itemize}
        Now, crucially, observe that by construction of $\ell'$ the following holds: if two edges $e$ and $e'$ of $G$ are non-overlapping in $\ell$, then any edge subdividing $e$ is non-overlapping with any edge subdividing $e'$ in $\ell'$.
        
        So now consider an arbitrary cut in $\ell'$.
        The above properties imply the following:
        \begin{itemize}
            \item Edges of $Q$ have unit weight in $\omega$ so at most one edge subdividing an edge of $Q$ crosses the current cut.
            \item The edges of $P_1, \dots, P_n$ are also of unit weight.
            Thus, for every $i \in [n]$, at most one edge subdividing an edge of $P_i$ crosses this cut. Hence, at most $n$ edges subdividing the edges of $P_1, \dots, P_n$ cross the current cut.
            \item There exists at most one index $j \in [m]$ such that the edges of $G'$ subdividing the edges incident with $\hat C_j$ in $G$ cross this cut.
            Now recall that the total weight of the edges incident with $\hat C_j$ in $G$ is upper-bounded by 
            \[
                1 + 2 \cdot d + (4d + 1) \cdot 3 d \in \bigoh_d(1): 
            \]
            the first addend is for the unique edge with the other end-vertex in $V(Q)$, the second addend is for the edges with the other end-vertex in $V(P_1) \cup \dots \cup V(P_n)$, and the last one is for the edges of $\hat C_j$.
            Hence, the current cut is crossed by $\bigoh_d(1)$ edges subdividing the edges incident with $\hat C_1, \dots, \hat C_m$.
        \end{itemize}
        Altogether, any cut in $\ell'$ is crossed by at most $n + \mathcal{O}_d(1)$ edges and therefore, the cutwidth of $\ell'$ is bounded by this value.
        \end{proof}

    Now we are ready to conclude the proof of the theorem.
    Suppose there exists a value $\varepsilon$ such that the \textsc{Max Cut} problem can be solved in time $\mathcal{O}^*((2-\varepsilon)^k)$ when the input graph is provided with a linear arrangement of cutwidth $k$.
    Then for any fixed $d \geq 2$ we could solve the \textsc{$d$-SAT} problem in time $\mathcal{O}^*((2-\varepsilon)^n)$ where $n$ denotes the number of variables in the input instance as follows.
    Let $I$ be an input instance.
    As described above, in polynomial time we can compute a graph $G'$, an integer $B'$, a linear arrangement of $G'$ of cutwidth at most $n+\mathcal{O}_d(1)$ such that $I$ is satisfiable if and only if $G'$ admits a partition with at least $B'$ crossing edges.
    Now by our assumption, we can decide if $G'$ admits such a partition in time $\mathcal{O}^*((2-\varepsilon)^{n+\mathcal{O}_d(1)}) = \mathcal{O}^*((2-\varepsilon)^{n})$ (recall that $d$ is a constant) and output the answer.
    This would give a $\mathcal{O}^*((2-\varepsilon)^n)$ algorithm for \textsc{$d$-SAT} for any $d \geq 2$ contradicting SETH.
\end{proof}
 
 \section{Conclusion and Future Work}\label{sec:conclusion}

Our results together with previous work~\cite{DBLP:conf/stacs/BojikianCHK23,DBLP:journals/tcs/JansenN19} show an interesting variety of behaviors for the complexity of problems relative to treewidth/pathwidth vs.\ relative to cutwidth: We may get the same tight bound, a small decrease in complexity, or a substantial decrease. We also discovered two more rare examples of tight exponential bounds with non-integral bases, especially \textsc{Triangle Packing}, which has an integral base relative to treewidth. In our opinion, this makes parameterization by cutwidth a very good test bed for deepening the understanding of dynamic programming and width parameters.

There are several avenues for future work: (i) Edge-disjoint packing problems may be intractable for treewidth/pathwidth but may have tight bounds for cutwidth. (ii) Going beyond classical problems, it would be interesting to determine the tight complexity of $(\sigma, \rho)$-domination problems for cutwidth which is known for treewidth~\cite{DBLP:conf/soda/FockeMINSSW23}. (iii) Closing the gap for \textsc{List-$H$-Homomorphism} left by Groenland et al.~\cite{DBLP:conf/icalp/GroenlandMNPR24} is a challenging open question.

\bibliography{ref}

\end{document}